\documentclass[reqno, 11pt]{amsart}
\usepackage[mathscr]{eucal}
\usepackage{amsfonts, amsmath, amsthm, amssymb, latexsym}
\newtheorem{theorem}{Theorem}[section]
\newtheorem{lemma}[theorem]{Lemma}
\newtheorem{prop}[theorem]{Proposition}
\newtheorem{cor}[theorem]{Corollary}
\theoremstyle{definition}
\newtheorem{defi}[theorem]{Definition}
\newtheorem{example}[theorem]{Example}
\theoremstyle{remark}
\newtheorem{remark}[theorem]{Remark}

\numberwithin{equation}{section}
\newcommand{\onto}{\,\,\twoheadrightarrow\,\,}
\newcommand{\la}{\label}
\newcommand{\GL}{\mathtt{GL}}
\newcommand{\Ker}{\mathtt{Ker}}
\newcommand{\Coker}{\mathtt{Coker}}
\newcommand{\Res}{\mathtt{Res}}
\newcommand{\Spec}{\mathtt{Spec}}

\newcommand{\ad}{\mathtt{ad}}
\newcommand{\Der}{\mathtt{Der}}
\newcommand{\End}{\mathrm{End}}

\newcommand{\U}{\mathtt{U}}

\newcommand{\into}{\hookrightarrow}

\def\c{\mathbb{C}}
\def\R{\mathbb{R}}
\def\e{\boldsymbol{e}}

\def\Z{\mathbb{Z}}

\def\qaw{Q_\A}%{Q_{\A}^{W}}
\def\qreg{Q_{\rm{reg}}}

\def\QQ{\mathbf{Q}}
\def\N{\mathbb{N}}
\def\O{\mathcal{O}}
\def\C{\mathcal{C}}

\def\g{\mathfrak{g}}

\def\grd{\mathtt{gr}}
\def\sl2{{\mathfrak{s}\mathfrak{l}}_2}
\def\vreg{V_{\rm{reg}}}
\def\VV{\widehat V}
\def\aalpha{\widehat{\alpha}}

\def\Ad{\mathrm{Ad}}

\def\A{\mathcal{A}}
\def\QQ{\mathcal Q}

\def\U{U} %{\mathcal{U}}
\def\M{\mathcal{M}}
\def\NN{\mathcal{N}}
\def\wM{\widetilde{\mathcal{M}}}

\def\D{\mathcal{D}}
\def\wD{\widetilde{\mathcal{D}}}
\def\wB{\widetilde B}

\begin{document}
%
%\begin{flushright}
%\textbf{Preliminary Version}
%\end{flushright}
%
%\vspace{1cm}
%
\title[Deformed Calogero--Moser operators and Cherednik algebras]{Deformed Calogero--Moser operators and ideals of rational Cherednik algebras}
\author{Yuri Berest}
\address{Department of Mathematics,
Cornell University, Ithaca, NY 14853-4201, USA}
\email{berest@math.cornell.edu}
%
%\thanks{Berest's work partially supported by NSF grant DMS 0901570.}
%
\author{Oleg Chalykh}
\address{School of Mathematics, University of Leeds, Leeds LS2 9JT, UK}
\email{o.chalykh@leeds.ac.uk}
%\date{\today}
%

%
\begin{abstract}
We consider a class of hyperplane arrangements $\mathcal A$ in ${\mathbb C}^n$ that generalise the locus configurations of \cite{CFV}. To such an arrangement we associate a second order partial differential operator of Calogero--Moser type, and prove that this operator is completely integrable (in the sense that its centraliser in $\mathcal{D}({\mathbb C}^n\setminus\mathcal A)$ contains a maximal commutative subalgebra of Krull dimension $n$). Our approach is based on the study of shift operators and associated ideals in the spherical Cherednik algebra that may be of independent interest. The examples include all known families of deformed (rational) Calogero--Moser systems that appeared in the literature; we also construct some new examples, including a BC-type analogues of completely integrable operators recently found by D. Gaiotto and M. Rap$\check{\rm c}$\'ak in \cite{GR}. 
We describe these examples in a general framework of rational Cherednik algebras close in spirit to \cite{BEG} and \cite{BC}.
\end{abstract}

\maketitle

\section{Introduction}
Let  $V = \R^n $ be an Euclidean space with standard inner product $(\cdot, \cdot)$. Consider a collection $ \A_+ =\{\alpha\} $ of nonparallel vectors in $V$ with prescribed `multiplicities' $k_\alpha$, which we assume (for the moment) to be arbitrary real numbers. We will refer to the pair $ (\A, k_\alpha) $, where  $ \A := \A_+ \cup (-\A_+) $ with $ k_{-\alpha} := k_\alpha $, as a {\it configuration} in $\R^n$. 
With such a configuration we associate a generalised \emph{Calogero--Moser operator} of the form
\begin{equation}
\label{gcm}
L_{\A}:=\Delta_n -\sum_{\alpha\in \A_+} \frac{k_\alpha(k_\alpha+1)(\alpha,\alpha)}{(\alpha,x)^2}\,,
\end{equation}
where $ \Delta_n $ is the Laplacian on $\mathbb R^n$.
The standard (rational) Calogero--Moser operator corresponds to the root system of type $A_{n-1}$ with all $k_\alpha=k$:
\begin{equation}
\label{cm}
L=\Delta_n -\sum_{i<j}^n \frac{2\,k(k+1)}{(x_i-x_j)^2}\,.
\end{equation}
The operator \eqref{cm} can be viewed as a quantum Hamiltonian of a system of $n$ interacting particles on the line. This is a celebrated example of a quantum completely integrable system: there exist $n$ algebraically independent partial differential operators $L_1, L_2,\dots, L_n$, including $L$, such that $[L_i, L_j]=0$ for all $i, j= 1,\dots,n$. In contrast, the quantum Hamiltonian \eqref{gcm} is not completely integrable for an arbitrary configuration. The natural question for which $ \A$'s exactly this Hamiltonian is  integrable has received a good deal of attention --- both in mathematical and physical literature --- but still remains open. 

The starting point for the present paper is the following observation, which is a simple consequence of the main result of \cite{T}.
\begin{theorem}\label{tci} Let $ L_\A $ be a completely integrable quantum Hamiltonian of the form \eqref{gcm}
such that its quantum integrals $L_1, \dots, L_n$ have algebraically independent constant principal symbols
$p_1, \dots, p_n \in \R[V^*]$.
Assume that $k_\alpha\notin \Z$ for all $\alpha\in\A$. Then the polynomials $p_i$ are invariant under a finite Coxeter group $W\subset\GL(V)$, and $ \A $ is a subset of the root system $R$ of $W$.
\end{theorem}
Indeed, if $s_\alpha$ is the orthogonal reflection corresponding to $\alpha\in\A$, then, as shown in \cite{T}, each $p_i$ must be invariant under $s_\alpha$. Now take $\alpha, \beta\in \A_+$, and assume that $s_\alpha s_\beta$ is of infinite order. Then $p_i$ must be invariant under an arbitrary rotation in  the two-dimensional plane spanned by $\alpha, \beta$. However, the ring of polynomials invariant under such rotations has Krull dimension $<n$, which implies that $p_1,\dots, p_n$ cannot be algebraically independent. By contradiction, we conclude that $s_\alpha s_\beta$ is of finite order for any $\alpha, \beta$, therefore the reflections $\{s_\alpha\}_{\alpha\in\A}$ generate a finite Coxeter group $W$, and so $\A $ is a subset of the root system $R$ of $W$. %\qed

\medskip

Theorem \ref{tci} tells us that for non-integral parameters $k_\alpha$, the completely integrable operators of the form \eqref{gcm} are closely related to Coxeter groups. Indeed, by a theorem of Heckman \cite{He}, the Calogero--Moser operator (introduced in \cite{OP}) 
\begin{equation}
\label{wcm}
L_{W}:=\Delta_n -\sum_{\alpha\in R_+} \frac{k_\alpha(k_\alpha+1)(\alpha,\alpha)}{(\alpha,x)^2}\,,
\end{equation}
is completely integrable for the root system $R$ of an arbitrary finite Coxeter group $W$ and an arbitrary $W$-invariant function $k:\, R \to \R$. (For all crystallographic groups $W$ this was already shown in \cite{HO4}.)

On the other hand, in the case when {\it all} the $k_\alpha$'s are integers, there are examples of completely integrable operators of the form \eqref{gcm} where $\A $ is not part of any root system (see \cite{CFV}). Instead, such configurations satisfy certain algebraic equations called the {\it locus relations}. The purpose of this paper is to study the general (`mixed') case: i.e., the completely integrable operators of the form \eqref{gcm} where some of the $k_\alpha$'s are integers and some are not. The first examples of such operators were constructed by A. Sergeev and A. Veselov in \cite{SV}; further examples were found by M. Feigin in \cite{F}.
Recently, D. Gaiotto and M. Rap$\check{\rm c}$\'ak \cite{GR} discovered the following family of operators in $ V=\R^{n_1}\times \R^{n_2}\times\R^{n_3} $ depending on (complex) parameters $ \epsilon_1, \epsilon_2, \epsilon_3$ satisfying 
$ \epsilon_1 + \epsilon_2 + \epsilon_3 = 0 $:
\begin{align}
L&=\epsilon_1\sum_{i=1}^{n_1}\partial_{z_i}^2+\frac{\epsilon_2\epsilon_3}{\epsilon_1}\sum_{i<j}\frac{2}{(z_i-z_j)^{2}}+\epsilon_1\sum_{i,j}\frac{2}{(z'_i-z''_j)^{2}}+ \nonumber\\
&+\epsilon_2\sum_{i=1}^{n_2}\partial_{z'_i}^2+\frac{\epsilon_1\epsilon_3}{\epsilon_2}\sum_{i<j}\frac{2}{(z'_i-z'_j)^{2}}+\epsilon_2\sum_{i,j}\frac{2}{(z_i-z''_j)^{2}}+ \la{GRex}\\
&+\epsilon_3\sum_{i=1}^{n_3}\partial_{z''_i}^2+\frac{\epsilon_1\epsilon_2}{\epsilon_3}\sum_{i<j}\frac{2}{(z''_i-z''_j)^{2}}+\epsilon_3\sum_{i,j}\frac{2}{(z_i-z'_j)^{2}}\,,\nonumber
\end{align}
where $z_i$, $z'_i$ and $z''_i$ are Cartesian coordinates in $\R^{n_1}$, $\R^{n_2}$ and $\R^{n_3}$, respectively\footnote{The operators \eqref{GRex} arise in the context of so-called $\Omega$-deformation of supersymmetric gauge theories (see, e.g., \cite{N, NW}). The parameters
$ \epsilon_1, \epsilon_2, \epsilon_3$ correspond to the `$\Omega$-deformed' $\c^3$ (denoted $ \c_{\epsilon_1} \times \c_{\epsilon_2} \times \c_{\epsilon_3} $ in \cite{GR}), with 
relation $ \epsilon_1 + \epsilon_2 + \epsilon_3 = 0 $ reflecting the Calabi-Yau condition (see {\it loc. cit.}, Sect. 1.4).}. 
We will show in Section~\ref{A123} that \eqref{GRex} also fits in our class of completely integrable operators. In addition to known examples, we will construct a number of new ones, including a BC-type generalisation of the Gaiotto-Rap$\check{\rm c}$\'ak 
family (see Section~\ref{BC123}).

A common feature of all these examples is that the vectors in $\A$ with non-integral multiplicities form a root system $R$ of a finite Coxeter group $W$, while those with integral multiplicities constitute a (finite) subset in $\R^n $ that is stable under the action of $W$. To ensure integrability, the vectors with integral multiplicities must satisfy certain compatibility conditions similar to the locus relations of \cite{CFV}. We call such $\A$'s the {\it generalised locus configurations}, or more precisely -- 
when $W$ is specified -- the {\it locus configurations of type $W$} (see Definition~\ref{defloc}). (In this
terminology, the original configurations considered in \cite{CFV} correspond to $ W = \{e\}$.)

Our main result -- Theorem \ref{gaic} -- states that for any generalised locus configuration, the operator \eqref{gcm} is completely integrable. In fact, we show that associated to a locus configuration of type 
$W$ there is a maximal commutative algebra of differential operators of rank $ |W| $, containing  
\eqref{gcm}. This algebra is isomorphic to the ring of {\it generalised quasi-invariant}  $Q_{\A} $ determined by $\A$ (see Definition \ref{gqi}). 
%The complete integrability immediately follows from that. 
%It is worth mentioning that this commutative ring of differential operators has, in general, rank 
%$ |W| > 1 $ (see Remark \ref{rk}).  
We also prove that there exists a linear differential operator $S$ such that
$$
L_{\A}\,S=S\,L_W 
$$
which we call (by tradition)  a {\it shift operator} from $ L_W $ to $ L_\A $.
In various special cases some results of this kind can be found in \cite{CFV1, SV, F, SV1}; our approach unifies them and applies to a broader class of operators. 

%We also extend some of our results to Calogero--Moser type operators with an additional quadratic oscillatory term. 
 
We now explain how generalised locus configurations  
are related to rational Cherednik algebras. Our approach is inspired by \cite{BEG} and \cite{BC}; 
however, the Cherednik algebras play a different role in our construction. 
If $X$ is an affine algebraic variety, we write $ \D(X) $ for the ring 
of (global) algebraic differential operators on $X$. It is well known that when
$X$ is singular, the ring $ \D(X) $ has a complicated structure. 
A natural way to approach $ \D(X) $ geometrically is to relate it to the ring of
differential operators on a non-singular variety $Y$, which is a resolution of $X$.
Specifically ({\it cf.} \cite{SS}), assuming that  the variety $X$ is  irreducible, 
one can choose a finite birational map $ \pi: Y \to X$  with $Y$ smooth and consider the space 
of differential operators from $ Y $ to $X\,$:
\begin{equation}
\label{1}
\D(Y,X) := \{D \in \D(\mathbb{K}) \,:\, 
D[\O(Y)] \subseteq \O(X) \}\, ,
\end{equation}
where $ {\mathbb K} $ is the field of rational functions on $ X $.
This space is naturally a right module over $ \D(Y)$ and a left module
over $ \D(X) $, and the two module structures are compatible: in other words, 
$ \D(Y,X) $ is an $\D(X)$-$\D(Y)$-bimodule. Taking the
endomorphism ring of $ \D(Y,X) $ over $\D(Y)$ and mapping the differential operators
in $ \D(X) $ to (left) multiplication operators on $ \D(Y,X) $
gives an algebra homomorphism: $ \D(X) \to \End_{\D(Y)} \D(Y,X)$, which ---
under good circumstances --- turns out to be an isomorphism. 
In \cite{BEG}, this construction was used for the varieties of classical
quasi-invariants, $ X = \Spec\, Q_m $, in which case the resolution $ \pi: Y \to 
X $ is given by the normalization map, with $ Y = \tilde{X} \cong  V $.

In this paper, we generalise (`deform') the above construction replacing the
ring  $ \D(Y) $ of differential operators on a smooth resolution of $X$
by a (spherical) Cherednik algebra. To be precise,
given a locus configuration $\A$ of type $W$, we consider
the variety of generalised quasi-invairants, $ X := \Spec\,\qaw $, together with a natural
map $ \pi: V/\!/W \to X $ corresponding to the inclusion
$ \qaw \subset \c[V]^W $ (see Definition 4.2). 
Instead of applying \eqref{1} directly to $ \pi $, we first restrict this map to the subspace
$ \vreg /\!/W $ of regular $W$-orbits in $ V/\!/W  $ obtained by removing from $V$
the reflection hyperplanes of $W$. We then define the ring 
$ Q_{\rm reg} \subseteq \c[\vreg]^W $, using the same 
algebraic conditions as for $ Q = \qaw $ (see (4.1)) but with $ \c[V]^W $
 replaced by $ \c[\vreg]^W $.  Now, taking $ X_{\rm reg} := \Spec \,Q_{\rm reg}$, we 
consider the bimodule $ \D(\vreg /\!/W, X_{\rm reg}) $ associated to
the natural map $ \pi_{\rm reg}: \vreg /\!/W \to X_{\rm reg} $. 
Since $W$ acts freely on $\vreg $, we have 
$ \D(\vreg/\!/W) \cong \D(\vreg)^W $, and therefore   
$\, \D(\vreg/\!/W, X_{\rm reg})  \subseteq \D(\vreg)^W[\delta_k^{-1}] \,$,
where $ \delta_k := \prod_{\alpha \in \A_{+}\setminus R} (\alpha,x)^{k_\alpha} $.
The spherical subalgebra $ B_k  $ of the rational Cherednik algebra $ H_k(W) $ 
with $ k = \{k_{\alpha}\}_{\alpha \in R} $ embeds naturally into $ \D(\vreg)^W $ 
via the Dunkl representation (see (2.5)); thus, we can define
$$
\M_{\A, W} := \D(\vreg/\!/W, X_{\rm reg})  \cap B_k\,.
$$
This  is a right $B_k$-module -- in fact, a right ideal of $B_k$ -- that we associate\footnote{For technical reasons, it will be more convenient for us to work with a twisted (fractional) ideal which is obtained by replacing $Q_{\rm reg} = \O(X_{\rm reg}) $ in the above construction by a rank one torsion-free $ \O(X_{\rm reg}) $-module $ \U_{\A} $ (see Section~\ref{S5.1}).} to our generalised locus configuration in place of \eqref{1}.

%\subsection*{Outline of the paper and main results}
For the reader's convenience, we now outline the contents of the paper and briefly summarise our main results. Section~\ref{S2} comprises background material: here, we recall a (well-known) relation of quantum Calogero--Moser systems to rational Cherednik algebras \cite{EG} and review  some results on locus configurations from \cite{CFV}. 

In Section~\ref{S3}, we give our main definition (Definition~\ref{defloc}) and state our main result: Theorem~\ref{gaic}. The proof of Theorem~\ref{gaic} appears in Section~\ref{S5}; however, we do not prove this theorem directly but deduce it from a (much more) general algebraic result --- Theorem~\ref{IT} --- that provides necessary and sufficient conditions for the existence of differential shift operators. 

Theorem~\ref{IT} is proven in Section~\ref{pro} and should be considered as the second main result of this paper. 
Our approach originates from an attempt to understand examples and unify various {\it ad hoc} constructions of shift operators known in higher dimension $(n>1)$. Some of the ideas go back to old observations of the authors in \cite{B98} and \cite{C98}; our main innovation is in clarifying the role of (Ore) localisation and its relation to the ad-nilpotency condition (see Lemma~\ref{shiftl}) as well as the use of a canonical ad-nilpotent filtration 
(Lemma~\ref{filtl}). This allows us to state Theorem~\ref{IT} in an abstract `coordinate-free' form and prove it
under very general assumptions. Other notable results in Section~\ref{pro} are Proposition~\ref{fatrefl} and Proposition~\ref{commprop} that establish the reflexivity property of differential ideals related to shift operators and the existence of `large' commutative subalgebras of differential operators, respectively.

Section~\ref{exloc} describes examples of generalised locus configurations --- in fact, all currently known examples ---
in dimension $n >2$. As mentioned above, our collection contains a new interesting family: a BC-type generalization of the Gaiotto-Rap$\check{\rm c}$\'ak operators \eqref{GRex} (see Section~\ref{BC123}). 

In Section~\ref{twodim}, we attempt to describe {\it all} two-dimensional locus configurations:
we give a general construction of such configurations and explicitly describe a new large class of examples of locus configurations of type $W$, where $ W = I_{2N} $ is a dihedral group. 
%This family can be viewed as 2D analog of the `even' family of 
%one-dimensional Schr\"odinger operators studied in \cite{DG}. 

In Section~\ref{S8}, we study deformed Calogero--Moser operators with harmonic oscillator terms. The main result of this section --- Theorem~\ref{gaicw} --- provides shift operators and quantum integrals for 
such Calogero-Moser operators; it can be viewed as a (partial) generalisation of Theorem~\ref{gaic}. We also draw reader's attention to Proposition \ref{suprop} and Example~\ref{Ex8.4} that link our results to
recent work on quantum superintegrable systems (see, e.g., \cite{MPR}).

Finally, in Section~\ref{S9}, we extend our results to Calogero-Moser operators associated with 
affine (noncentral) configurations. In dimension one, there are two famous examples: the Schr\"odinger operators with Adler-Moser potentials (a.k.a. the rational solutions of the KDV hierarchy \cite{AM}) 
and the so-called ``even'' family of bispectral operators discovered by Duistermaat and Gr\"unbaum in \cite{DG}. These examples correspond to affine locus configurations in $\c^1$ of types $ W = \{e\} $ and $ W = \Z_2 $, 
respectively. The main result of this section --- Theorem~\ref{gaicaff} --- can thus be viewed
as a natural multi-dimensional generalisation of the classical results of \cite{AM} and \cite{DG}.
%We also find one new curious example in dimension two: an affine locus configuration of type 
%$A_2$ that can be viewed as a deformation of root system $ G_2 $ (see \eqref{newaff}).

\subsection*{Acknowledgments} {\footnotesize We are grateful to P. Etingof, M. Feigin, A. Sergeev and A. Veselov for many questions and stimulating discussions. We also want to thank Pavel Etingof for drawing our attention to new examples of deformed Calogero-Moser operators that appeared in \cite{GR} 
and Davide Gaiotto for interesting correspondence clarifying to us the origin of these examples.
We are especially grateful to Misha Feigin who read the first version of this paper and pointed out several inaccuracies and misprints. 
The work of the first author is partially supported by NSF grant DMS 1702372 and the 2019 Simons Fellowship. The work of the second author was partially supported by EPSRC under grant EP/K004999/1.  }

\section{Cherednik algebras and Calogero--Moser systems}
\la{S2}
In this section we recall a well-known relation between rational Cherednik algebras and Calogero--Moser systems. For more details and references, we refer the reader to \cite{EG}. 

Let $W$ be a finite Coxeter group with reflection representation $V$. Throughout the paper we will work over $\c$, so $V$ is a complex vector space with a $W$-invariant bilinear form $(\cdot, \cdot)$. Each reflection $s\in W$ acts on $V$ by the formula 
\begin{equation}
\la{refls}
s(x) = x-2\,\frac{(\alpha,x)}{(\alpha,\alpha)}\,\alpha\ ,
\end{equation}
where $\alpha\in V$ is a normal vector to the reflection hyperplane. Denote by $R_+$ the set of all these normals and put $R=R_+\cup\,-R_+$. Only the direction of each normal $\alpha$ is important, so we may assume that they are chosen in such a way that the set $R$ is $W$-invariant (it is also customary to choose $R_+$ to be contained in some prescribed half-space). Let us choose a $W$-invariant function $k\,:\, R\to\c$. The elements $\alpha\in R$ are called the roots of $W$, and $k_\alpha:=k(\alpha)$ is called the \emph{multiplicity} of $\alpha$. Note that we do not assume that $W$ is irreducible, and $R$ may not span the whole $V$.

We set $\, \vreg := \{\,x\in V\, | \, (\alpha,x)\ne 0\ \forall\alpha\in R\}$ and denote by $ \c[\vreg] $ and $ \D(\vreg) $
the rings of regular functions and regular differential operators on $ \vreg $, respectively. The
action of $ W $ on $ V $ restricts to $ \vreg $, so $ W $
acts naturally on $ \c[\vreg] $ and $ \D(\vreg) $ by algebra automorphisms. We form the crossed products
$ \c[\vreg]*W $ and $ \D W :=  \D(\vreg)*W$. As an algebra,
$\D W$ is generated by its two subalgebras, $ \c W $ and $ \D(\vreg) $. %, and hence, by the elements of $ W $, $\, \c[\vreg] \,$, and the derivations $\, \partial_\xi \,$, $\, \xi \in V $.

The Calogero--Moser operator associated to $W$ and $k=\{k_\alpha\}$ is a differential operator $L_{W}\in \D(\vreg)^W$ defined by
\begin{equation}\label{cmo}
L_{W}:=\Delta-u_{W}\,,\qquad u_{W}=\sum_{\alpha\in R_+} \frac{k_\alpha(k_\alpha+1)(\alpha,\alpha)}{(\alpha,x)^2}\,,
\end{equation}
where $\Delta$ is the Laplacian on $V$ associated with the $W$-invariant form $(\cdot, \cdot)$.

To describe the link between $L_{W}$ and Cherednik algebra, we first define the {\it Dunkl operators}
$\,T_{\xi} \in \D W \,$ as
\begin{equation}
\label{du} T_\xi :=
\partial_\xi+\sum_{\alpha\in R_+}
\frac{(\alpha,\xi)}{(\alpha, x)}k_\alpha s_\alpha\ , \quad \xi \in V\ .
\end{equation}
Note that the operators \eqref{du} depend on $\, k = \{k_{\alpha}\} \,$, and
we sometimes write $ T_{\xi, k} $ to emphasize this dependence.
The basic properties of Dunkl operators are listed in the following lemma.
\begin{lemma}[\cite{D}]\la{duprop}
For all $\,\xi, \eta \in V\,$ and $ w \in W $, we have

$(1)$\ {\rm commutativity:}\ $\, T_{\xi}\,T_{\eta} -
T_{\eta}\,T_{\xi} = 0 \,$,

$(2)$\ {\rm $W$-equivariance:}\ $\,w\,T_\xi =
T_{w(\xi)}\,w\,$,

$(3)$\ {\rm homogeneity:}\ $ \,T_\xi \,$ is an operator  of degree $ -1 $
 with respect to the natural homogeneous grading on $ \D W $.
\end{lemma}

In view of Lemma~\ref{duprop}, the assignment $\,\xi \mapsto T_\xi\,$
extends to an (injective) algebra homomorphism
\begin{equation}
\la{hom}
\c[V^*] \hookrightarrow \D W \ ,\quad q \mapsto T_q \ .
\end{equation}
Identifying $ \c[V^*] $ with its image in $ \D W $ under
\eqref{hom}, we now define the {\it rational Cherednik algebra}
$H_k=H_k(W)$ as the subalgebra of $\D W$ generated by $ \c[V] $,
$\,\c[V^*] $ and $ \c W$.
The family $ \{H_k\} $ can be viewed as a deformation (in fact, 
universal deformation) of the crossed product $\, H_0 = \D(V)*W \,$
(see \cite{EG}, Theorem~2.16). The above realization of $\, H_k$ inside $\D W$
is referred to as the {\it Dunkl representation}
of $ H_k $.

The algebra $ \D W = \D(\vreg) * W $ carries a natural {\it differential} filtration,
defined by taking $\,\deg(x) =0 $,
$\,\deg(\xi) = 1 $ and $\deg(w) = 0$ for all $x\in V^*$, $\xi\in V$
and $w\in W$. Through the Dunkl
representation, this induces a filtration on $ H_k $
for all $ k $, and the associated graded ring
$\,\grd\, H_k\,$ is isomorphic to $\,\c[V\times V^*]*W$; in particular, it is independent of $ k $.
This implies the PBW property for $ H_k $, i.e. %the existence of 
a vector space isomorphism
\begin{equation}\la{pbw}
H_k\stackrel{\sim}{\to} \c [V] \otimes \c W \otimes \c[V^*]\,.
\end{equation}

By definition, the {\it spherical subalgebra} of $ H_k
$ is given by $\e\, H_k \,\e \,$, where $\, \e =
|W|^{-1} \sum_{w \in W } w \,$.
For $ k = 0 $, we have $H_0 = \D(V)
* W$ and $\e H_0\e \cong \D(V)^W \,$; thus, the family $ \e H_k \e$ is a
deformation (in fact, universal deformation) of the ring of
invariant differential operators on $ V $.

The Dunkl representation restricts to the embedding $\,\e H_k \e \into \e \D W \e\,$. 
If we combine this
with (the inverse of) the isomorphism $\, \D(\vreg)^W
\stackrel{\sim}{\to} \e \,\D W \e \,$, $\, u \mapsto \e u \e = \e
u = u \e \,$, we get an algebra map (cf. \cite{He})
\begin{equation}
\la{HC}
\Res:\,\e H_k \e \into \D(\vreg)^W \ ,
\end{equation}
representing the spherical subalgebra $\e H_k \e$ by invariant differential operators. 
We will refer to \eqref{HC} as the {\it spherical Dunkl
representation} and denote 
\begin{equation}\la{HC1}
B_k:=\Res (\e H_k\e)\subset \D(\vreg)^W\,.
\end{equation}
%
\iffalse
The differential filtration on $H_k$ induces filtrations on $\e H_k \e$ and $B_k$, with
\begin{equation*}
\grd\, B_k = \c[V\times V^*]^W\,.
\end{equation*}
\fi
%
\begin{theorem}[see \cite{He}]
\label{ci} 
Let $\xi_1\,\dots, \xi_n$ be an orthonormal basis of $V$, and 
$q=\xi_1^2+\dots+\xi_n^2\in \c[V^*]^W$. 
Then $\Res(\e\, T_q \,\e)=L_{W}$ is the Calogero--Moser operator \eqref{cmo}. 
Furthermore, the  image of $\e\,\c[V^*]^W\e$ under the spherical Dunkl representation \eqref{HC} 
forms a commutative subalgebra in $\D(\vreg)^W$, and the operator $L_{W}$, thus, defines a quantum 
completely integrable system.
\end{theorem}

Theorem \ref{tci} stated in the Introduction implies that if $k_\alpha\notin \Z$ for all $\alpha$, then the commutative algebra constructed 
in Theorem \ref{ci} is maximal (i.e. coincides with its centralizer) in $\D(\vreg)^W$. On the other hand, when $k_\alpha$'s are integers, this algebra can be extended to a larger commutative algebra. This stronger property is known as {\it algebraic integrability} \cite{CV, VSC}. 
To state the result, let us make the following definition, cf. \cite{CV, VSC, FV}.

\begin{defi} Let $\{\A, k\}$ be a configuration with $k_\alpha\in\Z_+$ for all $\alpha\in\A$. A polynomial $q\in\c[V]$ is called \emph{quasi-invariant} if
\begin{equation}\label{qi}
q(x)-q(s_\alpha x)\ \text{is divisible by}\ (\alpha,x)^{2k_\alpha}\quad \forall\alpha\in\A_+\,.
\end{equation}
The set of all quasi-invariant polynomials in $\c[V]$ is denoted by $Q_{\A}$. It is easy to check that $Q_{\A}$ is a subalgebra in $\c[V]$.
\end{defi}
In the case when $\A=R$ is a root system of a Coxeter group $W$, we have $\c[V]^W\subset Q_{\A} \subset \c[V]$, so the algebra of quasi-invariants $Q_k(W):=Q_{R}$ interpolates between the invariants and $\c[V]$.

\begin{remark} 
In the definition of $Q_\A$ one can replace $2k_\alpha$ by $2k_\alpha+1$ in \eqref{qi}, because $q(x)-q(s_\alpha x)$ is skew-symmetric under $s_\alpha$.   
\end{remark}

Consider the Calogero--Moser operator \eqref{cmo} with $W$-invariant multiplicities $k_\alpha\in\Z_+$ and write $L=L_{W}$, $L_0=\Delta$.

\begin{theorem}\label{aicc} 
$(1)$ There exists a nonzero linear differential operator $S\in \D(\vreg)$ such that \[L\, S \,=\, S\, L_0\,.\] 
%Equivalently, $L\psi=(\lambda,\lambda)\psi$, where $\psi(\lambda,x):= Se^{(\lambda,x)}$, $\lambda\in V$.

$(2)$  
%For any quasi-invariant polynomial $q\in Q_{k}(W)$ 
There exists pairwise commuting operators $L_q\in \D(\vreg)$, $q\in Q_{k}(W)$, such that the map $ q\mapsto L_q$ defines an algebra embedding $\theta\,:\,Q_k(W)\into \D(\vreg)$.

%$(3)$ The algebra $\theta(Q_k(W))$ is a maximal commutative subalgebra of $\D(\vreg)$.  

\end{theorem}

The first statement follows from the existence of the so-called shift operators, constructed explicitly (in terms of the Dunkl operators) in \cite{He2}. Part (2) is the result of \cite{VSC}. \qed

In the next section we generalise the above theorem for the case when root systems of Coxeter groups are replaced by more general systems of vectors. These configurations can be seen as a generalisation of the so-called {\it locus configurations} from \cite{CFV}.

\section{Generalised locus configurations}
%Here we generalise the notions of locus configurations and quasi-invariants from \cite{CFV}, and state our main results concerning the corresponding Calogero--Moser operators.
\la{S3} 
Let $ (\A, k_{\alpha}) $ be a configuration of vectors with complex multiplicities in a (complex) Euclidean space $V$. We assume that the vectors of $ \A$ are non-isotropic, i.e. $\,(\alpha, \alpha)\ne 0$ for all $ \alpha \in \A\,$; the corresponding orthogonal reflections $\,s_\alpha\,$ can then be defined by the same formula as in the real case, see \eqref{refls}. We write $\, H_{\A} := \{H_{\alpha}\} \subset V\,$ for the collection of hyperplanes  $ H_{\alpha} := \Ker(1 - s_{\alpha}) $ with $ \alpha \in \A\,$. As in the Introduction, we associate to $ (\A, k_{\alpha}) $ the second order differential operator in $ \D(V\!\setminus \! H_{\A}) $:
\begin{equation}\label{gcmu}
L_{\A}=\Delta-u_{\A}\ ,\qquad u_{\A} := \sum_{\alpha\in \A_+} \frac{k_\alpha(k_\alpha+1)(\alpha,\alpha)}{(\alpha,x)^2}\,.
\end{equation}

\begin{defi}
\label{defloc}
Let $R\subset V$ be the root system of a finite Coxeter group $W$ whose action on $V$ is generated
by reflections. A configuration $\A$ is called a {\it locus configuration of type $W$} if

(1) $\A$ contains $R$, and both $\A$ and $k:\,\A\to\c$ are invariant under $W$;

(2) For any $\alpha\in\A\setminus R$ one has $k_\alpha\in\Z_+$ and the function $u_{\A}$ in \eqref{gcmu} satisfies the condition
\begin{equation}\label{loc}
 u_{\A}(x)-u_{\A}(s_\alpha x) \quad\text{is divisible by}\quad (\alpha,x)^{2k_\alpha}\,.
\end{equation}
Here, we say that a rational function $f$ on $V$ is divisible by $(\alpha,x)^{2k}$ if $(\alpha,x)^{-2k}f$ is regular at a generic point of the hyperplane $ H_{\alpha} $.
\end{defi} 
In the trivial case $W=\{e\}, \,R =\varnothing $, the above definition reduces to the notion of a locus configuration introduced in \cite{CFV}. 
Explicitly,  \eqref{loc} can be described by the following set of equations
\begin{equation}
\label{loc1}
\sum_{\beta\in\A_+\setminus\{\alpha\}}\frac{k_\beta(k_\beta+1)(\beta, \beta)(\alpha,\beta)^{2j-1}}{(\beta,x)^{2j+1}}=0\quad\text{for $\ (\alpha,x)=0\,$} %and $j=1,\dots, k_\alpha$.}
\end{equation}
which should hold for each $\alpha\in\A_+\setminus R$ and all $j=1,\dots, k_\alpha\,$.
%\begin{remark}\label{2dim}
%An important feature of \eqref{loc1} is that it is sufficient for them to hold for any %two-dimensional subconfiguration of $\A$. For the case $R=\emptyset$ this is explained %in \cite{CFV}, and the same arguments work in general. 
%\end{remark}

Note that the root system of any Coxeter group $W$ with $W$-invariant integral $k_\alpha$ obviously satisfies the condition \eqref{loc}: these are basic examples of locus configurations (with $R=\varnothing $). There exist also many examples of locus configurations which do not arise from Coxeter groups (the so-called `deformed root systems'); a complete classification of such configurations is an open problem: all currently known examples will be described in Sections \ref{exloc} and \ref{twodim} below. They include the well-known deformations related to Lie superalgebras \cite{SV}, as well as some new examples.

Before proceeding further, let us compare our class of configurations with those introduced by Sergeev and Veselov in \cite{SV}. They consider trigonometric deformed Calogero--Moser operators, so one needs to specialise their setting to the rational case. In that case the configurations they consider are defined similarly to Definition \ref{gqi}, but with the additional requirement 
$k_\alpha=1$ for $\alpha\in\A\setminus R$, and with condition (2) replaced with the following identity: 
\begin{equation}\label{mi}
\sum_
%{\alpha\ne\beta\in \A_+}
{\alpha, \beta\in \A_+,\,\alpha\ne\beta} 
\frac{k_\alpha k_\beta (\alpha, \beta)} {(\alpha, x)(\beta ,x)}=0\,.
\end{equation} 
This is a rational version of ``the main identity" \cite[(12)]{SV} that may be equivalently stated as
\begin{equation}\label{dsf}
L_{\A}(\delta_\A^{-1})=0\,,\qquad \delta_\A=\prod_{\alpha\in\A_+}(\alpha, x)^{k_\alpha}\,.
\end{equation}
It is not obvious whether \eqref{mi} and \eqref{loc1} are related. Nevertheless, by going through the list of configurations in \cite[Section 2]{SV} one can check that they all satisfy our definition ({\it cf.} remark at the end of Section~2 of \cite{SV}). On the other hand, as will become clear from the examples in Sections \ref{exloc} and \ref{twodim}, there exist locus configurations that do \emph{not} fit in the axiomatics of \cite{SV}, either because $k_\alpha>1$ or because \eqref{mi} does not hold. Therefore, our class of configurations is {\it strictly larger} than in \cite{SV}, i.e. our approach is more general.

\medskip

Next, we introduce quasi-invariant polynomials. % associated with locus configurations.

\begin{defi}\label{gqi} Let $\A$ be a locus configuration of type $W$. A polynomial $q\in\c[V]^W$ is called a {\it quasi-invariant} if 
%it is $W$-invariant and if
\begin{equation}\label{qi1}
q(x)-q(s_\alpha x)\ \text{is divisible by}\ (\alpha,x)^{2k_\alpha}\quad \forall\alpha\in\A_+\setminus R\,.
\end{equation}
Write $\qaw$ for the space of quasi-invariants. It is easy to check that $\qaw$ is a graded subalgebra of $\c[V]^W$. 
%with \delta^2\c[V]^W\subseteq \qaw\subseteq \c[V]^W\,,\qquad 
\end{defi}
Again, in the case $W=\{e\}$, $R=\varnothing$ the above definition reduces to the quasi-invariants as defined in the previous Section. 
Below we will always identify $V$ and $V^*$ using the bilinear form $(\cdot, \cdot)$, thus making no distinction between $\c[V]$ and $\c[V^*]$ and regarding $\qaw$ interchangeably as a subalgebra in $\c[V]^W$ or $\c[V^*]^W$. 

%%%%%%%%%%%%%%%%%%%%%%
\medskip

With a locus configuration $\A$ of type $W$ we associate two quantum Hamiltonians, $L_0=L_{W}\in\D(\vreg)^W$ and $L=L_{\A}\in\D(V\!\setminus\! H_\A)^W$. 
By Theorem \ref{ci}, $L_0$ is a member of a commutative family of higher-order Hamitonians $L_{q,0}:=\Res(\e T_q\e)$, $q\in\c[V^*]^W$. Our goal is to prove the following theorem that extends Theorem \ref{aicc}.

\begin{theorem}\label{gaic} Let $\A $ be a locus configuration  of type $W$.

\begin{enumerate}
\item[(1)] There exists a nonzero differential $($shift$)$ operator  
$ S\in \D(V\!\setminus\! H_\A)^W$ such that $L  S=S L_0$. 

\item[(2)] For any homogeneous quasi-invariant $q\in \qaw$ there exists a differential operator $L_q$ such that $L_q S=S L_{q,0}$ where $L_{q,0}=\Res(\e T_q\e)$. The operators $L_q$ pairwise commute and the map $q\mapsto L_q$ defines an algebra embedding 
$\theta\,:\ \qaw \hookrightarrow \D(V\!\setminus\! H_{\A})^W$. 

\item[(3)] The algebra $\qaw$ has Krull dimension $n=\dim V$, i.e. it has $n$ algebraically independent elements; thus,  $L = L_\A $ is completely integrable. 

\item[(4)] $\,\theta(\qaw)$ is a maximal commutative subalgebra in $\D(V\setminus H_\A)^W$. 
\end{enumerate}
\end{theorem}
%Recall that the rank of a commutative subalgebra of operators is the dimension of the space of joint eigenfunctions, see \cite{BrEtGa}. To the best of our knowledge, these are the first examples of maximal commutative rings of differential operators of rank $>1$ in higher dimensions.
In the case $W=\{e\}$, these results are known and their proofs can be found in \cite{C98, CFV, C08}.  
The proof of Theorem~\ref{gaic} will be given in Section~\ref{S5} as a consequence of general results 
proved in the next section.

%Our proofs are constructive: the shift operator $S$ and the commuting higher order Hamiltonians $L_q$ will be given by explicit formulas in terms of iterated skew-commutators, in the style of \cite{B98}. 
%\section{Proof of Theorem \ref{gaic}}

\section{Shift Operators}
\label{pro}
In this section, we develop an abstract algebraic approach to the problem of constructing  
differential `shift' operators. Our main result --- Theorem~\ref{IT} --- provides necessary and sufficient conditions for the existence of such operators under very general assumptions. In the next section, we will verify these conditions for Calogero-Moser operators associated with  locus configurations, 
and thus deduce our main Theorem~\ref{gaic} from Theorem~\ref{IT}. 
\subsection{Existence of shift operators}
\la{S4.0}
Throughout this section,  $k$ will denote a fixed field of characteristic zero, and all rings will be 
$k$-algebras with $1$. If $A$ and $B$ are two rings and $M$ is an $A$-$B$-bimodule, we consider $M$ as a left ($ A\otimes B^{\circ}$)-module, with an element $ a \otimes b \in A \otimes B^{\circ} $ acting on  $ m \in M $ by $ (a\otimes b)\cdot m := amb $. Then, for $ a \in A $ and $ b \in B $, we say that the pair $ (a,b) $ acts on $ M $ {\it locally ad-nilpotently} if $\, a \otimes 1 - 1 \otimes b \in A \otimes B^{\circ} $  acts locally nilpotently: i.e., for every $m\in M$, there is $ n> 0 $ such that $ (a \otimes 1 - 1 \otimes b)^n \cdot m = 0 $. We will use the following notation for this action:
$$ 
\ad_{a,b}(m) :=  (a \otimes 1 - 1 \otimes b)\cdot m = am - mb
$$
Note that, using the binomial formula, we can write the elements $\ad_{a,b}^n(m) := (a \otimes 1 - 1 \otimes b)^n \cdot m $ explicitly for all $ n \ge 0 \,$:
\begin{equation}
    \label{binom}
\ad_{a,b}^n(m) = \sum_{k=0}^n (-1)^k {n \choose k}\, a^k m \,b^{n-k}    
\end{equation}
When $A=B$ and $a=b$, this becomes the adjoint action, in which case we use the standard notation $ \ad_a $ instead of $ \ad_{a,a} \,$; we say that $a$ acts locally ad-nilpotently on $M$ if so does $(a,a)$. We call an element of a ring {\it locally ad-nilpotent} if it acts locally ad-nilpotently on the ring viewed as a bimodule.

We begin with a general lemma from noncommutative algebra.

\begin{lemma}
\label{shiftl}
Let $B$ be a noncommutative integral domain, $ S \subset B $ a two-sided Ore subset in $B$, and $ A := B[S^{-1}]$ 
the corresponding ring of fractions. Let $ L_0 \in B $ be a locally ad-nilpotent element in $B$. 
Then, for $ L \in A $, the following  conditions are equivalent:
\begin{enumerate}

\item[$(a)$] there exists a nonzero $ D \in A $ such that $\, L \,D = D\, L_0  \,$ in $A\,$,

\item[$(b)$] there exists a nonzero $ D^* \in A $ such that $\, D^* L = L_0\,D^* $ in $A\,$,

\item[$(c)$] there exists $ \delta \in S $ such that $\, \ad_{L, L_0}^{N+1}(\delta) = 0 \,$ in $A$ for some  $ N \ge 0 $.
\end{enumerate}
\end{lemma}
\begin{proof}
The implication $\,(c) \Rightarrow (a) \,$ is immediate: if $(c)$ holds, choose the smallest $ N \ge 0 $ such that
$ \ad_{L, L_0}^{N+1}(\delta) = 0 $, then $ D:= \ad_{L, L_0}^{N}(\delta) \not= 0 $ satisfies $(a)$. 

Now, assume that $(b)$ holds. Since $S$ is a right Ore subset in $B$, for $ D^* \in B[S^{-1}] $, there is $ \delta \in S $ such that $ D^*\delta \in B $. Then, by \eqref{binom}, we have 
\begin{eqnarray*}
D^* \, \ad^n_{L, L_0}(\delta) &=&  \sum_{k=0}^n (-1)^k {n \choose k}\, D^* L^k \delta \,L_0^{n-k} \\
&=&   \sum_{k=0}^n (-1)^k {n \choose k}\, L_0^k D^* \delta \,L_0^{n-k} = \ad^n_{L_0}(D^* \delta)\ ,\quad \forall\,n\ge 0\,.
\end{eqnarray*}
Since $ L_0 $ acts locally ad-nilpotently on $B$, there is $ N \ge 0 $ such that $ \ad^{N+1}_{L_0}(D^* \delta) = 0 $.
Since  $ A = B[S^{-1}] $ is a domain and $ D^* \not= 0 $, the above formula implies 
$ \ad^{N+1}_{L, L_0}(\delta) = 0 $. This proves $(b) \Rightarrow (c) $. 

Finally, assume that $(a)$ holds. Since $S$ is a left Ore subset in $B$, for $ D \in B[S^{-1}] $, there is 
$ \delta^* \in S $ such that $ \delta^* D \in B $. 
Then, by \eqref{binom}, we have  $\,\ad_{L_0, L}^n(\delta^*)\, D = \ad_{L_0}^n(\delta^* D)  \,$ for all $ n \ge 0 $, 
which implies that $ \ad_{L_0, L}^n(\delta^*) = 0 $ for $ n \gg 0 $. Taking the smallest $ N \ge 0 $ such that
$ \ad_{L_0, L}^{N+1}(\delta^*) = 0 $, we put $ D^* := \ad_{L_0, L}^{N}(\delta^*)  \not= 0 $. This satisfies
$ L_0\, D^* = D^* L$, proving the  last implication $(a) \Rightarrow (b) $.
\end{proof}

In this paper, we will be concerned with differential operators. To proceed further
we therefore make the following general assumption on our noncommutative algebra $B$.
\vspace*{1ex}
\begin{enumerate}
\item[(A)] {\it The algebra $B$ contains a commutative subalgebra $R$, which is a Noetherian domain with %a multiplicative closed subset $ S $ and the 
quotient field $\, {\mathbb K} $, such that $\,
R \subset B \subset \D({\mathbb K})\,$,
where $ \D({\mathbb K}) $ is the ring of $k$-linear algebraic differential operators on $ {\mathbb K} $, with
$ R \subset \D({\mathbb K}) $ being the natural inclusion.}
\end{enumerate}
\vspace*{1ex}

\noindent
Under assumption (A), we may think of the elements of $B$ as usual `partial differential operators with rational coefficients'. To be precise,
since $R$ is Noetherian domain, by Noether's Normalization Lemma, we can choose finitely many algebraically independent elements in $R$, say $\, x_1, \ldots, x_n \,$, so that the quotient field  $ {\mathbb K} $ of $R$ is a 
finite extension of $ k(x_1, \ldots, x_n) $. The module $ \Der_k({\mathbb K}) $  of $k$-linear derivations of
$ {\mathbb K} $ is then freely generated (as a ${\mathbb K}$-module) by the `partial derivatives'
$\, \partial/\partial x_i: {\mathbb K} \to {\mathbb K} \,$, and the 
ring $ \D({\mathbb K}) $ can be identified as
$$ 
\D({\mathbb K}) \cong {\mathbb K}\left[\partial/\partial x_1, \ldots, \partial/\partial x_n\right]\ .
$$
Next, if $ S $ is a multiplicatively closed subset in $R$, then the assumptions of Lemma~\ref{shiftl} hold automatically for $B$ and $ S $. Indeed, since $ \D({\mathbb K}) $ is a noncommutative domain (see, e.g., \cite[Theorem~15.5.5]{MR}) and $B$ is a subalgebra of $ \D({\mathbb K}) $, $B$ is a domain as well. Furthermore, since $ S \subset {\mathbb K} $, the elements of $S$ are represented by zero order differential operators on ${\mathbb K}$ which act, by definition, locally ad-nilpotently on $ \D({\mathbb K}) $, and hence {\it a fortiori} on $B$. It follows that  $S$ is a two-sided Ore subset. Note that the elements of $S$ are actually units in $ \D({\mathbb K}) $, hence, by the universal property of Ore localisation,  the inclusion $ B \into \D({\mathbb K}) $ extends to $ A := B[S^{-1}] \,$: thus, if (A) holds, we have
$$
S \subset R \subset B \subset A \subset \D({\mathbb K})
$$
for any multiplicative closed subset $S$.

Now, fix $\, S \subset R \subset B \,$ as above, and let $ L_0 $ be a locally ad-nilpotent element in $B$. Following \cite{BW}, we associate to $L_0$ a (positive increasing) filtration on $B$:
\begin{equation*}
F_0 B \subseteq F_1 B \subseteq \ldots  \subseteq F_n B \subseteq F_{n+1} B \subseteq \ldots \subseteq B 
\end{equation*}
which is defined by induction:
\begin{equation}
\la{filt}
F_{-1}B := \{0\}\ ,\quad F_{n+1} B := \{b \in B\ :\ \ad_{L_0}(b) \in F_n B\} \ ,
\end{equation}
or equivalently, 
$$ 
F_n B := \{b \in B \ : \ \ad_{L_0}^{n+1}(b) = 0\} \quad \mbox{for all}\ n \ .
$$
Since $ \ad_{L_0} $ is a locally nilpotent derivation, $\{F_\ast B\}$ is an exhaustive 
filtration on $B$ satisfying $ (F_n B) \cdot (F_m B) \subseteq F_{n+m} B $
for all $ n,m \ge 0 $. Note that $ F_0 B = C_B(L_0) $ is the centralizer of $L_0$, which is a (not necessarily commutative) subalgebra of $B$. 

Associated to  \eqref{filt} is the degree (valuation) function $\, \deg_{L_0}:\, B\!\setminus\!\{0\} \to \N  \,$ defined by
\begin{equation} 
\la{degr}
\deg_{L_0}(b) \, :=\, n \quad \mbox{iff}\quad b \in F_n B \setminus F_{n-1} B\ ,\ n \ge 0 \,.
\end{equation}
Note that $\, \deg_{L_0}(b) = n \,$ whenever $\, \ad_{L_0}^{n+1}(b) = 0 \,$ while $\, \ad_{L_0}^{n}(b) \not= 0 \,$ in $B$. It is convenient to extend $\deg_{L_0} $ to the whole $B$ by setting $ \deg_{L_0}(0) := - \infty $, so that %
$\, F_n B = \{b \in B \, : \, \deg_{L_0}(b) \leq n \}\,$ for all $n$.

The next lemma shows that, under  assumption (A), the above filtration and associated
degree function on the algebra $B$ extend to its localisation.

\begin{lemma}
\label{filtl}
%Assume that  {\rm (A)} holds. Let $ L_0 $ be a locally ad-nilpotent element in $B$
%with degree function $ \deg_{L_0}: \,B \to \N \cup \{-\infty\} $. 
There is a unique function $\, \deg: A \to \Z \cup \{-\infty\} \,$ on the localised algebra $ A = B[S^{-1}]\,$ with the following 
properties\footnote{Just as the function $ \deg_{L_0} $ on $B$, its extension to $A$ depends
on the ad-nilpotent element $L_0$. To distinguish between these two degree functions we suppress the dependence  
of `$ \deg $' on $ L_0$ in our notation.}:
\begin{enumerate}
\item[(0)] $ \deg(b) = \deg_{L_0}(b) $,  \ $ \forall\,b \in B \,$,
\item[(1)] $ \deg\,(a_1 a_2) = \deg(a_1) +\deg(a_2) $, \ $ \forall\,a_1, a_2 \in A \,$,
\item[(2)]
$ \deg\,(a_1 + a_2) = \max\{\deg(a_1),\,\deg(a_2)\} $, \ $ \forall\,a_1, a_2 \in A \,$,
\item[(3)]
$\deg\,[\ad_{L_0}(a)] \leq \deg(a) - 1$, \ $ \forall\,a \in A \,$.
\end{enumerate}
\end{lemma}

\begin{proof}
First, observe that the properties $(1)$, $(2)$, $(3)$ hold for the function $ \deg_{L_0} $ on $B$.
Indeed, for $ \deg_{L_0} $, property $\,(2)$ is immediate from the definition \eqref{degr}, while $(3)$ follows 
from the inductive construction  of the filtration  \eqref{filt}. To verify $(1)$
take two elements $\, b_1, b_2\in B $ with $  \deg_{L_0}(b_1) = n_1 \ge 0 $ and
$  \deg_{L_0}(b_2) = n_2  \ge 0 $. Then, by \eqref{degr}, 
\begin{equation}
\la{n1n2}
\ad_{L_0}^{n_1 + 1}(b_1) = \ad_{L_0}^{n_2 + 1}(b_2) = 0\ ,
\end{equation}
while $ \ad_{L_0}^{n_1}(b_1) \not= 0 $ and $ \ad_{L_0}^{n_2}(b_2) \not= 0 $. Since
$ \ad_{L_0} $ is a derivation on $B$, by using \eqref{n1n2} and Leibniz rule, we have
%
%$$
%\ad_{L_0}^n(b_1 b_2) \,=\,\sum_{k=0}^n\,{n\choose k}\,\ad_{L_0}^k(b_1)\,  %\ad_{L_0}^{n-k}(b_2)\ .
%$$
%
%for all $ n \ge 0 $. In view of \eqref{n1n2}, for $\, n = n_1 + n_2 + 1 \,$,  
%the above formula implies 
$\, \ad_{L_0}^{n_1 + n_2 + 1}(b_1 b_2) = 0\,$, while %, for $  n=n_1 + n_2 $, 
$$
\ad_{L_0}^{n_1+n_2}(b_1 b_2) \,=\, \frac{(n_1 + n_2)!}{n_1 !\, n_2 !}\, \ad_{L_0}^{n_1}(b_1)\,
\ad_{L_0}^{n_2}(b_2)\ .
$$
Since $B$ is a domain, the last equation shows that $\,\ad_{L_0}^{n_1+n_2}(b_1 b_2) \not= 0\,$, which
means that $ \deg_{L_0}(b_1 b_2) = n_1 + n_2 \,$, or equivalently,
\begin{equation}
\la{bb}
  \deg_{L_0}\,(b_1 b_2) = \deg_{L_0}(b_1) +\deg_{L_0}(b_2)\ .  
\end{equation}

Now, we define the  function $\,\deg:\,A\! \setminus\! \{0\}\,\to\, \Z $ by 
\begin{equation}
\la{degA}
\deg(s^{-1} b) := \deg_{L_0}(b) - \deg_{L_0}(s)
\end{equation}
where $ s^{-1} b \in A $ with $s \in S $ and $ b \in B $. To see that this definition makes sense
consider two different presentations of an element in $A$ by (left) fractions: say 
$\, a = s_1^{-1} b_1 = s_2^{-1} b_2 $ with $ s_1, s_2 \in S $ and $ b_1, b_2 \in B $. 
Since $ S $ is commutative (by assumption (A)), we have $\, s_2 b_1 = s_1 b_2 \,$ in $B$, 
which, by \eqref{bb}, implies
$$
\deg_{L_0}(s_2) +\deg_{L_0}(b_1) = \deg_{L_0}(s_1) + \deg_{L_0}(b_2)\ .
$$
Whence
$$
\deg(s_1^{-1} b_1) := \deg_{L_0}(b_1) - \deg_{L_0}(s_1) = \deg_{L_0}(b_2) - \deg_{L_0}(s_2)  = \deg(s_2^{-1} b_2)\ ,
$$
as required. Note that the same argument shows that 
\begin{equation}
\la{lright}
 \deg(b s^{-1}) \,=\,  \deg_{L_0}(b) - \deg_{L_0}(s) \,=\, \deg(s^{-1} b)
\end{equation}
for all $b \in B $ and $s\in S$.

Now, with definition \eqref{degA}, the property $(0)$ of Lemma~\ref{filtl} is obvious. To prove $(1)$ write 
elements $ a_1 $ and $ a_2 $ in $A$ as left and right fractions: $ a_1 = s_1^{-1} b_1 $ and $ a_2 = b_2 s_2^{-1} $, and use \eqref{lright} to conclude:
\begin{eqnarray*}
\deg(a_1 a_2) &=& \deg(s_1^{-1} b_1 b_2 s_2^{-1})\\
&=& \deg_{L_0}(b_1 b_2) - \deg_{L_0}(s_1) - \deg_{L_0}(s_2)\\
&=& \deg_{L_0}(b_1) +  \deg_{L_0}(b_2) - \deg_{L_0}(s_1) - \deg_{L_0}(s_2) \\
&=& [\,\deg_{L_0}(b_1) - \deg_{L_0}(s_1)\,] + [\,\deg_{L_0}(b_2) - \deg_{L_0}(s_2)\,] \\
&=& \deg(a_1) + \deg(a_2)\ .
\end{eqnarray*}
Note that property $(1)$ implies formally that $\deg(s^{-1}) = - \deg(s) $ for all $s\in S$;  together with $(0)$, it entails \eqref{degA}, and hence the uniqueness of the function `$\deg $'.

To prove $(2)$ take $\, a_1 = s_1^{-1} b_1 $, $\, a_2 = s_2^{-1} b_2 $ in $A$ and assume (without loss of generality) that $\, \deg(a_1) \geq \deg(a_2) \,$. Note that, by \eqref{bb} and \eqref{degA}, this last condition is equivalent to
\begin{equation}
\la{ineq}
\deg_{L_0}(s_2 b_1) \geq \deg_{L_0}(s_1 b_2)    
\end{equation}
Now, using the fact that $(2)$ holds for the degree function $\deg_{L_0} $, we check
\begin{eqnarray*}
\deg(a_1 + a_2) &=& \deg[(s_1 s_2)^{-1}( s_2 b_1 + s_1 b_2)]\\
&=& \deg_{L_0}(s_2 b_1 + s_1 b_2) - \deg_{L_0}(s_1) - \deg_{L_0}(s_2)\\
& \leq & \max\{\deg_{L_0}(s_2 b_1), \deg_{L_0}(s_1 b_2)\}  - \deg_{L_0}(s_1) - \deg_{L_0}(s_2) \\
&=& \deg_{L_0}(s_2 b_1) - \deg_{L_0}(s_1) - \deg_{L_0}(s_2) \qquad [\,\mbox{by}\ \eqref{ineq}\,]\\
&=& \deg_{L_0}(b_1) - \deg_{L_0}(s_1) = \deg(a_1) \\
&=& \max\{\deg(a_1),\,\deg(a_2)\}\ .
\end{eqnarray*}

Finally, to prove $(3)$ we take $\, a = s^{-1} b \in A \,$ and write $ \ad_{L_0}(a) $ 
in the form
$$
\ad_{L_0}(s^{-1} b) = s^{-1} \ad_{L_0}(b) - s^{-1} \ad_{L_0}(s)\, s^{-1} b
$$
Since 
$$
\deg[s^{-1} \ad_{L_0}(b)] = \deg_{L_0}[\ad_{L_0}(b)] - \deg_{L_0}(s) \leq \deg_{L_0}(b) - \deg_{L_0}(s) - 1  
$$
and similarly
$$
\deg[s^{-1} \ad_{L_0}(s)\, s^{-1} b] \leq \deg_{L_0}(b) - \deg_{L_0}(s) - 1\ ,
$$
by property $(2)$ we conclude
$$
\deg[\ad_{L_0}(a)] \leq \deg_{L_0}(b) - \deg_{L_0}(s) - 1 = \deg(a) - 1\ .
$$
This completes the proof of the lemma.
\end{proof}
Using the degree function of Lemma~\ref{filtl}, we can extend the filtration \eqref{filt}
on the algebra $ B $ to a $\Z$-filtration on the algebra $A\,$:
$$ 
F_n A := \{a \in A\ :\ \deg(a) \leq n\}\ ,\quad \forall\,n \in \Z\,.
$$
We write $\,\mathtt{gr}(A) := \oplus_{n \in \Z}\, F_n A/F_{n-1} A\,$ for the associated graded ring, and
for each $ n \in \Z $, denote by $\,\sigma_n:\, F_n A \onto F_n A/F_{n-1}A  \into \mathtt{gr}(A) \,$ the {\it symbol map of degree} $n$. By definition, for $\, a \in F_n A \,$, the symbol 
$\, \sigma_n(a) = a + F_{n-1} A \,$ is nonzero if and only if $ \deg(a) = n $. For example, we have
$\sigma_0(L_0) = L_0  + F_{-1} A\,$, since $\deg(L_0) = 0\,$.

We can now state the main result of this section.
\begin{theorem}
\la{IT}
Assume that an algebra $B$ satisfy condition {\rm (A)}, let $S$ be a 
multiplicatively closed subset in $R \subset B$, and $ A := B[S^{-1}] \,$. Assume, in addition, that there is an $R$-submodule $\, U_0 \subseteq {\mathbb K} \,$ such that
$$
B = \{a \in A\,:\, a[U_0] \subseteq U_0 \} \,.
$$
Let $ L_0 $ be a locally ad-nilpotent operator in $ B $. Then, for
an operator $ L \in A $, there is a nonzero operator $ D \in A $ such that 
\begin{equation}
\la{sheq}
L \,D = D\, L_0  
\end{equation}
if and only if the following conditions hold:

\begin{enumerate}
\item[(1)] there is a $k$-linear subspace $\, U \subseteq {\mathbb K} \,$ such that 
\begin{enumerate}
    \item[a)] $\,U$ is stable under $L$, i.e. $\, L[U] \subset U \,$,
    \item[b)] $\, s\, U_0 \subseteq U \subseteq s^{-1} U_0 \,$ for some $ s \in S \,$,
\end{enumerate}
\vspace*{1ex}
\item[(2)] $\, \sigma_0(L) = \sigma_0(L_0) \,$ $($in particular, $ \deg(L) = \deg(L_0) = 0 $$)$.
\end{enumerate}

\vspace*{1ex}

\noindent
Given a subspace $U \subseteq {\mathbb K} $ satisfying condition {\rm (1b)}, there is at most one
operator $L \in A $ satisfying  {\rm (1a)} and {\rm (2)} $($and hence the identity \eqref{sheq}$)$.
\end{theorem}
\begin{proof}
First, we prove that conditions $(1)$ and $(2)$ are sufficient for the existence of  $D$. To this end, we consider the space
of all operators in $A$ mapping $U_0$ to $U$:
\begin{equation*}
\M := \{a \in A\ :\ a[U_0] \subseteq U \}\ .
\end{equation*}
Note that $\M$ is a right $B$-module which, by (1b), contains the ideal $s B$ and is contained in $ s^{-1} B\,$:
\begin{equation}
\la{Minc}
s B \subseteq \M \subseteq s^{-1} B\ .
\end{equation}
On the other hand, by (1a), $ \M $ is closed under the action of the operator $\ad_{L, L_0}$. We claim that this last operator acts on $\M$ locally nilpotently. Indeed, by  Lemma~\ref{filtl}, it follows from the inclusion $ \M \subseteq s^{-1} B $ in \eqref{Minc} that 
\begin{equation}
\la{comp}
\deg(a)  \geq - \deg(s)\quad \mbox{for
all}\ \, a\in \M \!\setminus\! \{0\}\ .
\end{equation}
On the other hand, letting $ P := L - L_0 \in A $, we can write 
$$\ad_{L, L_0}(a) = \ad_{L_0}(a) + P a\ . 
$$ 
By condition (2), we have $ \deg(P)  \leq -1 $,  and hence
$$\
\deg(Pa) = \deg(P) + \deg(a) \leq \deg(a) - 1 
$$
for all $a\in A $. Then, by Lemma~\ref{filtl}, 
\begin{equation}
\la{comp1}
\deg[\ad_{L, L_0}(a)] 
 \leq \max\{\deg[\ad_{L_0}(a)],\, \deg(Pa)\} \leq   
\deg(a) - 1\ 
\end{equation}
Now, by \eqref{Minc}, any  element $a \in \M$ can be written in the form $ a = s^{-1} b $ with $b\in B $. If we take $ N = \deg(b) $, then, for  $ a = s^{-1} b $,  \eqref{comp1} implies  by induction
$$
\deg[\ad^{N+1}_{L, L_0}(a)] \leq \deg(a) - N - 1 = \deg(b) - \deg(s) - N - 1 = -\deg(s) - 1\ .
$$
In view of \eqref{comp}, for $\, a \in \M $, this means that $ \deg[\ad^{N+1}_{L, L_0}(a)] = - \infty \,$, i.e.
$ \ad^{N+1}_{L, L_0}(a) = 0 $. Thus $ \ad_{L, L_0} $ acts on $\M$ locally nilpotently. Now, since $ s\in \M $ by (1b), we have $ \ad^{N+1}_{L, L_0}(s) = 0 \,$ with $ N = 2 \deg(s) $. This implies the existence of $D$ by Lemma~\ref{shiftl}.

Conversely, suppose that there is $ D \not= 0 $ in $A$ such that $\, L\,D = L_0\, D \,$. This last equation can be rewritten in the form $ \ad_{L_0}(D) = - P D $, where $ P := L - L_0 $. Hence, by Lemma~\ref{filtl}, 
$$
\deg(P) + \deg(D) = \deg(P D) = \deg[\ad_{L_0}(D)] \leq \deg(D) - 1\ ,
$$
which implies $ \deg(P) \leq -1\,  $. Thus  (2) holds.

To construct a subspace $ U \subseteq {\mathbb K} $ satisfying condition (1) we apply Lemma~\ref{shiftl}. According to this lemma,  there is an element $ \delta \in S $ such that  $ \ad^{N +1}_{L,L_0}(\delta) = 0 $ for some $ N \ge 0 $. We put $ D_k := \ad_{L, L_0}^k(\delta)  $ for $\, k=0,1,2,\ldots, N+1$, with $ D_{N+1} = 0 $, and define $U$ to be the smallest subspace of ${\mathbb K}$ that contains the images of $U_0$ under the $D_k$'s for all $k\,$: i.e.,
$$
U := \sum_{k=0}^N \,D_k[U_0]\,\subseteq \, {\mathbb K} \ .
$$
Since $\, D_{k+1} = L D_k - D_k L_0 \,$, we have
$$
L D_k[U_0] = D_{k+1}[U_0] + D_k L_0[U_0] \subseteq D_{k+1}[U_0] + D_k[U_0] \subseteq U
$$
for $\,k=0,1, \ldots, N\,$. Hence $\,L[U] \subseteq U\,$, which is condition (1a). To prove (1b) note that,
by construction, all the $D_k$'s are in $A$, hence there are elements
$ \delta_k \in S $ such that $ \delta_k D_k \in B $ for all
$k$. Put $\,s := \delta \,\delta_1 \ldots \delta_N \in S\,$. Then
$$
s \,U =  \sum_{k=0}^N \,s \,D_k[U_0]
\subseteq \sum_{k=0}^N B[U_0] = U_0\ .
$$
On the other hand, since $U_0$ is an $R$-module and $S \subset R $, we have
$$
s \,U_0 = \delta\,(\delta_1 \ldots \delta_N  U_0) \subseteq \delta \,U_0 = D_0[U_0] \subseteq U\ .
$$
Thus, $\,s\, U_0 \subseteq U \subseteq s^{-1} U_0\,$ for $s \in S$, which is the required condition (1b).

To prove the last claim of the theorem consider two operators $L_1$ and $L_2$ in $A$ satisfying 
{\rm (1a)} and {\rm (2)} for a given subspace $U$ which satisfies {\rm (1b)}. Put $ P := L_1 - L_2 $. Then, by {\rm (1a)}, $\, P[U] \subseteq U \,$, while by {\rm (1b)} and Lemma~\ref{filtl}, 
$$
\deg(P) = \deg(L_1 - L_0 + L_0 - L_2) \leq \max\{\deg(L_1 - L_0), \deg(L_2 - L_0)\} < 0 \ .
$$
The first condition implies that $ P \in \End_{B}(\M) $ so that $ P a \in \M $ for all $ a \in \M $,
while by the second, $\, \deg(Pa) < \deg(a) \,$. Taking $ a \not= 0 $ to be of minimal degree in $\M$,
we conclude $ Pa = 0 $ which means that $P=0$ or equivalently $L_1 = L_2$. 
This finishes the proof of the theorem.
\end{proof}
\begin{remark}
\la{Rsymbol}
Note that  an operator $L \in A $ satisfies condition (2) of Theorem~\ref{IT} if and only if
$ L = L_0 + P $ with $\, \deg(P) < 0 \,$. By Lemma~\ref{filtl}, the last inequality holds for $P \in A $ iff
there is an $\, s \in S \,$ and $ \,n \ge\, 0 $ such that $ s P \in B $ and $ \ad^n_{L_0}(sP) = 0 $,
while $ \ad_{L_0}^n(s) \not= 0\,$. In practice, these conditions are easily verifiable. In applying 
Theorem~\ref{IT} the main problem is to verify condition (1).

\end{remark}

\begin{remark}
\la{LinB}
Under the assumptions of Theorem~\ref{IT}, for an operator $L $ in $ B $, the identity  $\, L D = D L_0 \,$ may hold  (with nonzero  $ D \in \D(\mathbb K) $) if and only if $ L = L_0 $. This follows from the last claim of the theorem.
\end{remark}
\begin{remark}
\la{Rfamily}
Theorem~\ref{IT} extends naturally to the case when a single ad-nilpotent operator $L_0 \in B $ is 
replaced by an abelian ad-nilpotent family $ \C_0 \subset B $ (in the sense of \cite{BW}). The filtration $ F_\ast B $ is defined in this case by $\,F_{n+1} B := \{b \in B\,:\, \ad_{L_0}(b) \in F_n B\ \mbox{for all}\ L_0 \in \C_0\}\,$, and the associated  degree function on $B$ determines --- under Assumption (A) --- a degree function `$\deg $' on  $A = B[S^{-1}]$ with the same properties as in Lemma~\ref{filt}. The generalisation of Theorem~\ref{IT} says that, for a family of operators $ \C \subset A $, there is a nonzero $D \in A$
such that $\,\C \,D \,=\, D\, \C_0\,$ if and only if conditions $(1)$ and $(2)$ hold for all $ L \in \C $.
The family $\C $ is then necessarily abelian, and the algebra generated by $\C$ in $A$ is a commutative ad-nilpotent subalgebra of $ \End_{B}(\M) $. We will construct examples of such subalgebras in Section~\ref{S4.2} below.
\end{remark}
We give three general classes of examples of algebras  satisfying the assumptions of Theorem~\ref{IT}.

\begin{example}
\la{weyl}
Let $V$ be a finite-dimensional vector space over $\c $. Take
$ R = \c[V] $ to be the algebra of polynomial functions on $V$, and $B = \D(V) $ the ring of differential operators on $\,U_0 = R = \c[V] $. Then $ B \cong A_n(\c) $, where $ A_n(\c) = \c[x_1, \ldots, x_n; \partial_1, \ldots, \partial_n] $ is the $n$-th Weyl algebra with $n = \dim(V) $. The algebra $ A_n(\c) $ contains the commutative subalgebra $ \c[\partial_1, \ldots, \partial_n] $ of constant coefficient differential operators $ L_0 = P(\partial_1,\ldots, \partial_n) $ which act locally ad-nilpotently on $B$. 
In the one-dimensional case ($n=1$), there is a well-known inductive construction of shift operators, using 
the classical Darboux transformations, that works for an arbitrary $L_0$. This elementary construction does not extend to higher dimensions: for $n>1$, only some ad hoc constructions and a few explicit examples are known (see, e.g., \cite{B98}, \cite{BK}, \cite{BCM}, \cite{C98}, \cite{CFV1}, \cite{CFV} and references therein).
\end{example}
\begin{example}
\la{cherednik}
Let $B = B_k(W) $ be the spherical Cherednik algebra associated to a finite Coxeter group $W$ acting in its reflection representation $V$. Take $ R = \c[V]^W $ and $ U_0 = \c[V_{\rm reg}]^W $, where $ V_{\rm reg} $ is the (open) subvariety of $V$ (obtained by removing the reflection hyperplanes of $W$) on which $W$ acts freely (see Section~\ref{S2}) . It is well-known that $B$ contains a maximal commutative subalgebra of $W$-invariant differential operators $ L_{q,0} = {\rm Res}(e\, T_q \,e) $
associated to $ q \in \c[V^*]^W $ that act locally ad-nilpotently on $B$ (see, e.g., \cite{BEG}). The  Calogero-Moser operator $ L_W $ defined by \eqref{wcm} is a special example of the $L_{q,0}$ corresponding to the quadratic polynomial $q = |\xi|^2$ ({\it cf.} Theorem~\ref{ci}). The generalised Calogero-Moser operators $L_\A $ given by \eqref{gcm} are examples of the operators $L$ related to $L_0 = L_W $ by a shift operator in a properly localised Cherednik algebra;  in the next section, we will describe the subspaces 
$U = U_{\A}$ associated to these operators explicitly in terms of locus conditions. This is the main
example of the present paper.
\end{example} 
\begin{example}
\la{lie}
Let $G$ be a complex connected reductive algebraic group, $\g = {\rm Lie}(G) $ its Lie algebra. Take
$B = \D(\g)^G $ to be the ring of invariant polynomial differential operators on $ \g $ with the respect to the natural (adjoint) action of $G$ on $\g$. The algebra $B$ contains the subalgebra $ R = \c[\g]^G $ of invariant polynomial functions on $\g$ and acts naturally on $ U_0 = \c[\g_{\rm reg}]^G $, where $ \g_{\rm reg} \subset \g $ is the (open) subvariety of regular semisimple elements of $\g$ on which $G$ acts freely.
Moreover, $B$ contains the commutative subalgebra $\c[\g^*]^G $ of constant coefficient invariant differential operators $L_0$ which act locally ad-nilpotently on $B$. A special example of such an $L_0$
is the second order Laplace operator $\Delta_{\g} $ defined for a $G$-invariant metric on $\g$. 
%(In the rank one case, for $ \g = \mathfrak{sl}_2(\c) $, there is a well-known algebra isomorphism $ \D(\g)^G  \cong \mathcal{U}(\g)$, with $\Delta_{\g} $ corresponding to the canonical generator $ e \in \mathfrak{sl}_2(\c) $.) 
Applications of Theorem~\ref{IT} to this example seems to deserve a separate study.
Of particular interest is a relation to the previous example: specifically, the question whether the generalised Calogero-Moser operators constructed in this paper can be obtained via  the (properly  localised) deformed Harish-Chandra map $ \Phi_k: \D(\g)^G \to B_k(W) $ constructed in \cite{EG}? 
\end{example}
\subsection{Morita context}
\la{S4.1}
We return to the situation of Theorem~\ref{IT}. We take an operator $L$ satisfying conditions $(1)$ and $(2) $ of the theorem, fix a subspace $ U \subseteq {\mathbb K}$ corresponding to $L$ and consider the module $\M$  of all operators
in $A$ mapping $U_0 $ to $U$ (as defined in the proof of Theorem~\ref{IT}). This last module has some interesting  algebraic properties that we will describe next. 

First, we remark that the subspace $\,U \,$ satisfying condition $(1)$ of Theorem~\ref{IT} is not uniquely determined by  $L$. However, given two such subspaces, say $U_1$ and $U_2$, their sum $U_1 + U_2$ also satisfies $(1)$. Indeed, if $\,s_1 U_0 \subseteq  U_1 \subseteq s_1^{-1} U_0 \,$ and $\,s_2 U_0 \subseteq  U_2 \subseteq s_2^{-1} U_0 \,$, then $\,s U_0 \subseteq  U_1+U_2 \subseteq s^{-1} U_0 \,$ for $s = s_1 s_2 \in S $, while obviously $\,L[U_1 +U_2] \subseteq U_1 + U_2 \,$
whenever $L[U_1] \subseteq U_1 $ and $L[U_2] \subseteq U_2 $. This implies that the poset of all subspaces $ U \subseteq {\mathbb K}$ satisfying  $(1)$ has {\it at most} one maximal element --- the largest subspace $U_{\rm max}$. We will see that in our basic example --- for the operator $ L = L_\A $ associated to a generalised locus configuration --- such a subspace always exists (Lemma~\ref{inu}).  In what follows, we will therefore study the module $\M$  for the maximal subspace $ U = U_{\rm max}$, assuming that the latter exists.

Next, we recall a few basic definitions from noncommutative algebra.
For a right $B$-module $\M$, we will denote by $\M^* := {\rm Hom}_B(\M,B) $ its dual, which is naturally a left $B$-module (via left multiplication of $B$ on itself). Applying the Hom-functor twice, we get the  double dual $\, \M^{**} := {\rm Hom}_B({\rm Hom}_B(\M,B), B) $,  which is a right $B$-module equipped with a canonical map $ \M \to \M^{**} $. A $B$-module $\M\,$ is called {\it reflexive}  if the canonical map $\,\M \to \M^{**} $ is an isomorphism. It is easy to see that every f.g. projective
(in particular, free) module is reflexive but, in general, a reflexive module need not be projective. 
If $B$ is a Noetherian domain, we write $\QQ = \QQ(B) $ for the quotient skew-field\footnote{Recall, for a (left and/or right) Noetherian domain $B$, the set $ S = B \! \setminus \!\{0\} $ of all nonzero elements of $B$ satisfies a (left and/or right) Ore condition (Goldie's Theorem), and the quotient skew-field $ \QQ(B) $ is obtained in this case by Ore localisation $B[S^{-1}]$.} of $B$, and call $\M $ a {\it fractional  ideal} if $\M $ is a right submodule of $ \QQ $ such that $\,p B \subseteq \M \subseteq q B\,$ for some nonzero $p,q \in \QQ $. 
Furthermore, if $B$ is a Noetherian domain satisfying our condition (A), then 
$ B \subseteq \D({\mathbb K}) \subseteq \QQ $.
%; in this case, we call a fractional ideal $\M$ of $B$
%\,{\it fat}\, if $ \M \subseteq \D({\mathbb K})\,$ and $\, \M \cap {\mathbb K} \not= \{0\} $.
Finally, we recall the definition of the $B$-module $\M$ from Theorem~\ref{IT}:
\begin{equation}
\la{MB}
\M := \{a \in A\ :\ a[U_0] \subseteq U \}\ ,
\end{equation}
and in a similar fashion, we define the ring
\begin{equation}
\la{D}
\D := \{a \in A\ :\ a[U] \subseteq U \}\ .
\end{equation}
\begin{prop}
\la{fatrefl}
Assume that the algebra $B$, the operators  $ L_0 \in B $ and $L \in A $ satisfy the assumptions of Theorem~\ref{IT}. In addition, assume that $B$ is  Noetherian and the subspace $U \subseteq {\mathbb K}$ 
associated to $L$ by Theorem~\ref{IT} is maximal. Then 

$(a)$ $\,\M$  is a reflexive fractional ideal of $B$;

$(b)$ $\,\D \cong \End_B(\M) $, where $ \End_B(\M)$ is the endomorphism ring of  $\M$.

\end{prop}
\begin{proof}
$(a)$ Note that $\M$ being a fractional right ideal of $B$ follows immediately
from condition (1b) of Theorem~\ref{IT}: see \eqref{Minc}. We need  only to prove that $\M $
is reflexive. If $\M_1$ and $\M_2 $ are two fractional (right) ideals of $B$,  we can identify (see \cite[3.1.15]{MR}):
\begin{equation}
\la{IHom}
{\rm Hom}_B(\M_1, \M_2) \cong \{q \in \QQ\ :\ q \M_1 \subseteq \M_2 \}\ .
\end{equation}
In particular, 
\begin{equation}
\la{Mdual}
\M^* \cong \{q \in \QQ\ :\ q \M \subseteq B \}
\end{equation}
Now, in addition to the right $B$-module $\M$, we introduce the left $B$-module 
\begin{equation}\label{nn}
\NN :=  \{a \in A\ :\ a[U] \subseteq U_0 \} \ .
\end{equation}
By condition (1b) of Theorem~\ref{IT},  
\begin{equation}
\la{Ninc}
Bs \subseteq \NN \subseteq Bs^{-1}\, ,
\end{equation}
which shows that $\NN$ is a fractional left ideal. Since $B = \{a \in A\ :\ a[U_0] \subseteq U_0 \} $,
we have $\,\NN  \M \subseteq B \,$. With identification \eqref{Mdual}, this implies $\, \NN \subseteq \M^* \,$. Dualizing the last inclusion yields $\,\M^{**}\,\subseteq\, \NN^* $. On the other hand, for any fractional ideal, we have $\,\M \subseteq \M^{**} $. Hence, to prove that $\M$ is reflexive it suffices to show
\begin{equation}
\la{NM}
\NN^* \subseteq \M\, .
\end{equation}
We prove \eqref{NM} in two steps. First, we define
$\,\M^{\circ} := \{a\in A : \ad_{L, L_0}^N(a)=0 \ \, \text{for\ some}\ \, N\ge 0\} \,$ 
and show that
\begin{equation}
\la{NMc}
\NN^* \subseteq \M^\circ .
\end{equation}
Then, we will prove 
\begin{equation}
\la{McM}
\M^\circ \subseteq \M \,.
\end{equation}
To see \eqref{NMc} we identify $\NN^* \cong \{q \in \QQ\, :\, \NN q \in B\} \,$ similar to \eqref{Mdual}. Since
$\,s \in \NN \,$, for any $ q \in \NN^* $, we have $ s q \in B $, which implies
$q \in A $. Hence $ \NN^*  \subseteq A \,$. On the other hand, the
inclusion $ \NN \subseteq B s^{-1} $ in \eqref{Ninc} implies $\, \deg(a) \ge - \deg(s) \,$ for all $ a \in \NN $. Then,
the same argument as in the proof of Theorem~\ref{IT} shows that $ \ad_{L_0, L}$ acts on
$\NN$ locally nilpotently. In particular, for $s \in \NN $, there is $ N = N_s \ge 0 $ such that
$\ad_{L_0, L}^{N+1}(s) = 0 $, while $\,\ad_{L_0, L}^{N}(s) \not= 0 \,$. Set $\,S^* := \frac{1}{N!}\,\ad_{L_0, L}^{N}(s) \in \NN \,$, so that $\, L_0 S^* = S^* L\,$. Now, for any $\, q \in \NN^* $, we have 
$\,S^* q \in \NN  \,\NN^* \subseteq B \,$. Since $L_0$ acts on $B$ locally ad-nilpotently, there is $ n \ge 0 $
such that
$$ 
\ad^n_{L_0}(S^* q) = S^* \,\ad^n_{L, L_0}(q)  = 0\,.
$$
This implies $\,\ad^n_{L, L_0}(q) = 0\,$ since $S^* \not=0 $. Thus $\, q \in \M^\circ \,$
for any $ q \in \NN^* $, which proves \eqref{NMc}.

To prove  \eqref{McM} it suffices to show (by induction) that for $\,a \in A\,$,
$$ 
\ad_{L, L_0}(a) \in \M \ \Rightarrow\ a \in \M\, .
$$
Note that, if $ \ad_{L, L_0}(a) \in \M $, then
$$
La[U_0] = \ad_{L, L_0}(a)[U_0] + a L_0[U_0] \subseteq U + a[U_0]
$$
Hence, if we set $ \tilde{U} := U + a[U_0] \subseteq {\mathbb K} $, then $L[\tilde{U}] \subseteq \tilde{U} $,
i.e.  $\tilde{U} $ satisfies condition (1a) of Theorem~\ref{IT}. On the other hand, since $ a \in A $, we can find $ s' \in S $ such that $ s' a \in B $. Taking $\, \tilde{s} := s s' \in S \,$, with $ s \in S $ as in  (1b) of Theorem~\ref{IT}, we have $\,\tilde{s}\, U_0 \subseteq s\, U_0 \subseteq U \subseteq \tilde{U}\,$
and
$$
\tilde{s}\,\tilde{U} =  \tilde{s}\,U + \tilde{s}a[U_0] = s'(s\,U) + s(s'a[U_0]) \subseteq s'\,U_0 + s B[U_0] \subseteq U_0
$$
Thus, $\,\tilde{s}\,U_0 \subseteq \tilde{U} \subseteq \tilde{s}^{-1}\,U_0\,$ for $\tilde{s} \in S $, i.e. the $\tilde{U} $ also satisfies condition (1b) of Theorem~\ref{IT}. Since $\, U \subseteq \tilde{U} \,$, by maximality  of $U$, we conclude that $ \tilde{U} = U $ which implies that $ a[U_0] \subseteq U $, i.e. $ a \in \M $. This proves \eqref{McM}.

Summing up, we have shown that
$$
\M \,\subseteq \, \M^{**} \,\subseteq\, \NN^* \,\subseteq\, \M^\circ\,\subseteq \M\ .
$$
Thus, all these subspaces in $\QQ$ are equal. In particular, we have $ \M = \M^{**} $, which proves the
reflexivity of $\M$.

$(b)$ By \eqref{IHom}, we can identify $ \End_B(\M) \cong \{q \in \QQ\ :\ q \M \subseteq \M\}\,$.
Since $\M$ is naturally a left $\D$-module, we have $\,\D \subseteq \End_B(\M) \,$ via left multiplication in
$\QQ $. We need only to show the opposite inclusion 
\begin{equation}
\la{inc}
\End_B(\M) \subseteq \D \ 
\end{equation}
This can be proved in the same way as \eqref{NM} in part $(a)$: first, one defines the ring
$\,\D^\circ := \{a\in A : \ad_{L}^N(a)=0 \ \, \text{for\ some}\ \, N\ge 0\}\,$ and shows that
$\,\End_B(\M) \subseteq \D^\circ\,$, then one proves the inclusion $\D^\circ \subseteq \D $ arguing by 
induction (downwards) in $N$. Note that, just as in part $(a)$, the maximality of $U$
is needed only for the last inclusion. Thus we get the chain of subalgebras in $\QQ $:
$$
\D \subseteq \End_B(\M) \subseteq \D^\circ \subseteq \D \,,
$$
proving that all three are equal. This finishes the proof of the proposition.
\end{proof}
\begin{remark}\label{intri}
The proof of Proposition~\ref{fatrefl} shows that
\begin{eqnarray*}
\M &=& \{a\in A\,:\, \ad_{L, L_0}^N(a)=0\ \text{for some $ N\ge 0$}\}\,, \la{Mint}\\
\D &=& \{a\in A\,:\, \ad_L^N(a)=0\ \text{for some $N\ge 0$}\}\,, \la{Dint}
\end{eqnarray*}
which gives an intrinsic characterisation of  \eqref{MB} and \eqref{D} for the maximal $ U $. 
Dually, if we assume the maximality of $U_0$, i.e. that the $ U_0 $ is maximal among all subspaces $\, \tilde{U}_0 \subseteq U_0[S^{-1}] $ such that $ L_0[\tilde{U}_0] \subseteq \tilde{U}_0 $ and $ U_0 \subseteq \tilde{U}_0 \subseteq s^{-1} U_0 $ with $s \in S $, then we get  $ \NN = \NN^{**} = \M^* $ and 
\begin{eqnarray*}
\NN &=& \{a\in A\,:\, \ad_{L_0, L}^N(a)=0\ \text{for some $ N\ge 0$}\}\,,\\
B &=& \{a\in A\,:\, \ad_{L_0}^N(a)=0\ \text{for some $N\ge 0$}\}\,.
\end{eqnarray*}
\end{remark}

\vspace*{2ex}

Proposition~\ref{fatrefl} shows that the quadruple $(\M,\,\M^*,\,B,\,\D) $ forms a {\it Morita context} (in the sense of \cite[1.1.5]{MR}). It is natural to ask when this context gives an actual 
{\it Morita equivalence} between the algebras $B$ and $\D \,$: i.e., when do these algebras have equivalent module categories? Standard ring theory provides necessary and sufficient conditions for this in the form (see \cite[Cor. 3.5.4]{MR}):
$$
\M^* \M = B \quad \mbox{and} \quad  \M \,\M^* = \D \ .
$$
In general, these conditions are not easy to verify; however, in our situation, they hold 
automatically under additional homological assumptions on $B$:
\begin{cor}
\la{proj}
Assume that $B$ is a simple Noetherian ring of global dimension $ {\rm gldim}(B) \leq 2\,$. Then  $\D$ is Morita equivalent to $ B $; in particular,  $\D $ is a simple Noetherian ring of global dimension $ {\rm gldim}(\D) = {\rm gldim}(B)\,$. Moreover, if $U_0$ is a simple
$B$-module, then $U$ is a simple $\D$-module.
\end{cor}
\begin{proof} 
It is a standard fact of homological algebra that every nonzero reflexive module over a Noetherian ring of global dimension $ \leq 2\,$ is f.g. projective (see, e.g., \cite{Bass}). Hence, by part $(a)$ of Proposition~\ref{fatrefl}, the $B$-module $\M $ is f.g. projective; then part $(b)$ --- together with Dual Basis Lemma \cite[3.5.2]{MR} --- implies  $\, \M \,\M^\ast = \D \,$. On the other hand, if $B$ is a simple domain, we have  automatically $\,\M^* \M = B \,$, since $ \M^* \M $ is a (nonzero) two-sided ideal in $B$. Thus, by \cite[Cor. 3.5.4]{MR},  $B$ and $\D$ are Morita equivalent algebras. Being Noetherian, simple and having global dimension $n$ are known to be Morita invariant properties of rings, hence $\D$ shares these properties with $B$. 

To prove the last statement consider the map of left $B$-modules 
$$ 
f: \M \otimes_B U_0 \to U
$$ 
given
by the action of operators in $ \M $ on $U_0$. The cokernel of this map, $ \Coker(f) = U/\M[U_0] $, has a nonzero annihilator in $\D $: indeed, for $\,s \in S\,$ as in (1b) of Theorem~\ref{IT}, we have $\, s^2 U = s(s U) 
\subseteq s U_0 \subseteq sB[U_0] \subseteq \M[U_0])\,$ by \eqref{Minc}. Hence 
$\, \Coker(f) = 0\,$, since $ \D $ is simple. On the other hand, since $\M$ is a progenerator in 
$ \mathtt{Mod}(B) $, the $\D$-module $\,\M \otimes_B U_0 \,$ is simple, whenever $U_0$ is simple. Hence
$ \Ker(f) = 0 $. It follows that $f$ is an isomorphism and  $U$ is a simple $\D$-module.
\end{proof}
\begin{remark}
In the last statement of Corollary~\ref{proj}, we can replace the assumption that $U_0$ is a simple $B$-module by  $U_0$ being a finite $R$-module. The latter implies the former by an argument of
\cite[Proposition 8.9]{BW}.
\end{remark}
\subsection{Commutative subalgebras}
\la{S4.2}
The results of the previous section show that the  algebras $B$ and $\D$ containing the operators $L_0$ and $L$ share many common properties, provided $L_0$ and $L$ are related by the `shift' identity \eqref{sheq}. In this section, we will construct two commuting families of operators  (including $L_0$ and $L$) that generate two isomorphic commutative subalgebras in $B$ and $\D$ intertwined by a common shift operator $S$. It is interesting to note that the operator $S$ may differ from the operator $ D $ that appears in \eqref{sheq}: in general, there seems to be no simple relation between these two shift operators.

We will keep the assumptions of Theorem~\ref{IT} and keep using the notation
from the previous section. In addition, we will introduce a new notation: for a ``multiplicative'' version
of the operator $ \ad_{a,b} $ defined in the beginning of Section ~\ref{S4.0}. Specifically, for an algebra $A$ and a pair of elements $a, b\in A \,$, we define a linear map $\,\Ad_{a,b}: A \to A[[t]]\,$ with values in the ring of formal power series over $A$, by 
\begin{equation}
\la{Ad}
\Ad_{a,b}(x) \,:= \, \sum_{n=0}^{\infty}\, \frac{t^n}{n!}\,\ad^n_{a,b}(x)\ .
\end{equation}
(As in the case of `$\ad$', we will simply write $\Ad_a$ instead of $\Ad_{a,a}$ when $a=b$.) 

Note that $(a,b)$ acts locally ad-nilpotently on $x \in A $ if and only if $\Ad_{a,b}(x) \in A[t]\,$, where $A[t] \subset A[[t]] $ is the subring of polynomials in $t$ with coefficients in $A$. Moreover, \eqref{Ad} has the following useful `multiplicative' property.

\begin{lemma}
\la{Adlem}
For all $x,y \in A $, the following identity holds in $A[[t]]\,$:
\begin{equation}
\la{Adid}
\Ad_{a,c}(xy) \,=\,\Ad_{a,b}(x)\,\Ad_{b,c}(y)
\end{equation}
%
%In particular,  $\,\Ad_a: A \to A[[t]] \,$ is an injective
%algebra homomorphism, which extends formally $($by linearity$)$ %to an automorphism of
%$A[[t]]$ with $ \Ad_a^{-1} = \Ad_{-a}$.
%
\end{lemma}
\begin{proof}
The coefficient under $ t^n $ in the left-hand side of \eqref{Adid} is $\, \frac{1}{n!}\,\ad^n_{a,c}(xy) \,$, while in the right-hand side,
$$
\sum_{n_1 + n_2 = n} \frac{1}{n_1!\, n_2!}\ \ad_{a,b}^{n_1}(x)\ \ad^{n_2}_{b,c}(y)
$$
Thus \eqref{Adid} is equivalent to the sequence of identities in $A$:
\begin{equation}\label{coad}
\ad_{a,c}^n(xy)\,=\, \sum_{k=0}^{n}\,{n\choose k} \,\ad_{a,b}^k(x)\ \ad_{b,c}^{n-k}(y) \ ,\quad \forall\,n\ge 0 \ ,
\end{equation}
which can be easily verified by induction using the following `twisted' version of the Leibniz rule
$$
\ad_{a,c}(xy)\, =\, \ad_{a,b}(x) \,y \,+\, x \,\ad_{b,c}(y)\ .
$$
An alternative way to prove \eqref{Adid} is to use the identity
\begin{equation}
    \la{Adexp}
\Ad_{a,b}(x) \,=\, e^{ta} x\,e^{-tb}
\end{equation}
that formally holds in $A[[t]]$. To see \eqref{Adexp} it suffices to notice that
the both sides of \eqref{Adexp} agree at $ t = 0$, while satisfy the same differential equation $\,d F(t)/dt = \ad_{a,b}[F(t)]\,$ for $ F(t) \in A[[t]]$.
\end{proof}

Now, let $L_0 \in B $ and $ L \in A $ be as in Theorem~\ref{IT}, and let $ U \subseteq {\mathbb K} $ be a subspace (not necessarily maximal) associated to $L$.
Recall the fractional ideal $\M$, see \eqref{MB}, and the algebra $\D$, see \eqref{D}, attached to $U$. As shown in the proof of Theorem~\ref{IT}, $\,\ad_{L, L_0}\,$ acts locally nilpotently on $\M$; in particular, if we take $ s \in \M $
as in (1b), then $ \ad_{L,L_0}^{N+1}(s) = 0 $ for some $ N \leq 2 \deg(s)$. 
We take the smallest $N \in \N $ with this property and put
\begin{equation}
\la{shiftop}
S \,:=\, %\frac{1}{N!}\,
\ad_{L, L_0}^N(s)\,\in\,\M
\end{equation}
so that $ S \not= 0 $ while $\, L S = S L_0 \,$. Using \eqref{shiftop}, it is easy to show that $L$ is locally ad-nilpotent in
$\D\,$. Indeed, as $ \D \M \subseteq \M $,  we have $\, a S \in \M $ for any $ a\in \D $, and therefore $\, \ad^n_L(a)\,S = \ad^n_{L, L_0}(aS) = 0 \,$ for $\,n\gg 0\,$, which implies $ \ad^n_{L}(a) = 0 $ since $S \not=0 $.

Now, we define 
\begin{equation}
\la{Qring}
Q \,:= \,\D \cap R\, =\, \{q \in R \ : \ q\,U \subseteq U\}    
\end{equation}
which is a commutative subring in $ \D $. Note that $Q $ is nontrivial: i.e. $\,Q \not= \{0\}$, since at least $\, s^2 \in Q \,$ by condition (1b). Note also that $ Q \subseteq R \subseteq B $, i.e. $Q$ is a common commutative subring of $B$ and $\D$. Using the fact that $L_0$ is locally ad-nilpotent in $B$ and
$L$ is locally ad-nilpotent in $\D$, we define for every $\, q \in Q\,$:
\begin{eqnarray}
L_{q,0} & := & \frac{1}{N_{q,0} !}\,\ad_{L_0}^{N_{q,0}}(q)\,, \la{Loq}\\
L_{q} & := & \frac{1}{N_q!}\,\ad_{L}^{N_q}(q)\, , \la{Lq}
\end{eqnarray} 
where  $N_{q,0} \ge 0 $ and $ N_q \ge 0 $ are chosen to be the smallest numbers
such that $ \ad_{L_0}^{N_{q,0}+1}(q) = 0 $ and $ \ad_{L}^{N_{q}+1}(q) = 0 \,$. Thus, by definition, $L_{q,0} \in B $ and $L_q \in \D $ are nonzero operators satisfying $\,[L_{q,0},\,L_0] = 0 \,$ and $\,[L_q, \,L] = 0\,$. In addition, we have
\begin{prop}
\la{commprop}
The operators \eqref{Loq} and \eqref{Lq} commute in $B$ and $A$: i.e., 
\begin{equation}
\la{commid}
[L_{q, 0},\,L_{q',0}]\,=\,0\quad ,\quad [L_{q},\,L_{q'}]\,=\,0\ ,\quad \forall\, q,\,q' \in Q\ .
\end{equation}
Moreover, for all  $ q \in Q $, we have 
\begin{equation}
\la{commsh}
L_q\,S \,=\, S\, L_{q,0}\,, 
\end{equation}
where $S$ is the operator defined by \eqref{shiftop}.
\end{prop}
\begin{proof}
The commutation relations \eqref{commid} and \eqref{commsh} are proved  in a similar way, using the identity \eqref{Adid} of Lemma~\ref{Adlem}. For example, to prove \eqref{commsh} we apply  \eqref{Adid} 
to $ x = q \in Q $ and $ y = s  $ as in (1b):
\begin{equation}
\la{Adeq}
\Ad_{L}(q)\,\Ad_{L, L_0}(s)\,= \,\Ad_{L, L_0}(qs) \, =\,  
\Ad_{L, L_0}(sq)\,=\, \Ad_{L,L_0}(s) \,\Ad_{L_0}(q)
\end{equation}
Notice that, by ad-nilpotency, all the $\Ad$'s in equation \eqref{Adeq} take values in the polynomial ring
$A[t]$. Then, comparing the leading coefficients of polynomials in both sides of  \eqref{Adeq} gives precisely the identity \eqref{commsh}. Also, comparing the degrees (in $t$) of these polynomials shows that $\,N_{q} = N_{q,0} \,$ in \eqref{Loq} and \eqref{Lq}.
%, thus, both $L$ and $L_0$ induce the same filtration on $Q$.  
\end{proof}
As a consequence of (the proof of) Proposition~\ref{commprop}, we have
\begin{cor}
\la{commcor}
The ad-nilpotent filtrations \eqref{filtl} defined by $L_0$ and $L$ on the algebras $B$ and $\D$ induce the same filtration on their commutative subalgebra $Q \subset B \,\cap\,\D $. The associated graded algebra $\grd(Q)$ embeds into $B$ and $\D$ via the mappings $q\mapsto L_{q,0}$ and $q\mapsto L_{q}$, respectively. Thus, the operators $\{L_{q,0}\} $ and $\{L_{q}\} $ generate two commutative subalgebras in $B$ and $\D$, each isomorphic to  $\,\grd(Q)$.
\end{cor}
\begin{remark}
The operators $L_0 $ and $L$, although commuting with $ L_{q,0} \in B $ and $ L_{q} \in \D $ for all $ q \in Q $, may not belong to the subalgebras generated by these last operators. Thus, the commutative subalgebras
generated by the families  $\{L_{q,0}\}_{q \in Q}$ and $\{L_{q}\}_{q \in Q} $ need not be maximal in general. For explicit (counter)examples, see Section~\ref{S9.2}.
\end{remark}

%%%%%%%%%%%%%%%%%%%%%%%%%%%%%%%%%Section5

\section{Proof of Main Results}%Theorem \ref{gaic}}
\la{S5}
The first three parts of our main Theorem~\ref{gaic} follow from Theorem~\ref{IT} and Proposition~\ref{commprop}. To apply these general results
we need to verify their assumptions for locus configurations. This is done in Section~\ref{S5.1}. 
The last part of Theorem~\ref{gaic} is proven separately as  
Proposition~\ref{mad} in Section~\ref{S5.2}. In Section~\ref{S5.3}, we apply 
the results of Section~\ref{S4.1} (in particular, Proposition~\ref{fatrefl}) to ideals of 
Cherednik algebras  associated to locus configurations. 

\subsection{The space $\U $ and the ideal $ \M $ associated to $ \A$}
\la{S5.1}
Given a locus configuration $\A$ of type $W$, consider the polynomial $ \delta_k \in\c[V]^W$ defined by\footnote{The polynomial \eqref{del} should not be confused with the discriminant of the Coxeter group $W$, i.e. $\, \prod_{\alpha\in R_+}(\alpha, x) \,$, which is  also 
denoted frequently by $\delta$ in the literature.}
\begin{equation}
\la{del}
\delta_k :=\prod_{\alpha\in\A_+\setminus R}(\alpha,x)^{k_\alpha}\,.
\end{equation}
The fact that $\delta_k $ is $W$-invariant follows from the $W$-invariance of $\mathcal A$ and $k_\alpha$: indeed, we must have $\delta_k(s_\alpha x)=\pm \delta_k(x)$ for any $\alpha\in R$, but $\delta_k(s_\alpha x)=-\delta_k(x)$ is impossible since $\delta_k$ does not vanish along $(\alpha, x)=0$ for $\alpha\in R$.
The set $ S = \{1,\,\delta_k,\, \delta_k^2,\,\ldots\}$ is a two-sided Ore subset in the Cherednik algebra $ H_k $, and we write $H_k[\delta_k^{-1}]$ and $B_k[\delta_k^{-1}]$ for $H_k$ and $B_k$ localised at $S$. By \eqref{HC1}, $\,B_k\subset \D(\vreg)^W$, thus the algebras $ B:=B_k $, $\,R:=\c[V]^W$ and the set
$S \subset R $ satisfy the assumptions of Section \ref{pro}. Note that the quotient filed $\mathbb{K}$ of $R$ is $\c(V)^W$, hence $\D(\mathbb K)$ is the ring of $W$-invariant differential operators on $V$ with rational coefficients.

Let $L_0=L_W$ and $L=L_\A$ denote the Calogero--Moser operators \eqref{cmo} and \eqref{gcmu}, respectively; clearly, $L_0\in B_k$ and $L\in B_k[\delta_k^{-1}]$. The operator $L_0$ acts on $B_k$ locally ad-nilpotently (see \cite[Lemma 3.3(v)]{BEG}), so by Lemma \ref{filtl} we can associate to it a degree function on $B_k$ and $B_k[\delta_k^{-1}]$. Moreover, by Corollary 4.9 of {\it loc. cit.}, for any $f\in\c[V]^W$, $\deg_{L_0}f$ equals the usual homogeneous degree of $f$. This means that the number $N_{q,0}$ in \eqref{Loq} equals the degree of $q\in\c[V]^W$. In fact, by comparing the leading terms, one has the following formula, see \cite[(6.5)]{BEG}:
\begin{equation}\label{loqc} 
L_{q,0}:=\Res(\e T_q\e)=\frac{1}{2^NN!}\ad_{L_0}^Nq\,,\qquad  N=\deg q\,.
\end{equation}
Now, since $P=L-L_0$ is a rational $W$-invariant function of degree $-2$, we conclude that $\deg_{L_0}(L-L_0)=-2$ . This verifies the condition (2) of Theorem \ref{sheq}. 

Next, we have $B_k(\c[\vreg]^W)\subset \c[\vreg]^W$. Moreover, any $a\in B_k[\delta_k^{-1}]$ that preserves $\c[\vreg]^W$ must be regular away from the reflection hyperplanes of $W$, hence, $a$ must necessarily lie in $B_k$. This proves that  
\begin{equation}\label{dak00}
B_k=\{a\in B_k[\delta_k^{-1}]\,\mid\, a(\c[\vreg]^W)\subset \c[\vreg]^W\}\,. 
\end{equation}
This means that $U_0:=\c[\vreg]^W$ satisfies the assumptions of Theorem \ref{IT}.

Finally, we define the most important ingredient: the subspace $ \U = \U_{\A} $ attached to the 
operator $ L_{\A}$. We let $ \U_{\A} $ to be the subspace of $ \delta_k^{-1}\c[\vreg]^W $
that consists of functions $f$ satisfying
\begin{equation}
\label{loc11}
f(s_\alpha x) -(-1)^{k_\alpha}f(x) \ \text{is divisible by}\ (\alpha,x)^{k_\alpha}\quad\forall\alpha\in\A_+\setminus R\,.
\end{equation}
It is clear from this definition that
\begin{equation}
\la{uq}
\delta_k\,\c[\vreg]^W\subseteq \U_{\A}\subseteq \delta_k^{-1}\c[\vreg]^W\,,\qquad
\qaw\, \U_{\A}\subset \U_{\A}\,.
\end{equation}
Note that the above properties of $ \U_{\A} $ are generic: they hold without assuming the locus relations \eqref{loc}.
The next lemma establishes the two crucial properties of $\U_\A$ that do depend on  \eqref{loc}.

\begin{lemma}[{\it cf.} \cite{C98, CEO}] \la{inu} The space $\U_\A$ is invariant under the action of $L_\A$. Moreover, $\U_\A$ is maximal among all subspaces $\,U$ with the properties that $\,U\subseteq \delta_k^{-r}\c[\vreg]^W$ for some $r>0$ and $\,L_\A(U)\subseteq U\,$.
\end{lemma}

\begin{proof} 
The first claim in the case $W=\{e\}$ goes back to \cite{C98} while the second is a slight reformulation of \cite[Proposition 5.1]{CEO}. The same arguments work for the general $W$. Namely, one works `locally' with Laurent expansions in direction $\alpha$, for each $\alpha\in\A_+\setminus R$. Functions $f\in\U_\A$ are then characterised precisely by the property that their Laurent expansions contain terms $(\alpha, x)^j$ with $j\in \{-k_\alpha+2\Z_{\ge 0}\}\cup\{k_\alpha+1+2\Z_{\ge 0}\}$ only. On the other hand, the locus relations \eqref{loc} mean that in a similar Laurent expansion of $u$ there are no terms of degree $1, 3, \dots, 2k_\alpha -1$. The invariance of $\U_\A$ under $L_\A$ immediately follows from that. Moreover, if $f\notin \U_\A$, then one can repeatedly apply $L_\A$ to $f$ and obtain a function with a pole of an arbitrarily large order. This, in its turn, would violate the condition $\,U\subseteq \delta_k^{-r}\c[\vreg]^W$, thus proving that $\U_\A$ is maximal. See the proof of \cite[Proposition 5.1]{CEO} for the details.
\end{proof}
\begin{remark}
\la{inu0}
A result similar to Lemma~\ref{inu} hold for $L_0 = L_W$: namely, $\c[\vreg]^W$ is maximal among all subspaces satisfying $U\subset\delta_k^{-r}\c[\vreg]^W$ for some $r>0$ and $L_0(U)\subset U$. The proof is similar (formally, it corresponds to setting $k_\alpha=0$ in the above arguments). 
\end{remark}

With the above choices, definitions \eqref{MB}, \eqref{D} become 
\begin{eqnarray}
\label{mak}
\M_{\A}&=&\{a\in B_k[\delta_k^{-1}]\,\mid \, a(\c[\vreg]^W)\subset \U_{\A}\}\,,
%\\
%\label{mak0}
%\M_{\A}^*&=&\{a\in B_k[\delta^{-1}]\,\mid \, a(\U_\A)\subset \c[\vreg]^W\}\,,
\\
\label{dak}
\D_{\A}&=&\{a\in B_k[\delta_k^{-1}]\,\mid\, a(\U_{\A})\subset \U_{\A}\}\,. 
\end{eqnarray}
Note $\M_\A$ is a right $B_k$-module, and $\D_\A$ is a ring; we can also view $\M_\A$ as a $\D_\A$-$B_k$-bimodule. 
It is clear that $L\in \D_\A$ and, by \eqref{uq}, 
\begin{equation}
\label{db}
\delta_k \in\M_{\A}\,,\quad \delta_k B_k \subset \M_{\A} \subset \delta_k^{-1}B_k\,,%\qquad B_k\delta \subset \M_\A^* \subset B_k\delta^{-1}\,,
\qquad \delta_k\D_\A\delta_k \subset B_k\,.
\end{equation} 
By (the proof of) Proposition~\ref{fatrefl}, the operators $\ad_L$ and $\ad_{L, L_0}$ act locally nilpotently on $\D_\A$ and $\M_\A$, respectively. In fact, the $\M_\A$ and $\D_\A$ can be characterised 
intrinsically as the maximal subspaces of $B_k[\delta^{-1}]$ on which these operators act locally 
nilpotently (see Remark \ref{intri}).
%
%\begin{eqnarray*}
%\M_\A &=& \{a\in B_k[\delta^{-1}]\,:\, \ad_{L, L_0}^N(a)=0\ \text{for some $ N\ge 0$}\}\,,\\
%\D_\A &=& \{a\in B_k[\delta^{-1}]\,:\, \ad_L^N(a)=0\ \text{for some $N\ge 0$}\}\,.
%\end{eqnarray*}
%

\vspace*{2ex}

Summing up, given a locus configuration of type $W$, the following data
\begin{gather*}\label{ing}
R=\c[V]^W\,,\  S=\{1,\delta_k, \delta_k^2, \dots\}\,,\  B=B_k\,,\  A=B_k[\delta_k^{-1}]\,,\\
L_0=L_W\,,\ L=L_\A\,,\ U_0=\c[\vreg]^W\,,\  U=U_\A\,,\ \M = \M_\A\,,\ \D = \D_\A
\end{gather*}
satisfy the assumptions of Theorem~\ref{IT} and Proposition~\ref{fatrefl}; hence, all results 
proved in Section \ref{pro} can be applied to locus configurations.

\subsection{Proof of Theorem \ref{gaic}}
\la{S5.2}
Parts $(1)$ and $(3)$ are immediate from Theorem \ref{IT}. Note that
the shift operator $S\in\M_\A$ can be chosen in the form 
\begin{equation}\label{shiftopd}
    S=\frac{1}{2^NN!}\,\ad_{L,L_0}^N(\delta_k)\,.
\end{equation}
where $ N = \deg(\delta_k) $. Indeed, by an elementary calculation ({\it cf.} \cite{B98}),
\begin{equation*}
S=\prod_{\alpha\in\A_+\setminus R} (\alpha, \partial)^{k_\alpha}+\ldots\,,
\end{equation*}
where ``$ \ldots $'' denote the lower order terms. Hence $S\ne 0$. On the other hand, a simple argument based on the nilpotency of $\ad_{L, L_0}$ on $\M_\A$ and the $x$-filtration on $B_k[\delta_k^{-1}]$, shows that $\ad_{L, L_0}(S)=0$ (see \cite{C98}), which means that $LS=SL_0$.

Part (2) is the result of Proposition \ref{commprop} combined with \eqref{loqc}. Indeed, it follows that the commuting differential operators $L_q$, $q\in \qaw$, can be given by 
\begin{equation}\label{lq1}
L_q=\frac{1}{2^{r}r!}\ad_{L}^{r}q\,,%\qquad q\in\qaw\,,
\qquad r=\deg q\,.
\end{equation}
Furthermore, since the ring $\qaw$ is graded, we may replace $\grd{\qaw}$ with $\qaw$ and get the required algebra embedding $\theta\,:\ \qaw\into \D(V\setminus H_\A)^W$.

In remains to prove part $(4)$: namely, that $\theta(\qaw)$ is a maximal commutative subalgebra in $\D(V\!\setminus\! H_\A)^W$. We will prove a slightly stronger statement. Recall the ring $\D(\mathbb K)$ of differential operators on the field $\mathbb K=\c(V)^W$.
%Let $\mathbb M$ denote the field of $W$-invariant meromorphic functions on $V$, and write $\D(\mathbb M)$ for the algebra of $W$-invariant differential operators on $V$ with meromorphic coefficients. Since $\D(\mathbb K)\subset \D(\mathbb M)$,  we can view $\theta(\qaw)$ as a subalgebra of $\D(\mathbb M)$.
\begin{prop} 
\label{mad} 
$\theta(\qaw)$ is a maximal commutative subalgebra of $\D(\mathbb K)$.
\end{prop}
To prepare the proof, introduce $U\subset \mathbb K$ as the subspace of rational functions $f$ which (1) are allowed a pole of order at most $k_\alpha$ along each of the hyperplanes $H_\alpha=\Ker(1-s_\alpha)$ with $\alpha\in\A_+\setminus R$, and (2) satisfy the conditions \eqref{loc11}. The difference with the definition of the space $U_\A$ above is that the $W$-invariance of $f$ is not assumed and $f$ is allowed arbitrary singularities away from $H_\A$. Still, the property $L_\A(U)\subset U$ from Lemma \ref{inu} remains valid, because it was based on local analysis. 

\begin{lemma}\label{invu} If $a\in\D(\mathbb K)$ commutes with $L=L_\A$ then $a(U)\subset U$. 
\end{lemma}
In the case $R=\varnothing $, this is \cite[Proposition 5.1]{CEO}, and the same proof works in general. \qed

%
%\begin{proof} This is the same proof as in \cite[Proposition 5.1]{CEO}, but we include it for the reader's convenience. 
%Consider and fix such $a$ and let $U_1:=a(U)\subset \mathbb K$.  Since order of poles of $f\in U$ along $H_\alpha$ is bounded by $k_\alpha$, functions in $U_1$ may only have poles along $H_\alpha$ of bounded order. On the other hand, if $g=a(f)$ with $f\in U$, then $L^r(g)=L^ra(f)=a(L^r(f))\in U_1$ for any $r>0$. 
%If $U_1\not\subset U$, there is $g=a(f)\notin U$, so then the conditions describing $U$ must be violated for $g$ for at least one of $\alpha\in\A_+\setminus R$. Arguing as in the proof of Lemma \ref{inu}, the function $L^r(g)$ will then have a pole of arbitrarily high order for sufficiently large $r$. But this contradicts $L^r(g)\in U_1$. Hence, $U_1\subset U$, as needed.
%\end{proof}

\begin{proof}[Proof of Proposition \ref{mad}]
%For a differential operator $a\in\D(\mathbb K)$ we write $\grd(a)$ for its leading symbol with respect to the differential filtration. 
%Let us denote this subalgebra $\mathcal C$, and suppose $C\in\D[\vreg]^W[\delta_k^{-1}]$ commutes with each element in $\mathcal C$; we want to show that $C=L_q$ for some $q\in\qaw$.  

Suppose there is a differential operator $a$ of order $r$ which commutes with all $L_q\in\theta(\qaw)$ but such that $a\notin \theta(\qaw)$. Without loss of generality, we may assume that $a$ has the least order among all such operators. 

By part (3), there are $n=\dim V$ algebraically independent operators $L_1=L_{q_1},\dots, L_n=L_{q_n}$ with homogeneous $q_i\in\qaw$. Denote by $\grd (L_i)=q_i(\xi)$ their principal symbols with respect to the differential filtration. %; we have $\grd (L_i)=q_i(\xi)$. 
Let $a_0=\grd (a)$ be the principal symbol of $a$. As $a$ commutes with each $L_i$, its symbol $a_0(x, \xi)$ Poisson commutes with each of $q_i(\xi)$. Since $q_i$'s are $n$ algebraically independent Poisson commuting elements, $a_0$ must be a function of $\xi$ only. Therefore, $a$ has a constant principal symbol, i.e. $a=q(\partial)+\ldots$ for some homogeneous $W$-invariant polynomial $q(\xi)$.  
  
Let ${x^2}=(x,x)$ denote the quadratic polynomial corresponding to the $W$-invariant inner product on $V$. Clearly, ${x^2}\in\qaw$ and ${x^2}\,U\subset U$.  By Lemma \ref{invu}, $a(U)\subset U$ as well. By a straightforward calculation (cf. \eqref{loqc}),
\begin{equation*}
\ad_{{x^2}}^r(a)=2^rr!q(x)\,,\quad r=\deg q\,.
\end{equation*}
Hence, $q(x)U\subset U$ and so $q\in\qaw$. In that case $a':=a-L_q$ gives a lower order operator commuting with $\theta(\qaw)$, leading to a contradiction.
\end{proof}

\begin{remark} The above result and its proof remain valid if one replaces $\D(\mathbb K)$ with the ring of $W$-invariant differential operators with {\it meromorphic} coefficients. Moreover, if we assume additionally that $k_\alpha\notin\Z$ for all $\alpha\in R$, then we can further replace $\D(\mathbb K)$ with the ring of {\it all} differential operators with meromorphic coefficients. Proof goes in the same way until we get that $a=q(\partial)+\dots$ for some $q\in\c[V^*]$. Now we use that $[a, L_W]=0$, so by the result of \cite{T} the principal symbol $q(\xi)$ of $a$ must be $W$-invariant. The rest of the proof is unchanged. 

\end{remark}

\begin{remark}\la{rk}
For the commutative ring $\{L_q\,|\, q\in \qaw\}$ we can consider the eigenvalue problem 
\begin{equation}\la{ep}
L_q\psi=q(\lambda)\psi\,\quad  \forall\,q\in \qaw\,,
\end{equation}
where $\psi=\psi(x, \lambda)$ is a function of $x$ and the spectral variable $\lambda\in V$. The dimension of the solution space to \eqref{ep} for generic $\lambda$ is usually referred to as the \emph{rank} of the commutative ring (cf. \cite{BrEtGa}). By the arguments similar to those used in \cite[Section 3]{C08}, one shows that the solution space to \eqref{ep} has dimension equal to $|W|$, hence the commutative ring $\theta(\qaw)$ has rank $|W|$. When the group $W$ is trivial, $\theta(\qaw)$ has rank one, {\it cf.} \cite[Theorem 3.11]{C08}. 
\end{remark}

\subsection{Ideals of rational Cherednik algebras}%(THIS SECTION HASN'T BEEN "CLEANED" YET)
\la{S5.3}
In the case of Coxeter configurations, when $\A\subset \c^n$ is the (complexified) root system of a Coxeter group and all $k_\alpha$ are integers, the algebra of quasi-invairants $Q_\A$ is known to be Cohen--Macaulay (see \cite{FV}, \cite{EG1} and \cite{BEG}). In \cite{BEG}, this remarkable property of $Q_\A$ was deduced from the fact that the ring $\D_\A$ is Morita equivalent to the Weyl algebra $ \D(V) $; more precisely, Theorem~9.5 of \cite{BEG} states that  $ \D_{\A} \cong \End_{\D(V)}(\M_{\A}) $, $\,\M_\A$ being a projective ideal of $ \D(V) $. It is natural to ask if this last result holds for generalised locus configurations. Proposition~\ref{PPP} below shows that this is 'almost' the case.

%See also \cite{BCES, ER} where the rings $Q_{\A}$ are shown to be Cohen--Macaulay for some non-Coxeter cases. 
Recall that, for  generalised locus configurations, we defined the ring $ \D_{\A} $ and the $\D_\A$-$B_k$-bimodule $ \M_{\A} $ (see \eqref{mak} and \eqref{dak}). The definition of $ \M_{\A} $ shows that it is isomorphic to a right ideal $ \M_x \subseteq B_k $ --- specifically, $\, \M_x = \delta_k \M_{\A}\, $ (see \eqref{db}) --- 
with the property $ \M_x \,\cap\,\c[V]^W\ne \{0\}$. Extending the terminology of \cite{BGK} and \cite{BW}, we call such ideals of $B_k$ {\it fat}. Besides $ \c[V]^W $, the algebra $B_k$ contains another distinguished (maximal) commutative subalgebra,  $ \c[V^*]^W $, consisting of operators
of the form $\,\Res(\e T_q \e)\,$ ({\it cf.} Theorem~\ref{ci}). Following \cite{BCM}, we will say that a fat ideal $ \M $ of $B_k$  is \emph{very fat} if, in addition, $\M$ is isomorphic to an ideal $\M_y\subseteq B_k$ with the property $\M_y\cap \c[V^*]^W\ne\{0\}$. 
Now, recall that every reflexive module is automatically projective, but the converse is not true. As observed in \cite{BCM}, the property of a reflexive module $\M $ to be {\it very fat} provides a `closer' approximation to the projectivity of $\M$. In fact, a very fat reflexive module $\M$ behaves like a projective module with respect to localization: the localized modules $ \M_x^{\rm loc} $ and $ \M_y^{\rm loc} $ obtained from $ \M $ by inverting the nonzero polynomials in $ \c[V]^W $ and $ \c[V^*]^W $ are both free modules\footnote{If $ \dim(V) = 1 $, every fat ideal is very fat, and moreover,
every very fat one is automatically projective. Unfortunately, this is not in general true when $ \dim(V) > 1 $: there exist fat ideals which are not very fat, and not every very fat ideal is projective (see \cite{BCM}).}. 

%After these remarks, we state our next proposition, which can be viewed as a (partial) generalization of \cite[Theorem~9.7]{BEG}.
%

\begin{prop}
\la{PPP}
For any generalised locus configuration $\A$, the module $\M_{\A}$  
is a very fat reflexive ideal of $B_k $ with $\, \End_{B_k}(\M_{\A}) \cong \D_{\A} \,$.
\end{prop}

\begin{proof}
The facts that $\M_\A$ is reflexive and $ \D_{\A} \cong \End_{B_k}(\M_{\A}) $
follow from Proposition \ref{fatrefl}. We need only to prove that $\M_\A$ is very fat.
Following \eqref{nn}, consider
\begin{equation*}
%\M_{\A}^*:
\NN=\{a\in B_k[\delta_k^{-1}]\,\mid \, a(U_\A)\subset \c[\vreg]^W\}\,.
\end{equation*}
This is a $B_k$\,-\,$\D_{\A}$ bimodule. We can see directly from \eqref{dak00} that $\NN \M\subset B_k$.  In fact, by Remark \ref{intri}, $\NN$ is isomorphic as a left $B_k$-module to $\M^*$, the dual of $\M$. 
Take 
\begin{equation}\label{s8}
S^*=\frac{1}{2^NN!}\ad_{L_0, L}^N(\delta_k)\,,\qquad N=\deg\delta_k\,.
\end{equation} 
By \eqref{uq}, $\delta_k$ belongs to $\NN$, hence $S^*\in\NN$. We also have $S^*L=L_0S^*$, which is proved by the same arguments as part (1) of Theorem \ref{IT}. (Alternatively, this is obtained from $LS=SL_0$ after taking  formal adjoints.)  
Now, let $\M_y:=S^*\M_\A$. Since $S^*\in \NN$, we get 
\begin{equation*}
\M_y\subset \NN\M_\A\subset B_k\,.  
\end{equation*}
On the other hand, the relations $LS=SL_0$ and $S^*L=L_0S^*$ imply that $\ad_{L_0,L}^{N+1}(\delta_k)=\ad_{L,L_0}^{N+1}(\delta_k)=0$.
Using this fact and \eqref{coad}, we obtain that
\begin{equation*}
    S^*S=\frac{1}{2^{2N}(2N)!}\ad_{L_0}^{2N}(\delta_k^2).
\end{equation*}
This is one of the operators appearing in \eqref{loqc}, namely, $L_{\delta_k^2, 0}=\Res(T_{\delta_k^2})$.
We conclude that $\M_y\cap \c[V^*]^W\ne\{0\}$, so $\M_\A$ is very fat.   
\end{proof}
\begin{remark} 
As a special case, the above results apply to locus configurations $\A\subset\c^n$ with $R=\varnothing $. The Cherednik algebra $B_k$ in that case is simply the $n$th Weyl algebra $A_n$. Thus, such locus configurations produce very fat reflexive ideals $\M_\A$ of $A_n$. For a detailed study of fat/very fat reflexive ideals of $A_n$
we refer to \cite{BCM}.
\end{remark}

%\subsection{Non-twisted configurations} 
%\la{S5.4}
For some configurations (including those that appear in \cite{SV}) there is another natural way to associate ideals of Cherednik algebra. Namely, suppose that for each $\alpha\in\A_+\setminus R$ the following conditions are satisfied:
\begin{equation}\la{mi2} 
\partial_\alpha^{2j-1}
%\left(\frac{\partial}{\partial\alpha}\right)^{2j-1}
%\frac{\delta_\A}{(\alpha,x)^{k_\alpha}}
\left(\prod_{\beta\in \A_+\setminus\{\alpha\}}(\beta, x)^{k_\beta}\right)=0\quad\text{for}\ (\alpha, x)=0\quad\text{and}\ j=1,\dots, k_\alpha\,.
\end{equation}
Here $\partial_\alpha=(\alpha, \partial)$ denotes the directional derivative. 
In the case $k_\alpha=1$ \eqref{mi2} reduce to a single condition, 
\begin{equation}\la{mi1} 
\sum_{\beta\in \A_+\setminus \{\alpha\}} \frac{k_\beta (\alpha, \beta)} {(\beta ,x)}=0\quad\text{for}\ (\alpha, x)=0\,.
\end{equation}
If $\A$ satisfies the identity \eqref{mi}, then \eqref{mi1} follows by taking the residue along $(\alpha, x)=0$. For $k_\alpha>1$, however, \eqref{mi2} are stronger conditions than \eqref{mi}. Let us call $\A$ {\it non-twisted} if it satisfies the conditions \eqref{mi2} for all $\alpha\in \A_+\setminus R$, and {\it twisted} otherwise. When $\A$ is non-twisted, \eqref{mi} is always true as a consequence of \eqref{mi1}.   
For non-twisted $\A$, one may work with the following ``radial part'' versions of the Calogero--Moser operators: 
\begin{equation*}
\widetilde L_W=\Delta-\sum_{\alpha\in R_+}\frac{2k_\alpha}{(\alpha, x)}\partial_\alpha
%(\alpha, \partial)
\,,\qquad \widetilde L_\A =\Delta-\sum_{\alpha\in\A_+}\frac{2k_\alpha}{(\alpha, x)}
%(\alpha, \partial)
\partial_\alpha\,.
\end{equation*}
We have
$\widetilde L_W =\delta_0 L_W \delta_0^{-1}$
and $\widetilde L_\A =\delta_\A L_\A \delta_\A^{-1}$,
where $\delta_0=\prod_{\alpha\in R_+}(\alpha,x)^{k_\alpha}$ and $\delta_\A$ is as in \eqref{dsf}. %To verify these, one uses that $L_W(\delta_0^{-1})=L_\A(\delta_\A^{-1})=0$.  

Let us modify the Cherednik algebra accordingly, by
setting $\widetilde H_k$ to be the subalgebra of $\D W$ generated by $\c W$, $\c[V]$ and all $\widetilde T_\xi$, where the Dunkl operators are given in the ``standard gauge'':
\begin{equation*}\la{stdu}
\widetilde T_\xi :=
\partial_\xi-\sum_{\alpha\in R_+}
\frac{(\alpha,\xi)}{(\alpha, x)}k_\alpha (1-s_\alpha)\ , \quad \xi \in V\,.
\end{equation*}
As before, the spherical subalgebra is $\widetilde B_k:=\Res (\e \widetilde H_k \e)$ (in fact, $\widetilde B_k=\delta_0B_k\delta_0^{-1}$). We also need to modify the space $U_\A$, replacing it with $U=Q_{\rm{reg}}$ which is defined as the space of all $q\in\c[\vreg]^W$ that satisfy the quasi-invariance conditions \eqref{qi1}. (This is the same ring $\qreg$ that appears in the Introduction.) Due to \eqref{mi2} and the relation $\widetilde L_\A=\delta_\A L_\A \delta_\A^{-1}$, we still have the crucial property $\widetilde L_\A(\qreg)\subseteq \qreg$. Furthermore, the choice $L_0=\widetilde L_W$, $L=\widetilde L_\A$, $B=\widetilde B_k$, $A=\widetilde B_k[\delta_k^{-1}]$, $U_0=\c[\vreg]^W$, $U=\qreg$ satisfies all the assumptions of Section \ref{pro}. In particular, we may define  
\begin{align*}
    \wM_{\A}&=\{a\in \wB_k[\delta_k^{-1}]\,\mid \, a(\c[\vreg]^W)\subset \qreg\}\,,\\
    \wD_{\A}&=\{a\in \wB_k[\delta_k^{-1}]\,\mid\, a(\qreg)\subset \qreg\}\,. 
\end{align*}
Note that $\qreg\subset U_0=\c[\vreg]^W$, so $a(\c[\vreg]^W)\subset \qreg$ implies $a\in \wB_k$. Thus, $\wM_\A$ is an honest ideal of $\wB_k$. We then have the following analogue of Proposition \ref{PPP}, with the same proof.  

\begin{prop}
\la{PPP1}
For a non-twisted locus configuration $\A$, $\wM_{\A}$  
is a very fat reflexive ideal of $\wB_k $ with $\, \End_{\wB_k}(\wM_{\A}) \cong \wD_{\A} \,$.
\end{prop}

%Let $L=\widetilde L_\A$ and $L_0=\widetilde L_W$. 
%\begin{theorem}\label{lnilt}
%$(1)$ The action of $\ad_{L, L_0}$ on $\wM_\A$ is locally nilpotent.

%$(2)$ With respect to the differential filtration, we have $\grd (\wM_\A)\subset \delta\, \c[V\times V^*]^W$ and 
%$\grd (\wD_{\A})\subset \c[V\times V^*]^W$.

%$(2)$ 
%Define 
%$S=\frac{1}{2^NN!}\ad_{L, L_0}^N(\delta^2)$, $N=\deg \delta$.
%Then $LS=SL_{0}$. 
%Moreover, 
%$S$ belongs to $\wM_\A$ (hence, $S\in \wB_k$) and has 
%$\delta(x) \delta(\partial)$ as its leading term.

%$(3)$ For any $q\in \qaw$ with $\deg q=r$, define 
%$L_q=\frac{1}{2^rr!}\ad_{L}^r(q)$ and $L_{q,0}=\frac{1}{2^rr!}\ad_{L_0}^r(q)$.
%Then the operators $L_q$, $q\in\qaw$ pairwise commute, and we have $L_qS=SL_{q,0}$.   
%$L_q(\qaw)\subset \qaw$, and the map $q\mapsto L_q$ defines an algebra embedding of $\qaw\into %\wD_\A$.

%$(4)$ The right ideal $\wM_\A\subset \wB_k$ is very fat reflexive.
%\end{theorem} 
Assuming the algebra $\qreg$ is finitely generated, it can be viewed as the algebra of functions on a (singular) affine variety $X_{\rm reg}=\Spec\, \qreg$, and so the elements of $\wD_{\A}$ are regular differential operators on $X_{\rm reg}$. 
%, and we have two large commutative subalgebras in $D(X_k)$, namely, $\qaw$ and the ``dual'' subalgebra $\{L_q,\ q\in \qaw\}$. 
For comparison, if $\A$ is twisted, 
%we define the ring $Q_{\rm reg}\subset \c[\vreg]^W$, using the same algebraic conditions as for $Q=Q_{\A,W}$, see (4.1). Then 
we have a projective, rank-one $\qreg$-module $\U_{\A}$ (i.e. a line bundle over $X_{\rm reg}$), and so the elements of $\D_\A$ can be viewed as {\it twisted} differential operators.  

\begin{remark}
Theorem \ref{gaic} trivially extends to $L=\widetilde L_\A$, $L_0=\widetilde L_W$. Namely, the maximal commutative ring $\theta(\qaw)$ gets replaced with $\delta_\A\theta(\qaw)\delta_\A^{-1}$, and the shift operator takes the form $S=\frac{1}{2^N N!}\ad_{L, L_0}^N(\delta_k^2)$, where $N=\deg\delta_k$. 
\end{remark}

\section{Examples of locus configurations} 
\label{exloc}

Finding explicitly all generalised locus configurations is an open problem. In this section we describe all examples currently known in dimension $>2$. The two-dimensional configurations will be discussed in Section~\ref{twodim} below.

%In many examples, one has $k_\alpha=1$ for $\alpha\in\A\setminus R$ in which case %\eqref{loc1} reduces to    
%\begin{equation}
%\label{loc12}
%\sum_{\beta\in\A_+\setminus\{\alpha\}}\frac{k_\beta(k_\beta+1)(\beta, %\beta)(\alpha,\beta)}{(\beta,x)^{3}}=0\quad\text{for $(\alpha,x)=0$.}
%\end{equation}

Recall that a locus configuration $\A$ of type $W$ is obtained by adding a $W$-invariant set of vectors to the root system $R$ of $W$. These vectors and their multiplicities $k_\alpha\in\Z_+$ must satisfy the algebraic relations \eqref{loc1}. Although the equations \eqref{loc1} are quite explicit, finding (and classifying) their solutions is a difficult problem. One approach to this problem relies on the following observation that allows one to build (nontrivial) 
locus configurations in arbitrary dimension from the known locus configurations in dimension $2$. 
\begin{prop}\label{2dim} $\A$ is a locus configuration if and only if every two-dimensional sub-configuration $\A'\subset\A$ is a locus configuration.  
\end{prop}
Here, by a two-dimensional sub-configuration we mean the intersection $\A'=\A\,\cap\, V'$, where $V'$ is a two-dimensional subspace of $V$. 
%(and with the multiplicities obtained by restriction). 
The above proposition is a generalisation of Theorem~4.1 of \cite{CFV} that treats the case $W=\{e\}$; the same argument works for an arbitrary $W$. 

The following examples of two-dimensional locus configurations will serve as `building blocks' for higher dimensional configurations that we will describe in this section. In each of these examples, we write the vectors in $\A_+$ relative to some orthonormal basis $\{e_1, e_2, e_3\}$ in $ \c^3 $. We also indicate the subset of roots $R_+ \subset \A_+ $ and the corresponding Coxeter group $W$.
All non-Coxeter vectors in $ \A_+$ have multiplicity $1$, so that the relations \eqref{loc1} need to be checked only for $j=1$. This is an easy exercise left to the reader. 

\begin{example}\label{2dex}

\begin{enumerate}
    \item[(1)] 
$\A_+=\{e_1-e_2, ae_1-be_3, ae_2-be_3\}$ with multiplicities $k=(m,1,1)$ , where $b^2=ma^2$ or ${b^2}=(-1-m){a^2}$. In this  example, $R_+=\{e_1-e_2\}$ and $W=\Z_2$.
\item[(2)] 
$\A_+=\{e_1, e_2, ae_1-be_2, ae_1+be_2\}$ with multiplicities $k=(l,m,1,1)$, where $(2l+1){a^2}=\pm(2m+1){b^2}$. In this example, $R_+=\{e_1, e_2\}$ and $W=\Z_2\times \Z_2$.
\item[(3)] 
$\A_+=\{ae_1-be_2, be_2-ce_3, ae_1-ce_3\}$ with multiplicities $k=(1,1,1)$, where  $a^2+b^2+c^2=0$. In this example, $R_+=\varnothing$ and $W=\{e\}$.

\end{enumerate}

\end{example}

\begin{remark} The configurations (1) and (3) are contained in a two-dimensional subspace of $\c^3$ and so we may think of them as two-dimensional.  
When $l,m$ are integers, all of the above configurations can be viewed as locus configurations of type $W=\{e\}$, and as such they can be found in Section 4.2 of \cite{CFV}.
\end{remark}

\subsection{Coxeter configurations}\la{cc}

The simplest examples can be obtained by considering a pair $W\subset W'$ of finite Coxeter groups acting on $V$. Let $R\subset R'$ be the corresponding root systems with a $W'$-invariant {\it integral} multiplicities $k: R'\to \Z_+$. Then we can view $\A=R'$ as a generalised locus configuration of type $W$. Indeed, the locus relations \eqref{loc} or \eqref{loc1} trivially follow from the $W'$-invariance. In the special case when $R$ coincides with one of the $W'$-orbits of roots in $R'$, we may allow non-integer multiplicities $k_\alpha$ for $\alpha\in R$. 

\subsection{Examples related to Lie superalgebras} \la{sa}
These examples were discovered in \cite{SV}; there are two infinite series and three exceptional cases. 

%\medskip
%\noindent {\bf 1.} $A(n, m)$ {\it configuration}

\subsubsection{$A(n, m)$ configuration}
Here $V=\c^{n+m}$ with the standard scalar product, and the group $W=S_n\times S_m$ acts by permuting the first $n$ and the last $m$ of the coordinates. We set 
\begin{equation}\label{i1i2}
I_1=\{1, \dots, n\}\,,\quad I_2=\{n+1, \dots, n+m\}\,. 
\end{equation}
The configuration depends on one parameter $k\ne 0$. It consists of the vectors $\alpha= e_i-e_j$, $i, j\in I_1$, $i\ne j$ with $k_\alpha=k$, the vectors $\alpha=e_i-e_j$, $i, j\in I_2$, $i\ne j$ with $k_\alpha=k^{-1}$, and the vectors $\pm(e_i-\sqrt{k} e_j)$, $i\in I_1$, $j\in I_2$ with $k_\alpha=1$.

Its $W$-invariance is obvious. To check the locus conditions, we use Proposition \ref{2dim}.
It is then easy to see that any two-dimensional $\A_0\subset \A$ either lies entirely in $R$ (so is Coxeter), is of type $\A_0=\{\pm\alpha, \pm\beta\}$ with orthogonal $\alpha, \beta$, or is equivalent to the configuration (1) from Example \ref{2dex} (with $m=k$ or $m=k^{-1}$).

\begin{remark}\label{dualk}
Since $k_\alpha$ enters \eqref{gcm} as a combination $k_\alpha(k_\alpha+1)$, one can always replace $k_\alpha\mapsto -1-k_\alpha$ for $\alpha\in R$ in a locus configuration. For example, in $A(n,m)$ one can take $k_\alpha=-1-k^{-1}$ for $\alpha=e_i-e_j$ with $i,j\in I_2$.    
\end{remark}

%\medskip

%\noindent{\bf 2.} $BC(n, m)$ {\it configuration} 

\subsubsection{$BC(n, m)$ configuration}
We keep the notation of the previous case, so $V=\c^{n+m}$ with the standard Euclidean product. The configuration depends on parameters $k\ne 0$ and $l_1,l_2$ related by
\begin{equation*}
2l_1+1=k(2l_2+1)\,.
\end{equation*}
It consists of the vectors  $\alpha= \pm e_i\pm e_j$, $i, j\in I_1$, $i\ne j$ with $k_\alpha=k$, the vectors $\alpha=\pm e_i\pm e_j$, $i, j\in I_2$, $i\ne j$ with $k_\alpha=k^{-1}$, the vectors $\alpha=\pm e_i$, $i\in I_1$ with $k_\alpha=l_1$, the vectors $\alpha=\pm e_i$, $i\in I_2$ with $k_\alpha=l_2$, and 
the vectors $\pm e_i\pm \sqrt{k} e_j$, $i\in I_1$, $j\in I_2$ with $k_\alpha=1$.

Let us write down the corresponding Calogero--Moser operator. Using Cartesian coordinates $x_1, \dots, x_n$, $y_1, \dots, y_m$ on $V$, we obtain:
\begin{align*}
L_{BC(n,m)}=\Delta&-2k(k+1)\sum_{i<j}^n(x_i\pm x_j)^{-2}\\
&-2k^{-1}(k^{-1}+1)\sum_{i<j}^m(y_i\pm y_j)^{-2}\\
&-l_1(l_1+1)\sum_{i=1}^n x_i^{-2}-l_2(l_2+1)\sum_{i=1}^m y_i^{-2}\\
&-2(k+1)\sum_{i=1}^n\sum_{j=1}^m(x_i\pm\sqrt{k} y_j)^{-2}\,.
\end{align*}
The first three lines of this expression describe the Calogero--Moser operator for the root system $R=B_n\times B_m$; the remaining sum is invariant under the action of the Coxeter group $W=W(B_n)\times W(B_m)$. 

Again, to check that $\A=B(n,m)$ is a locus configuration, it sufficient to check this for its rank-two subsystems. Some of them lie entirely in $R$ or are isomorphic to $A_1\times A_1$, e.g., $\A_0=\{\pm\alpha, \pm\beta\}$ with orthogonal $\alpha, \beta$. The remaining cases, up to permutations of indices and sign changes, are equivalent to the cases (1)--(2) from Example \ref{2dex}.

\medskip

\begin{remark}
In \cite{SV}, the operator $L_{BC(n,m)}$ is presented in a trigonometric form and in non-Cartesian coordinates. Our parameters $l_1, l_2, k$ are related to the parameters $p,q,r,s,k$ in \cite{SV} through $l_1=p+q$ and $l_2=r+s$.   
\end{remark}

%\medskip

%\noindent{\bf 3.} $AB(1, 3)$ {\it configuration} 

\subsubsection{$AB(1, 3)$ configuration}
In this case $V=\c^4$, and the corresponding Calogero--Moser operator in Cartesian coordinates $(x_1, x_2, x_3, y)$ is given by
\begin{align*}
L_{AB(1,3)}=\Delta&-\sum_{i=1}^3 a(a+1)x_i^{-2}-b(b+1)y^{-2}
-2c(c+1)\sum_{i<j}^3(x_i\pm x_j)^{-2}\\&-2(3k+3)\sum_{\pm}(\sqrt{3k}y\pm x_1\pm x_2\pm x_3)^{-2}\,.
\end{align*} 
Here the last sum is over all $8$ possible combinations of the signs. The parameters $a,b,c,k$ are related by
\begin{equation*}
a=\frac{3k+1}{2}\,,\quad b=\frac12(k^{-1}-1)\,,\quad c=\frac{3k-1}{4}\,.
\end{equation*} 
The $AB(1,3)$ configuration contains a Coxeter configuration of type $R=B_3\times A_1$. The remaining vectors $\alpha=(\pm 1, \pm 1, \pm 1, \pm \sqrt{3k})$ have multiplicities $k_\alpha=1$. One easily checks that any non-Coxeter rank-two subsystem is isomorphic to the cases (1)--(2) from Example \ref{2dex}, thus $AB(1,3)$ is a locus configuration. 

\begin{remark}
In \cite{SV}, the formula for $L_{AB(1,3)}$ contains a misprint: the numerical factor in front of the last sum in \cite[(14)]{SV} should be $\frac12$, not $\frac14$. 
\end{remark}

%\medskip

%\noindent{\bf 4.} $G(1, 2)$ {\it configuration} 

\subsubsection{$G(1, 2)$ configuration}
In this case $V=\c^4$, and the corresponding Calogero--Moser operator in Cartesian coordinates $(x_1, x_2, x_3, y)$ is given by
\begin{align*}
L_{G(1,2)}=\Delta&-2p(p+1)\sum_{i<j}^3 (x_i-x_j)^{-2}\\
&-3q(q+1)\sum_{i\ne j\ne l}^3 (x_i+x_j-2x_l)^{-2}\\
&-r(r+1)y^{-2}-4(k+1)\sum_{i\ne j}^3(\sqrt{2k}y-x_i+x_j)^{-2}\,.
\end{align*} 
The parameters $p,q,r,k$ are related by
\begin{equation*}
p={2k+1}\,,\quad q=\frac{2k-1}{3}\,,\quad r=\frac{3}{2}(k^{-1}+1)\,.
\end{equation*} 
The $G(1,2)$ configuration contains a Coxeter configuration of type $R=G_2\times A_1$. The remaining vectors $\alpha=\pm(-e_i+e_j+\sqrt{2k}e_4)$ have multiplicities $k_\alpha=1$.

\begin{remark}
The configuration $G(1,2)$ is contained in the hyperplane $x_1+x_2+x_3=0$ in $\c^4$. In \cite{SV}, the operator $L_{G(1,2)}$ is restricted onto this hyperplane (in non-Cartesian coordinates). Our parameters $p,q,r,k$ are related to the parameters $a,b,c,d,k$ in \cite{SV} through $p=a$, $q=b$ and $r=c+d$.   
\end{remark}

%\medskip

%\noindent{\bf 5.} $D(2, 1, \lambda)$ {\it configuration} 

\subsubsection{$D(2, 1, \lambda)$ configuration}
In this case $V=\c^3$. Let $\lambda_1, \lambda_2, \lambda_3$ be arbitrary non-zero parameters. Introduce 
\begin{equation*}
m_i=\frac{\lambda_1+\lambda_2+\lambda_3}{2\lambda_i}-1\quad(i=1,2,3)\,.
\end{equation*}
The configuration $D(2,1, \lambda)$ consists of the vectors $\alpha=\pm e_i$ with $k_\alpha=m_i$ and eight additional vectors 
$\alpha=\pm\sqrt{\lambda_1}e_1\pm\sqrt{\lambda_2}e_2\pm\sqrt{\lambda_3}e_3$ with $k_\alpha=1$.

The corresponding Calogero--Moser operator is given by
\begin{align*}
L_{D(2,1,\lambda)}=\Delta&-\sum_{i=1}^3 m_i(m_i+1)x_i^{-2}\\
&-2(\lambda_1+\lambda_2+\lambda_3)\sum_{\pm}(\sqrt{\lambda_1}x_1\pm\sqrt{\lambda_2}x_2\pm\sqrt{\lambda_3}x_3)^{-2}\,.
\end{align*} 
%(The last sum is over four possible combinations of the signs.) 

%\medskip

\subsection{$\boldsymbol{A_{n-1,2}}$ configuration}

\noindent $\A_+$ consists of the following vectors in
$\mathbb C^{n+2}$: $$\left\{
\begin{array}{lll}
e_i - e_j, &  1\le i<j\le n, & \text{with }\ k_\alpha=k\,,\\ e_i -
\sqrt{k}e_{n+1}, &  i=1,\ldots ,n\,,  & \text{with }\
k_\alpha=1\,,\\ e_i - \sqrt{k^*}e_{n+2}, & i=1,\ldots ,n\,, &
\text{with }\ k_\alpha=1\,,\\ \sqrt{k}e_{n+1}-\sqrt{k^*}e_{n+2} &
& \text{with }\ k_\alpha=1\,.
\end{array}
\right. 
$$ 
Here $k$ is an arbitrary parameter, $k^*=-1-k$, and $W=S_n$. A new feature in this case is that among the rank-two subsystems we have $e_i-\sqrt{k}e_{n+1}$, $e_i-\sqrt{k^*}e_{n+2}$, $\sqrt{k}e_{n+1}-\sqrt{k^*}e_{n+1}$, which is the case (3) in Example \ref{2dex}, with $a=1, b=\sqrt{k}, c=\sqrt{k^*}$. For $k\in\Z$ this is a locus configuration from \cite{CV1}. 

\subsection{$\boldsymbol{A(n_1, n_2, n_3)}$ configuration}
\la{A123}
The following configuration was described by D.~Gaiotto and M.~Rap\v{c}\'ak in \cite{GR}. It depends on three arbitrary integers $n_1, n_2, n_3\ge 0$ and $a,b,c\in\c$ such that $a^2+b^2+c^2=0$. The space $V$ is $V_1\oplus V_2\oplus V_3$ with $V_i=\c^{n_i}$. We denote as 
$\{e_i\}_{i=1\dots n_1}$ the standard basis in $V_1$, and similarly $\{e'_i\}_{i=1\dots n_2}$ and $\{e''_i\}_{i=1\dots n_3}$ for the other two spaces. We also write $x_i$, $x'_i$, $x''_i$ for the Cartesian coordinates in each of the spaces. For simplicity, we may not specify the index range explicitly so, for instance,  $e'_i$ or $x'_i$ below will automatically assume that $i\in\{1,\dots, n_2\}$.

The configuration is $\A=\A_+\sqcup (-\A_+)$ where $\A_+=R_1\sqcup R_2\sqcup R_3\sqcup R_{12}\sqcup R_{23}\sqcup R_{13}$ with the vectors and multiplicities in each group as follows:
\begin{align*}
R_1&=\{\alpha = e_i-e_j\}_{i<j} \quad\text{with $k_\alpha=b^2/a^2$}\,,\\
R_2&=\{\alpha = e'_i-e'_j\}_{i<j} \quad\text{with $k_\alpha=c^2/b^2$}\,,\\
R_3&=\{\alpha = e''_i-e''_j\}_{i<j} \quad\text{with $k_\alpha=a^2/c^2$}\,,\\
R_{12}&=\{\alpha = ae_i-be'_j\}_{i,j} \quad\text{with $k_\alpha=1$}\,,\\
R_{23}&=\{\alpha = be'_i-ce''_j\}_{i,j} \quad\text{with $k_\alpha=1$}\,,\\
R_{13}&=\{\alpha = ae_i-ce''_j\}_{i,j} \quad\text{with $k_\alpha=1$}\,.
\end{align*}
The Coxeter group $W$ is $S_{n_1}\times S_{n_2}\times S_{n_3}$ with $R_+=R_1\sqcup R_2\sqcup R_3$. An easy inspection shows that all two-dimensional sub-configurations are either Coxeter or are equivalent to cases (1) or (3) in Example \ref{2dex}. The 
deformed Calogero--Moser operator is
\begin{align}
L_\A&=\Delta_1+\frac{b^2c^2}{a^4}\sum_{i<j}\frac{2}{(x_i-x_j)^{2}}+a^2\sum_{i,j}\frac{2}{(bx'_i-cx''_j)^{2}}+\nonumber\\
&+\Delta_2+\frac{a^2c^2}{b^4}\sum_{i<j}\frac{2}{(x'_i-x'_j)^{2}}+b^2\sum_{i,j}\frac{2}{(ax_i-cx''_j)^{2}}+\label{GRex1}\\
&+\Delta_3+\frac{a^2b^2}{c^4}\sum_{i<j}\frac{2}{(x''_i-x''_j)^{2}}+c^2\sum_{i,j}\frac{2}{(ax_i-bx'_j)^{2}}\,.\nonumber
\end{align}
Here $\Delta_i$ denotes the Laplace operators on $V_i$. To compare this with \cite{GR}, one makes a change of variables $z_i=ax_i$, $z'_i=bx'_i$, $z''_i=cx''_i$ and sets $\epsilon_1=a^2$, $\epsilon_2=b^2$, $\epsilon_3=c^2$, $\epsilon_1+\epsilon_2+\epsilon_3=0$, after which $L$ takes the form \eqref{GRex} identical to $t_{2,0}$ in \cite[(2.15)]{GR}.
Note that for $(n_1,n_2,n_3)=(n,m,0)$ we have $R_3=R_{23}=R_{13}=\varnothing $, in which case $A(n_1,n_2,n_3)$ reduces to $A(n, m)$ with $k=b^2/a^2$. 
On the other hand, for $(n_1,n_2,n_3)=(n,1,1)$ it reduces to $A_{n-1,2}$ configuration. 

The fact that the operator \eqref{GRex1} is completely integrable is a corollary of Theorem \ref{gaic}. To spell this out, introduce the following set of ``deformed power sums'' (cf. \cite[(2.15)]{GR}):
\begin{equation*}
    p_d=a^{d-2}\sum_{i=1}^{n_1}x_i^d+b^{d-2}\sum_{i=1}^{n_2}(x'_i)^d+
    c^{d-2}\sum_{i=1}^{n_3}(x''_i)^d\,,\qquad d=1,2,\dots .
\end{equation*}
These polynomials are 
obviously symmetric with respect to $W=S_{n_1}\times S_{n_2}\times S_{n_3}$. Moreover, each $p_d$ is quasi-invariant. Indeed, for $\alpha=ae_i-be'_j$ we have
\begin{equation*}
  \partial_{\alpha}(p_d)=d(ax_i)^{d-1}-d(bx'_j)^{d-1}=0\qquad\text{for}\quad ax_i-bx'_j=0\,,  
\end{equation*}
hence $p_d(x)-p_d(s_\alpha x)$ is divisible by ${(\alpha,x)^2}$, and similarly for the other vectors $\alpha\in\A\setminus R$.
Following the proof of Theorem \ref{gaic}, we set 
\begin{equation*}
L_d:=\frac{1}{2^{d}d!}\ad_{L_\A}^{d}(p_d)\,,%\qquad q\in\qaw\,,
\qquad d=1,2,\dots\,.
\end{equation*}
Note that $L_d$ has the form
$L_d=p_d(\partial) +\ldots$\,, with $L_2=L_\A$. Hence, we obtain the following result. %(stated without proof in \cite{GR}).

\begin{prop}\label{a3int}
The operator \eqref{GRex1} is completely integrable, with $[L_{d_1}, L_{d_2}]=0$ for $d_1, d_2\ge 1$.
\end{prop}

\begin{remark}
It is not possible to further extend the $A(n_1, n_2, n_3)$ configuration by allowing four (or more) groups of variables. The obstruction to that is the Calabi-Yau condition $a^2+b^2+c^2=0$: if we had four groups, with parameters $a,b,c,d$, say, we would need the sum of squares to be zero for each three of them which would force $a=b=c=d=0$.    
\end{remark}

\subsection{$\boldsymbol{BC(n_1, n_2, n_3)}$ configuration}
\la{BC123}
This is a $BC$-type generalisation of the $A(n_1,n_2,n_3)$ family. To describe it
we will use the notation and  conventions introduced in the previous section.

The configuration has the form 
$\A=R_1\sqcup R_2\sqcup R_3\sqcup R_{12}\sqcup R_{23}\sqcup R_{13}$ with the vectors and multiplicities in each group as follows:
\begin{align*}
R_1&=
\{\pm e_i\pm e_j \ \text{with $i<j$ and $k_\alpha=b^2/a^2$}\}
\sqcup 
\{\pm e_i\ \text{with $k_\alpha=l_1$}\}\,,
\\
R_2&=
\{\pm e'_i\pm e'_j \ \text{with $i<j$ and $k_\alpha=c^2/b^2$}\}
\sqcup 
\{\pm e'_i\ \text{with $k_\alpha=l_2$}\}\,,
\\
R_3&=
\{\pm e''_i\pm e''_j \ \text{with $i<j$ and $k_\alpha=a^2/c^2$}\}
\sqcup 
\{\pm e''_i\ \text{with $k_\alpha=l_3$}\}\,,
\\
R_{12}&=\{\alpha = \pm ae_i\pm be'_j\}_{i,j} \quad\text{with $k_\alpha=1$}\,,\\
R_{23}&=\{\alpha = \pm be'_i\pm ce''_j\}_{i,j} \quad\text{with $k_\alpha=1$}\,,\\
R_{13}&=\{\alpha = \pm ae_i\pm ce''_j\}_{i,j} \quad\text{with $k_\alpha=1$}\,.
\end{align*}
Here $a^2+b^2+c^2=0$ and $l_1, l_2, l_3$ are such that
\begin{equation*}
a^2(2l_1+1)=\pm b^2(2l_2+1)=\pm c^2(2l_3+1)\,.
\end{equation*}
The Coxeter group $W$ in this case has the root system $R=R_1\sqcup R_2\sqcup R_3$ of type $B_{n_1}\times B_{n_2}\times B_{n_3}$. By an easy inspection, all two-dimensional sub-configurations in $BC(n_1, n_2, n_3)$ are equivalent to those in Example \ref{2dex}.

The deformed Calogero--Moser operator is
\begin{align*}
L_\A&=\Delta_1+\frac{b^2c^2}{a^4}\sum_{i<j, \pm}\frac{2}{(x_i\pm x_j)^{2}}- \sum_{i} \frac{l_1(l_1+1)}{x_i^{2}}+a^2\sum_{i,j,\pm}\frac{2}{(bx'_i\pm cx''_j)^{2}}+\\
&+\Delta_2+\frac{a^2c^2}{b^4}\sum_{i<j, \pm}\frac{2}{(x'_i\pm x'_j)^{2}}- \sum_{i} \frac{l_2(l_2+1)}{(x_i^{'})^2}+b^2\sum_{i,j, \pm}\frac{2}{(ax_i\pm cx''_j)^{2}}+\\
&+\Delta_3+\frac{a^2b^2}{c^4}\sum_{i<j, \pm}\frac{2}{(x''_i\pm x''_j)^{2}}- \sum_{i} \frac{l_3(l_3+1)}{(x_i^{''})^2}+c^2\sum_{i,j, \pm}\frac{2}{(ax_i\pm bx'_j)^{2}}\,.
\end{align*}
When $(n_1,n_2,n_3)=(n,m,0)$ we have $R_3=R_{23}=R_{13}=\varnothing $ in which case $BC(n_1,n_2,n_3)$ reduces to $BC(n, m)$ with parameters $k=b^2/a^2$ and $l_1, l_2$. 

By Theorem \ref{gaic}, the above Calogero--Moser operator $L_\A$ is completely integrable. We also have a result analogous to Proposition \ref{a3int}. To be precise, for $\, d=1,\,2,\ldots\,$,
consider the following  quasi-invariant polynomials
\begin{equation*}
    p_d=a^{2d-2}\sum_{i=1}^{n_1}x_i^{2d}+b^{2d-2}\sum_{i=1}^{n_2}(x'_i)^{2d}+
    c^{2d-2}\sum_{i=1}^{n_3}(x''_i)^{2d}
\end{equation*}
and set 
\begin{equation*}
L_d := \frac{1}{2^{2d}(2d)!}\ad_{L_\A}^{2d}(p_d)\,.
\end{equation*}
\begin{prop}
$\,[L_{d_1}, L_{d_2}]=0\,$ for all $ d_1,\,d_2 \ge 1 $. Moreover,
$ L_1 = L_\A $. 
\end{prop}

\begin{remark}
There exists a trigonometric version of the $BC(n_1,n_2,n_3)$ that has additional parameters $k_1,k_2, k_3$ satisfying the relation $a^2k_1=b^2k_2=c^2k_3$. In the case $(n_1,n_2,n_3)=(n,m,0)$ it reduces to the operator \cite[(5)]{SV} with parameters 
$k=b^2/a^2$, $p=k_1$, $q=l_1$, $r=k_2$, $s=l_2$.  
\end{remark}

%\begin{remark}
%The families $A(n_1, n_2, n_3)$ and $BC(n_1, n_2, %n_3)$ do not admit a further generalisation with %four (or more) groups of indices. For example, 
%
%\end{remark}
\subsection{Restricted Coxeter configurations}\la{rc}

Another class of generalised locus configurations can be found in \cite{F}. These configurations appear as restrictions of Coxeter root systems onto suitable subspaces (parabolic strata). They are labelled by pairs $(\Gamma, \Gamma_0)$ of Dynkin diagrams where $\Gamma_0\subset \Gamma$ is possibly disconnected. For a given $\Gamma$, the admissible sub-diagrams $\Gamma_0$ are characterized by a certain geometric condition (see \cite{F}, Theorems 1-3). The classical case $\Gamma=A_n, B_n, D_n$ leads to special cases of the configurations already listed above. The list of possibilities for the exceptional cases $\Gamma=F_4, E_{6-8}, H_4, H_3$ includes $43$ cases and can be found in Section 6 of \cite{F}. For example, one has the following configuration $(F_4, A_1)$ in $\c^3$, see \cite[(27)]{F}:
$$\left\{
\begin{array}{lll}
\pm e_i, &  1\le i\le 3, & \text{with }\ k_\alpha=2c+\frac12\,,\\ \pm e_i \pm e_{j}, &  1\le i< j\le 3\,,  & \text{with }\
k_\alpha=c\,,\\ \pm e_1\pm e_2\pm e_3 &  &
\text{with }\ k_\alpha=1\,.
\end{array}
\right. $$
We have checked, case-by-case, that all two-dimensional configurations in \cite{F} satisfy the generalised locus conditions (and, in fact, can be constructed by the method of Proposition \ref{2d} below). Apart from the two-dimensional and Coxeter configurations, the list in \cite[Section 6]{F} contains $23$ additional cases. Note that the fact that all of these are indeed generalised locus configurations does \emph{not} follow directly from their construction in \cite{F}, and so one has to rely upon case-by-case verification. We expect that all configurations in \cite{F} satisfy Definition \ref{defloc}, though we have not checked this for all of the cases. Note that by Proposition \ref{2dim}, it is sufficient to check the conditions of Definition \ref{defloc} for all two-dimensional subconfigurations. 

\medskip

\begin{remark} The configurations described in \ref{cc}, \ref{sa} and \ref{rc} are non-twisted. This is obvious for the case \ref{cc}; for the cases in \ref{sa} it can be checked directly. For those cases in \ref{rc} where $k_\alpha=1$ for $\alpha\in\A\setminus R$, this follows from \cite[Proposition 2]{F}; for the remaining cases it can be verified directly, case by case. The $A_{n-1, 2}$ configuration is twisted; the same is true for $A(n_1,n_2,n_3)$ and $B(n_1,n_2,n_3)$ with $n_1,n_2,n_3\ge 1$. Also, the configurations constructed in Proposition \ref{2d} below
are twisted in general.
\end{remark}

\begin{remark}
Locus configurations with $R=\varnothing$ can be obtained from the above cases by specialising parameters. The list of all such configurations currently known consists of: (1) Coxeter configurations, with all $k_\alpha\in\Z$, (2) $A(n,1)$ with $k\in \Z$, (3) $B(n,1)$ with $k, a, b\in\Z$, (4) $A_{n-1, 2}$ with $k\in\Z$, and (5) the Berest--Lutsenko family in $\c^2$ (see \cite{BL, CFV, C08}).
\end{remark}

\section{Two-dimensional configurations}
\label{twodim}

In dimension two, the problem of describing generalised locus configurations can be reduced to a one-dimensional problem which, in turn, can be solved with the help of Darboux transformations. To begin with, note that for a locus configuration in $\c^2$, the Coxeter group $W$ can be one of the following: (1) $W=\{e\}$ with $R=\varnothing$, (2) $W=\Z_2$ with $R=\{\pm\alpha\}$, or (3) $W=I_N$, the dihedral group of order $2N$, $N\ge 2$. We analyse first the case of $W=I_2$; the other cases will follow a similar pattern. 

\subsection{The $W=I_2$ case}\label{i2}
The Calogero--Moser operator for $W=I_2$ can be written in polar coordinates $r, \varphi$ as    
%\begin{equation*} 
%L=\frac{\partial^2}{\partial r^2}-r^{-2}L_0\,,\qquad %L_0=-\frac{\partial^2}{\partial\varphi^2}+\frac{g(g-1)}{\sin^{2}\varphi}+\frac{h(h-1)}{\cos%^{2}\varphi}\,,
%\end{equation*}
\begin{equation}\label{dpt} 
L=\frac{\partial^2}{\partial r^2}+r^{-2}L_0\,,\qquad L_0=\frac{\partial^2}{\partial\varphi^2}-\frac{\alpha^2-\frac14}{\sin^{2}\varphi}-\frac{\beta^2-\frac14}{\cos^{2}\varphi}\,,
\end{equation}
where $L_0$ is known as the Darboux--P\"oschl--Teller operator and is closely related to the Jacobi differential operator. A generalised locus configuration $\A$ of type $W=I_{2}$ is obtained by adding to $e_1, e_2$ with multiplicities $\alpha-\frac12$, $\beta-\frac12$ a number of vectors with integral multiplicities. %The configuration $\A$ and multiplicities must be invariant under $\varphi\mapsto \varphi+\pi/n$. 
Therefore,
\begin{align}\label{dpt1} 
L_\A&=\frac{\partial^2}{\partial r^2}+r^{-2}L_1\,,\qquad L_1=\frac{\partial^2}{\partial\varphi^2}-v(\varphi)\,,
\\
\label{dpt2}
v(\varphi)&=\frac{\alpha^2-\frac14}{\sin^{2}\varphi}+\frac{\beta^2-\frac14}{\cos^{2}\varphi}+\sum_i \frac{k_i(k_i+1)}{\sin^{2}(\varphi-\varphi_i)}\,,
\end{align}
where $\varphi_i\in\c$ describe the positions of the added lines with $k_i\in\Z_+$. The locus conditions \eqref{loc} for $\alpha=(\cos\varphi_i, \sin\varphi_i)$ translate into
\begin{equation}\label{trigloc}
 v(\varphi) -v(s_i\varphi)\quad\text{is divisible by}\quad (\varphi-\varphi_i)^{2k_i}\,,
\end{equation}
for each of the reflections $s_i:\,\varphi\mapsto 2\varphi_i-\varphi$. The problem of describing such functions $v$ is closely related to Darboux transformations. We have the following result.
\begin{prop}\label{dt}
The potential $v$ of the form \eqref{dpt2} satisfies conditions \eqref{trigloc} if and only if there exists a differential $D(\varphi, \frac{\partial}{\partial\varphi})$ with trigonometric coefficients, which intertwines the operators $L_0$ \eqref{dpt} and $L_1$ \eqref{dpt1}, i.e. such that $L_1D=DL_0$. 
\end{prop}
We will only prove the ``if'' part; the ``only if'' part will be discussed elsewhere. 
\begin{proof} Suppose such an intertwiner $D$ exists. This means that for generic $\lambda\in\c$, $D$ sets up a bijection between $\Ker(L_0-\lambda)$ and $\Ker(L_1-\lambda)$. Note that eigenfunctions of $L_0$ are meromorphic near $\varphi=\varphi_i$; hence, so are eigenfunctions of $L_1$. Thus, for generic $\lambda$, all solutions $f$ to the one-dimensional Schr\"odinger equation \begin{equation*}
\left(\frac{\partial^2}{\partial\varphi^2}-v(\varphi)\right)f=\lambda f    
\end{equation*} 
are meromorphic (single-valued) near $\varphi=\varphi_i$.
By \cite[Prop. 3.3]{DG}, this is equivalent to the following property of the Laurent series for $v$:
\begin{equation*}\label{trigloc1}
    v=\sum_{j=-2}^\infty v_j(\varphi-\varphi_i)^j\,,\qquad v_{-2}=k_i(k_i+1)\,,\quad v_{-1}=v_1=\dots =v_{2k_i-1}=0\,.%\quad\text{for}\quad j=-1,1,\dots, 2k_i-1\,.
\end{equation*}
%The properties 
This is obviously equivalent to \eqref{trigloc}.
\end{proof}
Proposition \ref{dt} tells us that $L_1$ is related to $L_0$ by a higher-order Darboux transformation. Under suitable circumstances, one can obtain $D$ by iterating elementary (first order) Darboux transformation, leading to an explicit formula for the potential $v$. To state the result, recall %(see, for instance, \cite{STZ}) 
that $L_0$ has a well-known family of eigenfunctions of the form %expressed in terms of the Jacobi polynomials as 
\begin{equation*}
\psi_n(x)=(\sin x)^{\alpha+\frac12}(\cos x)^{\beta+\frac12} P_n^{\alpha, \beta}(\cos 2x)\,,\quad n=0, 1, 2, \dots\,,
\end{equation*} 
where $P_n^{\alpha, \beta}(z)$ are the classical Jacobi polynomials. Since $L_0$ does not change under $\alpha\mapsto -\alpha$ or $\beta\mapsto -\beta$, we obtain four families of (formal) eigenfunctions by using $\pm \alpha, \pm \beta$ in the above $\psi_n$. Let $\mathcal F$ denote the union of these four families of functions.

\begin{prop}
\label{2d}
\begin{enumerate}
\item[(1)] For distinct $f_1, \dots, f_m\in\mathcal F$, the potential \begin{equation}\label{wr}
v=\frac{\alpha^2-\frac14}{\sin^{2}\varphi}+\frac{\beta^2-\frac14}{\cos^{2}\varphi}-2\frac{d^2}{d\varphi^2}\log \mathcal W\,,\quad \mathcal W=\mathrm{Wr}(f_1, \dots, f_m)\,,
\end{equation}  
satisfies the conditions \eqref{trigloc} and, therefore, the singularities of $u=r^{-2}v(\varphi)$ form a locus configuration of type $W=I_2$. 

\item[(2)] Assuming $\alpha, \beta$ are generic, the formula \eqref{wr} produces \emph{all} locus configurations of type $W=I_2$ in $\c^2$.
\end{enumerate}
\end{prop}
We will only prove part (1); part (2) will be discussed elsewhere. 
\begin{proof}
Let 
\begin{equation*}
v_0=\frac{\alpha^2-\frac14}{\sin^{2}\varphi}+\frac{\beta^2-\frac14}{\cos^{2}\varphi}\,,\qquad v=v_0-2\frac{d^2}{d\varphi^2}\log \mathcal W
\,.\end{equation*}
 By a standard result on Darboux transformations (see e.g. \cite{Cr}), the operators   
\begin{equation*}
L_0=\frac{\partial^2}{\partial\varphi^2}-v_0\,,\qquad L_1=\frac{\partial^2}{\partial\varphi^2}-v
\end{equation*}
are intertwined by a (monic) differential operator $D$ whose kernel is spanned by all of $f_i$. The choice of $f_i$ makes it clear that $D$ will have trigonometric coefficients. By Proposition \ref{dt}, the potential $v$ satisfies \eqref{trigloc}. This implies locus conditions \eqref{loc} for $u=r^{-2}v(\varphi)$, as needed. 
\end{proof}
\begin{remark} For Proposition \ref{2d}(2) to be valid, it is enough to assume that $\alpha, \beta, \alpha\pm\beta\notin\Z$.  For special values of $\alpha, \beta$ there are more general solutions $v(\varphi)$ than those described by \eqref{wr}. The corresponding Darboux transformations have been studied in the context of exceptional orthogonal polynomials, see \cite{GGMM} and references therein. Their full classification is not known, to the best of our knowledge.  
\end{remark}

\subsection{The $W=I_{N}$ case}\label{in}
The Calogero--Moser operator for $W=I_N$ can be written as
$L=\Delta-r^{-2}v_0$,
where
\begin{equation*}
v_0=\begin{cases} (\alpha^2-\frac14)N^2\sin^{-2}N\varphi\quad&\text{for $N$ odd}\,,\\(\alpha^2-\frac14)n^2\sin^{-2}n\varphi+(\beta^2-\frac14)n^2\cos^{-2}n\varphi\quad&\text{for $N=2n$ even}\,. \end{cases}
\end{equation*}
The group $W$ is generated by the reflection $\varphi\mapsto -\varphi$ and rotation $\varphi\mapsto \varphi+2\pi/N$. Note, however, that any configuration is automatically invariant under $\varphi\mapsto \varphi +\pi$. Thus, even when $N$ is odd, the full symmetry group can be taken as $W=I_{2N}$. This allows us to consider the case of odd $N$ as a special case of $W=I_{2N}$, with $\beta=1/2$. Thus, below we restrict ourselves to the case $W=I_{2n}$ and $v_0=(\alpha^2-\frac14)n^2\sin^{-2}n\varphi+(\beta^2-\frac14)n^2\cos^{-2}n\varphi$. 

\medskip

A locus configuration $\A$ of type $W=I_{2n}$ must be invariant under $\varphi\mapsto \varphi+\pi/n$, therefore, similarly to the $W=I_2$ case,
\begin{equation}\label{cm2d} 
L_\A=\frac{\partial^2}{\partial r^2}+r^{-2}\left(\frac{\partial^2}{\partial\varphi^2}-v(\varphi)\right)\,,\qquad v(\varphi)=v_0+\sum_i\frac{k_i(k_i+1)n^2}{\sin^{2}(n\varphi-\varphi_i)}\,,
\end{equation}
for some $\varphi_i\in\c$ and $k_i\in\Z_+$. As we have seen in the proof of Proposition \ref{dt}, the locus conditions \eqref{loc} express the property that eigenfunctions of $L_1=\frac{\partial^2}{\partial\varphi^2}-v(\varphi)$ are single-valued near each of the singular points $\varphi=\varphi_i$. Clearly, this property is preserved under rescaling  $\varphi\mapsto \varphi/n$, which puts $L_1$ in the form
\begin{equation*} 
L_1=n^2\left(\frac{\partial^2}{\partial\varphi^2}-\frac{(\alpha^2-\frac14)}{\sin^{2}\varphi}-\frac{(\beta^2-\frac14)}{\cos^{2}\varphi}-\sum_i \frac{k_i(k_i+1)}{\sin^{2}(\varphi-\varphi_i)}\right)\,.
\end{equation*}
This reduces our locus configuration to the one of type $W=I_2$. Hence, we have the following result.

\begin{prop}
\label{2ton}
For a generalised locus configuration $\A\subset\c^2$, write the potential $u_\A$ \eqref{gcmu} in polar coordinates as $u_\A=r^{-2}v(\varphi)$. The mapping 
\begin{equation*}
u_\A\mapsto u_{\A'}\,,\qquad r^{-2}v(\varphi)\mapsto n^2r^{-2}v(n\varphi)
\end{equation*}
establishes a one-to-one correspondence between locus configurations $\A$ of type $W=I_2$ and locus configurations $\A'$ of type $W=I_{2n}$.
\end{prop}

%Every locus configurations $\A$ of type $W=I_{2n}$ %and the corresponding Calogero--Moser operators $\A$ are 
%is produced from some generalised Calogero--Moser operator $L$ of type $W=I_2$,
%\begin{equation*}
%L=\frac{\partial^2}{\partial r^2}+r^{-2}\left(\frac{\partial^2}{\partial\varphi^2}-v(\varphi)\right)\,,
%\end{equation*} 
%by the formula
%\begin{equation*}
%L_\A=\frac{\partial^2}{\partial r^2}+r^{-2}\left(\frac{\partial^2}{\partial\varphi^2}-n^2v(n\varphi)\right)\,.
%\end{equation*}

%\begin{remark} 
%As explained in Section~\ref{pro}, the existence of shift operators (and hence complete integrability) for the Calogero-Moser operators $L$ with potentials described in Proposition~\ref{2d} follows immediately from Lemma~\ref{shiftl}. The question whether Proposition \ref{2d} gives \emph{all} generalised locus configurations in $\c^2$ is more delicate and will be discussed elsewhere. 
%\end{remark}

Let us illustrate this with two examples of type $W=I_2$, with the Calogero-Moser operator of the following form:
\begin{equation*}
L=\Delta-\frac{k_1(k_1+1)}{x_1^2}-\frac{k_2(k_2+1)}{x_2^2}-\frac{k_3(k_3+1)(1+a^2)}{(x_1-ax_2)^2}-\frac{k_3(k_3+1)(1+a^2)}{(x_1+ax_2)^2}\,.
\end{equation*}
The parameters $k_1, k_2, k_3$ and $a$ are as follows:

\medskip

(1) $k_1, k_2$ arbitrary, $k_3=1$, $a=\sqrt{\frac{2k_1+1}{2k_2+1}}$;

(2) $k_1=\frac{3a^2}{4}-\frac{1}{4}$, $k_2=\frac{3}{4a^2}-\frac{1}{4}$, $k_3=2$, $a$ arbitrary.

\medskip

\noindent In both cases, the locus conditions \eqref{loc1} for $\alpha=e_1\pm ae_2$ are easy to verify directly. %Note that both configurations are of non-twisted type. 
The first case corresponds to $m=n=1$ in the $B(n,m)$ family. The second case was proposed and studied in \cite{T1}, where the integrability of $L$ (including its elliptic version) was confirmed. Note that the first case can be obtained by setting $m=1$, $f_1=\psi_1$ in Proposition \ref{2d}. 

Now let us apply the substitution $\varphi\mapsto 2\varphi$ in accordance with Proposition \ref{2ton}. This leads to the Calogero--Moser operators of type $W=I_4$ of the form
\begin{multline*}
L=\Delta-\frac{k_1(k_1+1)}{x_1^2}-\frac{k_1(k_1+1)}{x_2^2} -\frac{k_2(k_2+1)}{(x_1+x_2)^2}-\frac{k_2(k_2+1)}{(x_1-x_2)^2}
\\-\frac{k_3(k_3+1)(1+b^2)}{(x_1-bx_2)^2}-\frac{k_3(k_3+1)(1+b^2)}{(x_1+bx_2)^2}
\\-\frac{k_3(k_3+1)(1+b^2)}{(x_2-bx_1)^2}-\frac{k_3(k_3+1)(1+b^2)}{(x_2+bx_1)^2}\,.
\end{multline*}
Here $k_1, k_2, k_3$ are the same as above, while $a$ and $b$ are related by $a={2b}/({1-b^2})$. In the case $k_3=1$, the above $L$ coincides with \cite[(28)]{F}. 

As another example, applying the substitution $\varphi\mapsto 3\varphi$ to the case $(k_1, k_2, k_3, a)=(\frac13, 4, 1, \frac{\sqrt{5}}{3\sqrt{3}})$, one obtains the configuration that coincides (up to an overall rotation) with the configuration $(\mathcal H_4, \mathcal A_2)$ in \cite{F}.     

\medskip

\begin{remark} Most configurations constructed in Propositions \ref{2d}, \ref{2ton} are twisted. For the study of non-twisted planar configurations we refer to \cite{FJ} (see, in particular, Proposition 3.1 and Theorem 3.4 in {\it loc.~cit.}).
\end{remark}

\subsection{The $W=\Z_2$, $R=\{\pm\alpha\}$ case} Since $\A\to -\A$, the symmetry group $W=\Z_2$ can be extended to $W=I_2$, with arbitrary $\alpha$ and with $\beta=1/2$. Hence, the results of \ref{i2} can be applied. In particular, for generic values of $\alpha$ all such locus configurations can be obtained from Proposition \ref{2d}. 

\subsection{The $W=\{e\}$, $R=\varnothing$ case}\label{blu}
In this case, $\A$ must be a plane locus configuration in the sense of \cite{CFV}. All such $\A$ belong to the so-called Berest--Lutsenko family, see \cite{BL} and \cite[Theorem 4.3]{CFV}. Namely, %$u_\A=r^{-2}v(\varphi)$ where 
\begin{equation*}
    u_\A=r^{-2}v(\varphi)\quad\text{with}\quad v=-2\frac{d^2}{d\varphi^2}\log \mathrm{Wr}(f_1, \dots, f_m)\,,\quad f_i=\cos(l_i\varphi+\theta_i)\,,
\end{equation*}
where $1\le l_1<\dots<l_m$ are arbitrary integers and $\theta_i\in\c$. 
Hence, \eqref{wr} can be seen as a generalisation of the above family. An important difference is that here, in addition to discrete parameters $l_i$, we also have continuous parameters $\theta_i$.

\section{Deformed Calogero--Moser operators with harmonic oscillator confinement}
\la{S8}
With a generalised locus configuration $\A\subset V$ one can associate the Calogero--Moser operator with an extra ``oscillator term'':

\begin{equation}\label{gcmo}
L_{\A}^\omega:=\Delta-\omega^2x^2-\sum_{\alpha\in \A_+} \frac{k_\alpha(k_\alpha+1)(\alpha,\alpha)}{(\alpha,x)^2}\,,
\end{equation}
where $x^2=(x,x)$ is the Euclidean square norm in $V=\R^n$, and $\omega$ is an arbitrary parameter. In particular, when $\A=R$ is the root system of a Coxeter group $W$, we have
\begin{equation}\label{cmoo}
L_{W}^\omega:=\Delta-\omega^2x^2-\sum_{\alpha\in R_+} \frac{k_\alpha(k_\alpha+1)(\alpha,\alpha)}{(\alpha,x)^2}\,.
\end{equation}
For the classical groups $W=A_n, B_n, D_n$ the operator $L_W^\omega$ is known to be Liouville integrable (see \cite{F} and references therein). For the exceptional groups this does not seem to be known (for dihedral groups the complete integrability is easy to show). Still, for any Coxeter group the operator $L_W^\omega$ has several hallmarks of integrability, also shared by $L_\A^\omega$. These are collected in the following theorem.

\begin{theorem}\label{gaicw} For a Coxeter group $W$, let $\A\subset \c^n$ be a locus configuration of type $W$, with $\delta$ given by \eqref{del}. Let $S$ be the shift operator from Theorem \ref{gaic}\,$(1)$, see \eqref{shiftopd}. For $q\in\c[V]^W$ and $q\in\qaw$, respectively, let $L_{q,0}$ and $L_q$ be the differential operators from Theorem \ref{gaic}\,$(3)$, see \eqref{loqc}, \eqref{lq1}.

\begin{enumerate}
    \item[(1)]  Set $S^\omega=e^{-\omega x^2/2} S e^{\omega x^2/2}$, and write $(S^\omega)^*$ for its formal adjoint. Then 
\begin{gather}
  \la{into}
L_\A^\omega S^\omega=S^\omega (L_W^\omega-2\omega N)\,,\qquad N=\deg\delta\,, 
\intertext{and}
[S^\omega (S^\omega)^*, L_\A^\omega]=0\,.\label{into1}  
\end{gather}

\item[(2)] For a homogeneous $q\in\c[V]^W$, set $L_{q,0}^\omega=e^{-\omega x^2/2} L_{q,0} \,e^{\omega x^2/2}$. Then 
\begin{equation*}
L_{q,0}^\omega  L_{W}^\omega =(L_{W}^\omega +2\omega r)L_{q,0}^\omega\,,\qquad r=\deg q\,.
\end{equation*}
Hence, the operator $L_{q,0}^\omega L_{q,0}^{-\omega}$ commutes with $L_{W}^\omega$.  

\item[(3)] Similarly, for a homogeneous $q\in Q_\A$, set $L_{q}^\omega=e^{-\omega x^2/2} L_{q} \,e^{\omega x^2/2}$. 
Then 
\begin{equation*}
L_{q}^\omega  L_{\A}^\omega =(L_{\A}^\omega +2\omega r)L_{q}^\omega\,,\qquad  r=\deg q\,.
\end{equation*}
Hence, the operator $L_{q}^\omega L_{q}^{-\omega}$ commutes with $L_{\A}^\omega$.

\end{enumerate}

\end{theorem}

\begin{proof}
Define
\begin{equation*}
L_0^\omega=e^{\omega x^2/2} L_W^\omega e^{-\omega x^2/2}\,,\qquad  L^\omega=e^{\omega x^2/2} L_\A^\omega e^{-\omega x^2/2}\,.
\end{equation*}
By direct calculation, $L_0^\omega=L_W-2\omega E$ and $L^\omega=L_\A-2\omega E$, where $E=\sum_{i=1}^n x_i\partial_i$ is the Euler operator.
From the homogeneity of $L_W$, $L_\A$, and $S$,
\begin{equation*}
[E, L_W]=-2L_W\,,\quad [E, L_\A]=-2L_\A\,,\quad [E, S]=-NS\,.
\end{equation*} 
Thus, using that $L_\A S=SL_W$, we obtain
\begin{equation*}
(L_\A-2\omega E)S=S(L_W-2\omega E)+2\omega[E,S]=S(L_W-2\omega E-2\omega N)\,,
\end{equation*}
or $L^\omega S= S(L_0^\omega-2\omega N)$. Conjugating this relation by $e^{\omega x^2/2}$ gives \eqref{into}.
Furthermore, taking formal adjoints in \eqref{into}, we obtain $(S^\omega)^* L_\A^\omega =(L_W^\omega-2\omega N)(S^\omega)^*$.  Combining this with \eqref{into} gives $L_\A S^\omega (S^\omega)^*=S^\omega (S^\omega)^* L_\A$, which is \eqref{into1}.

For part (3), we first note that $L_q$ given is homogeneous of degree $-r$. Using this and $L_qL_\A=L_\A L_q$, we get  
\begin{equation*}
L_q(L_\A-2\omega E)=(L_\A-2\omega E)L_q+2\omega r L_q\,,\quad\text{or}\quad L_q L^\omega
=(L^\omega +2\omega r)L_q\,.
\end{equation*}
Conjugating this by $e^{\omega x^2/2}$ gives $L_{q}^\omega  L_{\A}^\omega =(L_{\A}^\omega +2\omega r)L_{q}^\omega$, as needed. Changing $\omega\mapsto -\omega$, we obtain $L_{q}^{-\omega}  L_{\A}^\omega =(L_{\A}^\omega -2\omega r)L_{q}^{-\omega}$. The commutativity of $L_{q}^\omega L_{q}^{-\omega}$ and $L_\A^\omega$ is then obvious. This proves part (3). Part (2) is entirely similar.
\end{proof}

\begin{remark}
Even for $L_W^\omega$, the result of part (2) seems new. For locus configurations with $W=\{e\}$, the existence of an intertwiner $S^\omega$ satisfying \eqref{into} was established in \cite{CO} by a considerably more involved argument.
\end{remark}

\medskip

Let us apply these results to construct a large family of {\it quantum superintegrable systems} in two dimensions.
Take a locus configuration $\A$ of type $W=I_{2n}$ in the plane. The deformed Calogero--Moser operator  $L=L_\A^\omega$ in polar coordinates is given by
\begin{equation*}
L_\A^\omega=\frac{\partial^2}{\partial r^2}+r^{-2}\left(\frac{\partial^2}{\partial\varphi^2}-v(\varphi)\right)-\omega^2r^2\,,     
\end{equation*}
where $v$ is as in \eqref{cm2d}.
Obviously, $L_\A^\omega$ commutes with $L_1=\frac{\partial^2}{\partial\varphi^2}-v(\varphi)$, hence it is completely integrable. According to Theorem \ref{gaicw}, operators $S^\omega(S^\omega)^*$ as well as $L_q^\omega L_q^{-\omega}$ for $q\in\qaw$ also commute with $L_\A^\omega$. Therefore, $L_\A^\omega$ is (maximally) superintegrable for any locus configuration of type $W=I_2$ or $W=I_{2n}$ discussed in \ref{i2}, \ref{in}. The same is true for locus configurations $\A$ of type $W=\{e\}$ in \ref{blu}. Hence, we have the following result.
\begin{prop}\label{suprop}
For any generalised locus configuration $\A$ in the plane, the operator \eqref{gcmo} is superintegrable. 
\end{prop}
Some special cases of these systems have been studied in the literature on quantum superintegrability. %In comparison, our method allows for a simple unified treatment and gives stronger results. 

\begin{example}
\la{Ex8.4}
Taking $\A$ from Example \ref{2dex}\,(2), we have
\begin{equation*}
L_\A^\omega=\Delta
%\frac{\partial^2}{\partial x^2}+ \frac{\partial^2}{\partial y^2}
-\omega^2(x^2+y^2)-\frac{l(l+1)}{x^2}-\frac{m(m+1)}{y^2}-\frac{4(a^2+b^2)(a^2x^2+b^2y^2)}{(a^2x^2-b^2y^2)^2}\,,
\end{equation*}
where $l,m$ are arbitrary and $(2l+1)a^2=\pm(2m+1)b^2$. 
In this case 
$$\A_+\setminus R=\{ae_1-be_2, ae_1+be_2\}$$ so the shift operator $S$ has order two: $S=a^2\frac{\partial^2}{\partial x^2}-b^2 \frac{\partial^2}{\partial y^2}+\ldots$. Hence, the commuting operator 
$S^\omega(S^\omega)^*$ is of order four: 
$$S^\omega(S^\omega)^*=\left(a^2\frac{\partial^2}{\partial x^2}-b^2 \frac{\partial^2}{\partial y^2}\right)^2+\ldots\,.$$
\end{example}
This example appears in \cite{PTV}; references to more recent work can be found in \cite{MPR} in which the above operator $L_\A^\omega$ appears in Eq. (1).

\section{Affine configurations}
\la{S9}
Our main results can be easily extended to affine (i.e., noncentral) hyperplane arrangements.
As before, we start with a Coxeter group $W$ with root system $R$, in its reflection representation $V$ equipped with a $W$-invariant scalar product $(\cdot , \cdot)$. Let $\VV$ be the vector space of affine-linear functions on $V$. We identify $\VV$ with $V\oplus\c c$, where vectors in $V$ are considered as linear functionals on $V$ via the scalar product $(\cdot, \cdot)$ and where $c\equiv 1$ on $V$. 
The action of $W$ extends onto $\VV$ in an obvious way, with $w(c)=c$ for all $w\in W$. For any $\aalpha=\alpha+rc\in \VV$ we have the orthogonal reflection with respect to the hyperplane $\aalpha(x)=0$ in $V$,
\begin{equation*}
s_{\aalpha}(x)=x-2\aalpha(x)\alpha/(\alpha, \alpha)\,,\quad x\in V\,.
\end{equation*}
Given a finite affine hyperplane arrangement in $V$ with prescribed multiplicities, we encode it in a finite set $\A_+=\{\aalpha\}\subset \VV$ and a collection of multiplicities $k_{\aalpha}\in\c$. The hyperplanes that pass through the origin $0\in V$ will be thus associated with vectors $\alpha\in V$. If the configuration is \emph{central} (with all hyperplanes passing through $0$), we are back to the previously considered case.  
As before, we extend the map $k:\,\A_+\to \c$ to $\A:=\A_+\sqcup (-\A_+)$ by putting $k_{-{\aalpha}}=k_{\aalpha}$.
With such a configuration of hyperplanes we associate a generalised Calogero--Moser operator 
 \begin{equation}\label{gcma}
L_{\A}=\Delta-u_{\A}\,,\qquad u_{\A}=\sum_{\aalpha\in \A_+} \frac{k_{\aalpha}(k_{\aalpha}+1)(\alpha,\alpha)}{(\aalpha(x))^2}\,.
\end{equation}
Definitions \ref{defloc}, \ref{gqi} require obvious modifications in the affine case.

\begin{defi}\label{defloca}
An affine configuration $\{\A, k\}$ is a \emph{locus configuration of type $W$} if
\begin{enumerate}
    \item[(1)] $R\subset \A$, and both $\A$ and $k:\,\A\to\c$ are $W$-invariant;
    \item[(2)] For any $\aalpha\in\A\setminus R$, $\,k_{\aalpha}\in\Z_+$, with $u_{\A}(x)-u_{\A}(s_{\aalpha} x)$ divisible by ${\aalpha}^{2k_{\aalpha}}$.
\end{enumerate}
\end{defi} 

\begin{defi}\label{gqia} For a locus configuration $\A$ of type $W$, $q\in\c[V]^W$ is \emph{quasi-invariant} if, for any $\aalpha\in\A\setminus R$,
%it is $W$-invariant and if
\begin{equation*}
q(x)-q(s_{\aalpha} x)\ \ \text{is divisible by}\ {\aalpha}^{2k_{\aalpha}}\,.
%\quad \forall\aalpha\in\A_+\setminus R\,.
\end{equation*}

\end{defi}
With these modifications, our results in Section \ref{S5} extend to the affine case in a straightforward manner. Below we discuss an analogue of Theorem \ref{gaic}: first in the general case, and then in dimension one where, as we explain, it is closely related to classical works \cite{AMM, AM, DG}. 

\subsection{General case}
\la{S9.2}
For an affine locus configuration $\A$ of type $W$, write $L=L_\A$, $L_0=L_W$. Note that in the affine case the ring $\qaw$ is no longer graded. Recall that we have the filtration \eqref{filtl} on $\c[V]^W$ associated with $L_0$; as explained in \ref{S5.1}, this filtration coincides with the standard filtration by degree. We write $\grd\,{\qaw}$ for the associated graded ring. We have the following analogue of Theorem \ref{gaic}.
\begin{theorem} 
\la{gaicaff}
\begin{enumerate}
\item[(1)] There exists a nonzero differential $($shift$)$ operator  
$ S$ such that $L  S=S L_0$. 

\item[(2)] For any quasi-invariant $q\in \qaw$ there exists a differential operator $L_q$ such that $L_q S=S L_{q,0}$ where $L_{q,0}=\Res(\e T_q\e)$. The operators $L_q$ pairwise commute and commute with $L$, and the map $q\mapsto L_q$ defines an algebra embedding 
$\theta\,:\ \grd\,{\qaw} \hookrightarrow \D(V\!\setminus\! H_{\A})^W$. 

\item[(3)] The algebras $\qaw$ and $\grd\,{\qaw}$ have Krull dimension $n=\dim V$; thus,  $L$ is completely integrable. 
\end{enumerate}
\end{theorem}
\begin{proof}
These results are proved by the same arguments as Theorem \ref{gaic}. For example, the shift operator $S$ is constructed in a similar way:
\begin{equation*}
    S=\frac{1}{2^NN!}\,\ad_{L, L_0}^N(\delta)\,,\quad\text{with}\quad \delta=\prod_{{\aalpha}\in\A_+\setminus R}\aalpha^{k_{\aalpha}}\,,\quad N=\deg\delta\,.
\end{equation*}
\end{proof}

\begin{remark}
The commutative ring $\theta(\grd\,\qaw)$, or even the larger ring obtained by adjoining $L$, is no longer maximal in the affine case. This can be already seen in dimension one, see Remark \ref{notmax}.
\end{remark}
 
 \subsection{One-dimensional case}
\la{S9.1}
In the case $V=\c$, where we have two options: $W=\{e\}$, $R=\varnothing$ or $W=\Z_2$, $R=\{\pm 1\}$. As we will see, this has a close relation with the works \cite{AMM, DG}. 
First, rescaling the elements of $\A$ if needed, we may assume that the affine-linear functions $\aalpha\in\A_+$ are of the form $x-x_i$. Hence, we may think of $\A_+$ as a finite collection of distinct points $x_i$, $i\in I$ and multiplicities $k_i\in\c$. 

\subsubsection{$W=\{e\}$} In this case $R=\varnothing$, so each $x_i$ comes with $k_i\in\Z_+$, leading to the operator
\begin{equation}\la{dg01}
L=\frac{d^2}{dx^2} - u(x)\,,\qquad u(x)=\sum_{i\in I}\frac{k_i(k_i+1)}{(x-x_i)^2}\,.
\end{equation}
The locus conditions require that $u(x)-u(s_ix)$ is divisible by $(x-x_i)^{2k_i}$, for every $i\in I$, where $s_i: x\mapsto 2x_i-x$. This is equivalent to the following relations:
\begin{equation}\la{dg02}
\sum_{j\in I\setminus\{i\}}\frac{k_j(k_j+1)}{(x_i-x_j)^{2s+1}}=0\quad\text{for $1\le s\le k_i$ and all $i\in I$.}
\end{equation}
In the case when $k_i=1$ for all $i$, these relations describe the so-called rational ``locus'' in \cite{AMM}. The more general relations \eqref{dg02} are due to Duistermaat and Gr\"unbaum \cite{DG}, who interpreted them as conditions for trivial local monodromy of $L$ near $x=x_i$ and showed the following.   

\begin{prop}[Theorem 3.4, \cite{DG}]\label{kdv}
For any operator $L$ of the form \eqref{dg01} with properties \eqref{dg02}, there exists a differential operator $D$ with rational coefficients, intertwining $L$ and $L_0=\frac{d^2}{dx^2}$, i.e. such that $LD=DL_0$.   \end{prop}

In fact, in \cite{DG} it is proved that $D$ can be obtained by iterating elementary rational Darboux transformations (of order one); this gives an effective method for constructing all such $L$ (see also \cite{AM}). %In the case when all $k_i=1$ this result can be found already in \cite{AM}. 
The corresponding $u(x)$ are called rational KdV potentials due to their link to rational solutions of the KdV equation. %We will refer to the family of such operators $L$ as {KdV family}.

 \begin{remark}
\la{notmax}  In the case when all $k_i=1$, it is known that the number $N$ of poles must be of the form $N=l(l+1)/2$ for some $l\in\Z_+$, and the maximal commutative ring containing $L$ is generated by $L$ and an operator $A$ of order $2l+1$ (see \cite{AMM}). In comparison, the commuting operators $L_q$ in Theorem \ref{gaicaff}(2) are obtained by $L_q=\frac{1}{2^rr!}\ad_{L}^{r}q$, $r=\deg q$. Now, for a quasi-invariant polynomial $q$, its derivative should vanish at $N$ points, hence $\deg q\ge N+1$ which is $>2l+1$ for $l>3$. Hence, the ring obtained by adjoining $L$ to $\theta(\grd\,{\qaw})$ is not maximal in that case.
 \end{remark}

\subsubsection{$W=\Z_2$} In this case $\A$ has to be invariant under $x\mapsto -x$, so each $\aalpha=x-x_i$ appears together with $-x-x_i=-(x+x_i)$, with the same multiplicity $k_i\in\Z_+$. Thus, we may interpret $\A$ as a finite subset of $\mathcal P$ in $\c\setminus\{0\}$, symmetric around $0$, with $k_p\in\Z_+$ satisfying $k_{-p}=k_p$. The corresponding operator is
\begin{equation}\la{dg1}
L=\frac{d^2}{dx^2} - u(x)\,,\qquad u(x)=\frac{k(k+1)}{x^2}+\sum_{p\in \mathcal P}\frac{k_p(k_p+1)}{(x-p)^2}\,,
\end{equation}
where $k$ is arbitrary. The locus conditions in this case mean that (cf. \cite[(4.45)--(4.46)]{DG})
\begin{equation}\la{dg2}
\frac{k(k+1)}{p^{2j+1}}+\sum_{q\in\mathcal P\setminus\{p\}}\frac{k_q(k_q+1)}{(p-q)^{2s+1}}=0\quad\text{for $1\le s\le k_p$ and all $p\in\mathcal P$.}
\end{equation}
The following result is due to Duistermaat and Gr\"unbaum. \cite{DG}.
\begin{prop}[cf. Proposition 4.3, \cite{DG}]\label{even}
For any operator $L$ of the form \eqref{dg1} with properties \eqref{dg2}, there exists a differential operator $D$ with rational coefficients, intertwining $L$ and $L_0=\frac{d^2}{dx^2}-\frac{k(k+1)}{x^2}$, i.e. such that $LD=DL_0$.   
\end{prop}
In fact, assuming $\mathcal P\ne \varnothing$, it follows from \cite{DG} that (1) $k$ must be a half-integer (see \cite[Eq. (4.44)]{DG}), and (2) $D$ can be found by iterating rational Darboux transformations of order one, starting from $L_0=\frac{d^2}{dx^2}+\frac{1}{4x^2}$. %This family of operators $L$ is called \emph{even family}.

%\begin{remark}
%For the ``even family'' \eqref{dg1}--\eqref{dg2} %in dimension one, the commutative ring %$\theta(\grd\,{\qaw})$ is not very interesting. %Indeed, all operators $L_q$ will be %$\Z_2$-symmetric, so they are polynomials in $L$.
%\end{remark}

\medskip

Comparing Propositions \ref{kdv} and \ref{even} with Theorem \ref{gaicaff}(1), we see that the Calogero--Moser operators $L_\A$ for locus configurations provide a \emph{multi-variable generalisation} of ``even'' family \eqref{dg1}--\eqref{dg2}, with the Coxeter group $W$ taking place of $W=\Z_2$, as well as the KdV family \eqref{dg01}--\eqref{dg02} (if $W=\{e\}$). 
 (It is interesting that in dimension $>1$ the multiplicities $k_\alpha$ for $\alpha\in R$ do not have to be half-integers.) As we have seen, there are plenty of examples of locus configurations with different groups $W$. Unfortunately, we know very few genuinely affine examples in dimension $>1$. Here is one two-dimensional example; it is of type $A_2$, and it can be viewed as a deformation of the root system of type $G_2$. It can be realised in $\R^3$ as $\widetilde G_2=A_2\cup \widetilde A_2$, where
\begin{equation*}
\la{newaff}
\begin{array}{lll}
A_2 &=\{\pm(e_i - e_j),\   1\le i<j\le 3\},  \text{with }\ k_\alpha=-1/3\,,\\ 
\widetilde A_2 &=\{\pm (3e_i -e_1- e_2-e_3 +c\delta), \ 1\le i \le 3\} \,,  \text{with }\
k_\alpha=1\,.
\end{array}
\end{equation*}
The corresponding Calogero--Moser operator is
\begin{multline*}\la{g2w}
L_{\widetilde G_2}=\Delta-\frac49\sum_{1\le i<j\le 3}\frac{1}{(x_i-x_j)^2}\\-\frac{12}{(2x_1-x_2-x_3+c)^2}-\frac{12}{(2x_2-x_1-x_3+c)^2}-\frac{12}{(2x_3-x_1-x_2+c)^2}\,.
\end{multline*}
 Here $c$ is arbitrary; for $c=0$ we have a $G_2$ configuration.   

\begin{remark} 
A trivial way of producing examples in dimension $>1$ is by taking direct sums of one-dimensional configurations. Another possibility is to use the methods %apply the procedures of \emph{isotropic reduction} and \emph{isotropic projectivization} as described in 
of \cite[Sec. 5.3]{CFV}. Such examples are reducible in a certain sense, so not so interesting.
\end{remark}

\begin{remark}
By analogy with the results of \cite{DG}, it is natural to expect that the Calogero--Moser operators for generalised locus configurations are \emph{bispectral}. In particular, we expect them to be {\it self-dual} when the configuration is central (cf. \cite[Theorem 2.3]{CFV}). Affine configurations, such as $\widetilde G_2$, should lead to examples of non-trivial bispectral duality.    
\end{remark}

\begin{remark}
In \cite{SV1}, the deformed Calogero--Moser operators were considered in their trigonometric form. Our methods cannot be applied verbatim to that case and require non-trivial modifications. We hope to return to this problem elsewhere. Some results about the trigonometric locus configurations can be found in \cite[Section 4]{C08}. Let us also mention a paper \cite{FVr}, where a trigonometric version of the above operator $L_{\widetilde G_2}$ is proposed.
\end{remark}


\begin{thebibliography}{50}

\bibitem[AM]{AM} Adler, M., Moser, J.: {\it On a class of polynomials connected with the Korteweg--de Vries equation.} Commun. Math. Phys. \textbf{61}, 1--30 (1978)

\bibitem[AMM]{AMM} Airault, H.,, McKean, H.~P., Moser, J.: {\it  Rational and elliptic solutions of the Korteweg--de Vries equation and a related many-body problem.} Commun. Pure Appl. Math. \textbf{30}, 95--148 (1977)

\bibitem[BGK]{BGK} Baranovsky, V., Ginzburg, V., Kuznetsov, A.: {\it Wilson's Grassmannian and a noncommutative quadric} IMRN \textbf{2003 (21)}, 1155--1197 (2003)

\bibitem[Bass]{Bass}
Bass, H.: \textit{Finitistic dimension and a homological generalization of semi-primary rings},
Trans. Amer. Math. Soc. \textbf{95}, 466--488 (1960).
\bibitem[B]{B98} Berest, Yu.: {\it Huygens' principle and the bispectral problem.}
{CRM Proceedings and Lecture Notes} \textbf{14}, 11--30 (1998)

\bibitem[BC]{BC} Berest, Yu., Chalykh, O.: {\it Quasi-invariants of complex reflection groups.} Compos. Math. \textbf{147}, no. 3, 965--1002 (2011)  

\bibitem[BCM]{BCM} Berest, Yu., Chalykh, O., Muller, G.: {\it Reflexive ideals and factorization in the rings of differential operators.} In preparation. 

\bibitem[BEG]{BEG} Berest, Yu., Etingof, P., Ginzburg, V.: {\it Cherednik algebras and
differential operators on quasi-invariants.} {Duke Math. J.}
\textbf{118}, no. 2, 279--337 (2003)
\bibitem[BK]{BK}
Berest, Yu.~Yu., Kasman, A.: 
\textit{$\D$-modules and Darboux transformations}, 
Letters in Mathematical Physics \textbf{43}, 279--294 (1998). 

\bibitem[BL]{BL} Berest, Yu.~Yu., Lutsenko I.~M.: {\it Huygens’ principle in Minkowski spaces and soliton solutions of the Korteweg--de Vries equation.} Commun. Math. Phys. \textbf{190}, 113--132 (1997)

\bibitem[BW]{BW} Berest Yu., Wilson G., \textit{Mad subalgebras of rings of differential operators on curves}, 
Adv. Math. \textbf{212}, no. 1, 163–-190 (2007) 

\bibitem[BrEtGa]{BrEtGa} Braverman, A., Etingof, P., Gaitsgory, D.: {\it Quantum integrable systems and differential Galois theory.} Transfor. Groups \textbf{2}, 31–57 (1997)

\bibitem[BCES]{BCES}Brookner, A., Corwin, D., Etingof, P.,  Sam, S.: {\it On Cohen--Macaulayness of $S_n$-invariant subspace arrangements.} IMRN \textbf{2016 (7)}, 2104--2106 (2016)

\bibitem[C1]{C98} Chalykh, O.: {\it Darboux transformations for multidimensional Schr\"odinger operators.} {Russian Math. Surveys} \textbf{53}, no. 2, 167--168 (1998)


\bibitem[C2]{C08} Chalykh, O.: {\it Algebro-geometric Schr\"odinger operators in many dimensions.} Phil. Trans. Royal Soc. A \textbf{366}, 947--971 (2008)

\bibitem[CEO]{CEO} Chalykh, O., Etingof, P., Oblomkov, A.: {\it Generalized Lam\'e operators.} Commun. Math. Phys. \textbf{239}, 115--153 (2003)

\bibitem[CFV1]{CFV1} Chalykh, O.~A., Feigin, M.~V., Veselov A.~P.: {\it New integrable generalizations of Calogero--Moser quantum problem.} J. Math. Phys. \textbf{39}, no. 2, 695--703 (1998) 

\bibitem[CFV2]{CFV} Chalykh, O. A., Feigin, M. V., Veselov, A. P.: {\it Multidimensional
Baker--Akhiezer functions and Huygens' principle.} {Commun.
Math. Phys.} \textbf{206}, 533--566 (1999)

\bibitem[CO]{CO} Chalykh, O.~A., Oblomkov, A.~A.:  {\it Harmonic oscillator and Darboux
transformations in many dimensions.} {Phys. Lett.} A \textbf{
267}, no. 4, 256--264 (2000)

\bibitem[CV1]{CV} Chalykh, O.~A., Veselov, A.~P.: {\it Commutative rings of partial
differential operators  and  Lie  algebras.} {Commun.  Math.
Phys.} \textbf{126}, 597--611 (1990)


\bibitem[CV2]{CV1} Chalykh, O.~A., Veselov, A.~ P.: {\it Locus configurations and
$\vee$-systems.} {Phys. Lett.} A, \textbf{285}, no. 5-6, 339--349 (2001)

\bibitem[Cr]{Cr} Crum, M.~M.: {\it Associated Sturm-Liouville systems.} Quart.  J. Math. \textbf{2}, no. 6, 21--126 (1955)

\bibitem[DG]{DG} Duistermaat, J.~J.,  Gr\"unbaum, F.~A.: {\it Differential equations in the spectral parameter.} Commun. Math. Phys. \textbf{103}, 177--240 (1986)

\bibitem[D]{D} Dunkl, C.~F.: {\it Differential-difference operators associated to reflection groups.} Trans. Amer. Math. Soc. \textbf{311}, no. 1, 167--183 (1989) 

\bibitem[EG1]{EG} Etingof, P., Ginzburg, V.: {\it Symplectic reflection algebras, Calogero--Moser space, and deformed Harish-Chandra homomorphism.} Invent. Math. \textbf{147}, 243--348 (2002)

\bibitem[EG2]{EG1} Etingof, P., Ginzburg, V.: {\it On $m$-quasi-invariants of a Coxeter
group.} {Mosc. Math. J.} \textbf{2}, no. 3, 555--566 (2002)

\bibitem[ER]{ER} Etingof, P., Rains, E. (with an appendix by M.~Feigin): {\it On Cohen--Macaulayness of algebras generated by generalised power sums.}  Commun. Math. Phys. \textbf{347}, 163--182 (2016)


\bibitem[F]{F} Feigin, M.: {\it Generalized Calogero-Moser systems from rational Cherednik algebras}. Selecta Math. \textbf{218}, no. 1, 253--281 (2012)

\bibitem[FJ]{FJ} Feigin, M., Johnston, D.: {\it A class of Baker--Akhiezer arrangements.} Commun. Math. Phys. \textbf{328}, no. 3, 1117--1157 (2014) 

\bibitem[FV1]{FV} Feigin, M.~V., Veselov, A.~P.: {\it Quasi-invariants of Coxeter groups and $m$-harmonic polynomials.} {IMRN} \textbf{2002 (10)},
2487--2511 (2002)

\bibitem[FV2]{FV1} Feigin, M.~V., Veselov, A.~P.: {\it Quasi-invariants and quantum
integrals of deformed Calogero--Moser systems.} {IMRN}
\textbf{2003 (46)}, 2487-2511 (2003)

\bibitem[FVr]{FVr} Feigin, M., Vrabec, M.: {\it Intertwining operator for $AG_2$ Calogero--Moser--Sutherland system.} J. Math. Phys. \textbf{60}, no.7, 073503 (2019)

\bibitem[GGMM]{GGMM} Garc\'ia-Ferrero, M.~A., G\'omez-Ullate, D., Milson, R., Munday, J.: {\it Exceptional Gegenbauer polynomials via isospectral deformations.} {\tt arXiv:2110.04059}.


\bibitem[GR]{GR}
Gaiotto, D., Rap$\check{\rm c}$\'ak, M., \textit{Miura operators, degenerate fields and the
M2-M5 intersection}, {\tt arXiv:2012.04118}.

\bibitem[H1]{He} Heckman, G.~J.: {\it A remark on Dunkl operators.} In: Harmonic Analysis on Reductive Groups, 181-193. Progress in Mathematics \textbf{101}, Birkhauser, 1991.


\bibitem[H2]{He2} Heckman, G.~J.: {\it An elementary approach to the hypergeometric shift operators of Opdam.} Invent. Math. \textbf{103}, 341--350 (1991) 

\bibitem[MPR]{MPR} Marquette, I., Post, S., Ritter, L.: {\it A family of fourth-order superintegrable systems with rational potentials related to Painlev\'e VI.} {\tt arXiv:2108.13533}.
 
\bibitem[MR]{MR} J. C. McConnell and J. C. Robson, \textit{Noncommutative Noetherian Rings},
Graduate Studies in Mathematics \textbf{30}, American Mathematical Society, Providence, RI, 2001.

\bibitem[N]{N}
N. A. Nekrasov, \textit{Seiberg-Witten prepotential from instanton counting}, Adv. Theor. Math. Phys. \textbf{7} (2003), no. 5, 831-–864.

\bibitem[NW]{NW}
Nekrasov, N., Witten, E., \textit{The Omega deformation, branes, integrability and Liouville theory}, J. High Energy Phys. \textbf{2010}, 92 (2010). 

\bibitem[OP]{OP} Olshanetsky, M.~A., Perelomov, A.~M.: {\it Quantum integrable systems related to Lie algebras}. Phys. Rep. \textbf{94}, no. 6, 313--404, (1983)

\bibitem[O]{HO4} Opdam, E.~M.: {\it Root systems and hypergeometric functions IV.} Compos. Math. \textbf{67}, no. 2, 191--209 (1988)


\bibitem[P]{P} Polychronakos, A.~P.: {\it  Exchange operator formalism for integrable systems of particles}. Phys. Rev. Lett. \textbf{69}, 703--705 (1992)
 
 \bibitem[PTV]{PTV} Post, S., Tsujimoto, S., Vinet, L.: {\it Families of superintegrable Hamiltonians constructed from exceptional polynomials.} J. Phys. A. Math.~Theor. \textbf{45} 405202 (2012)
 
 
\bibitem[SS]{SS} Smith, S.~P., Stafford, J.~T.: {\it Differential operators on an affine curve.} Proc. London Math. Soc. (3) \textbf{56}, 229--259 (1988) 

\bibitem[SV1]{SV} Sergeev, A.~N., Veselov, A.~P.: {\it Deformed quantum Calogero-Moser problems and Lie superalgebras.} Commun. Math. Phys. \textbf{245}, no. 2, 249--278 (2004)



\bibitem[SV2]{SV1} Sergeev, A.~N., Veselov, A.~P.: {\it Dunkl operators at infinity and Calogero--Moser systems.} IMRN \textbf{2015 (21)}, 10959--10986 (2015)


\bibitem[T1]{T} Taneguchi, K.: {\it On the symmetry of commuting differential operators with singularities along hyperplanes.} IMRN \textbf{2004 (36)}, 1845--1867 (2004)

\bibitem[T2]{T1} Taneguchi, K.: {\it Deformation of two body quantum Calogero--Moser--Sutherland
models.} Preprint (2006), available at: {\tt arXiv:math-ph/0607053}


\bibitem[VSC]{VSC} Veselov, A.~P., Styrkas, K.~L., Chalykh, O.~A.: {\it Algebraical integrability for Schr\"odinger equation and finite reflection groups.} Theor. Math. Phys. \textbf{94}, 253--275 (1993)




\end{thebibliography}
\end{document}